\documentclass[
	final,
	paper=a4,
	pagesize=auto,
	fontsize=10pt,
	version=last,
]{scrartcl}

\usepackage[utf8]{inputenc}
\usepackage[
	english,
]{babel}

\KOMAoptions{%
   DIV=10,
}%
\KOMAoptions{%
   numbers=noenddot,
}%
\KOMAoptions{
   twoside=semi,  
   twocolumn=false, 
   headinclude=false,%
   footinclude=false,%
   mpinclude=false,%
   headlines=2.1,%
   headsepline=true,%
   footsepline=false,%
   cleardoublepage=empty 
}%
\KOMAoptions{%
   parskip=relative, 
}%
\KOMAoptions{%
}%
\KOMAoptions{%
}%
\KOMAoptions{%
   bibliography=totoc 
}%
\KOMAoptions{%
}%
\KOMAoptions{%
}%
\KOMAoptions{%
}%

\makeatletter
\newcommand{\LoadPackagesNow}{}
\newcommand{\LoadPackageLater}[2][]{%
   \g@addto@macro{\LoadPackagesNow}{%
      \usepackage[#1]{#2}%
   }%
}
\makeatother

\usepackage{xspace}
\usepackage{xargs}
\usepackage{ifthen}
\usepackage{etoolbox}

\usepackage{calc}
\usepackage[
	a4paper,
	margin=3cm,
	]{geometry}
\usepackage{framed}
\usepackage{mdframed}
\usepackage{pbox}
\usepackage{enumitem}
\usepackage{array}
\usepackage{multirow}
\usepackage{longtable}
\usepackage{hhline}
\usepackage{afterpage}
\usepackage{pdflscape}
\usepackage{rotating}
\usepackage{booktabs}
\usepackage{setspace}
\usepackage[
   bottom,      
   stable,      
   multiple,    
   flushmargin,
   hang,
]{footmisc}

\usepackage{mathpazo}

\usepackage[%
headsepline,
automark,
]{scrlayer-scrpage}

\clearmainofpairofpagestyles
\clearplainofpairofpagestyles
\lehead{\pagemark}
\rohead{\pagemark}
\rehead{\headmark}
\lohead{\shortauthor}
\automark[section]{section} 

\usepackage{relsize}
\usepackage[normalem]{ulem}      
\usepackage{soul}		           
\usepackage{url}
\usepackage{lipsum}

\usepackage[table]{xcolor}
\usepackage{tikz}
\usetikzlibrary{arrows,shapes,calc,decorations.pathreplacing,intersections,quotes,angles,positioning,fit,petri,through,plotmarks}
\tikzstyle{blackdot}=[shape=circle,fill=black,minimum size=1mm,inner sep=0pt,outer sep=0pt]
\usepackage[%
]{graphicx}

\usepackage{epstopdf}
\usepackage{wrapfig}
\usepackage{caption}
\usepackage{subcaption}
\usepackage[
   tbtags,    
   sumlimits,  
   nointlimits, 
   namelimits, 
   reqno,     
]{amsmath} %

\usepackage{amsfonts}
\usepackage{mathrsfs} 
\usepackage{dsfont}
\usepackage{amssymb}
\usepackage{units}
\usepackage{stmaryrd}
\LoadPackageLater{amsthm}
\LoadPackageLater{thmtools}
\usepackage[fixamsmath,disallowspaces]{mathtools}
\mathtoolsset{showonlyrefs}
\mathtoolsset{centercolon=true}
\usepackage{bm} 
\makeatletter
\g@addto@macro\bfseries{\boldmath}
\makeatother
\allowdisplaybreaks[4]
\numberwithin{equation}{section}
\usepackage{arydshln}
\usepackage{diagbox}
\usepackage{makecell}

\usepackage[toc]{appendix}

\usepackage{csquotes}
\usepackage[
	bibstyle=alphabetic,
	citestyle=alphabetic,
	sorting=ynt,
	sortcites=true,
	giveninits=true,
	maxbibnames=99,
	backend=biber,
	maxalphanames=5,
]{biblatex}

\renewbibmacro{in:}{}
\DeclareFieldFormat{pages}{#1}

\usepackage[textsize=tiny,english,colorinlistoftodos,
	disable,
	]{todonotes}

\definecolor{pdfurlcolor}{rgb}{0,0,0.6}
\definecolor{pdffilecolor}{rgb}{0.7,0,0}
\definecolor{pdflinkcolor}{rgb}{0,0,0.6}
\definecolor{pdfcitecolor}{rgb}{0,0,0.6}
\usepackage[
  colorlinks=true,         
  urlcolor=pdfurlcolor,    
  filecolor=pdffilecolor,  
  linkcolor=pdflinkcolor,  
  citecolor=pdfcitecolor,  %
  raiselinks=true,			 
  breaklinks,              
  verbose,
  hyperindex=true,         
  linktocpage=true,        
  hyperfootnotes=false,     
  bookmarks=true,          
  bookmarksopenlevel=2,    
  bookmarksopen=true,      
  bookmarksnumbered=true,  
  bookmarkstype=toc,       
  plainpages=false,        
  pageanchor=true,         
  pdfdisplaydoctitle=true, 
  pdfstartview=FitH,       
  pdfpagemode=UseOutlines, 
  pdfpagelabels=true,           
  pdfpagelayout=OneColumn, 
]{hyperref}

\LoadPackagesNow

\newcommand{\ifargdef}[3][{}]{\ifthenelse{\equal{#2}{}}{#1}{#3}}

\newlength{\hangind}
\newcommand{\myhangindent}[1]{\settowidth{\hangind}{\widthof{#1}}\hangindent=\the\hangind}

\setkomafont{section}{\Large\rmfamily\bfseries}
\setkomafont{subsection}{\large\rmfamily\bfseries}
\setkomafont{subsubsection}{\rmfamily\bfseries}
\setkomafont{paragraph}{\rmfamily\bfseries}
\setkomafont{pageheadfoot}{\smaller\scshape\rmfamily}

\usepackage{tabularx}
\hypersetup{
	pdfauthor={Martin Genzel, Maximilian März, and Robert Seidel},
	pdftitle={Compressed Sensing with 1D Total Variation: Breaking Sample Complexity Barriers via Non-Uniform Recovery},
	pdfsubject={Article},
	pdfcreator={PDF-LaTeX},
}

\newenvironment{highlight}{\begin{quote}\itshape}{\end{quote}}

\newenvironment{rmklist}
{\begin{enumerate}[label={(\arabic*)},itemindent=2em,leftmargin=0em]}
{\end{enumerate}}

\newenvironment{thmlist}
{\begin{enumerate}[label={(\alph*)}]}
{\end{enumerate}}
\newenvironment{prooflist}
{\begin{enumerate}[label={(\alph*)},itemindent=2em,leftmargin=0em]}
	{\end{enumerate}}
\newenvironment{thmproperties}
{\begin{enumerate}[label={(\roman*)}]}
{\end{enumerate}}

\newenvironment{deflist}
{\begin{enumerate}[label={(\arabic*)}]}
{\end{enumerate}}

\renewcommand{\qedsymbol}{$_\blacksquare$}

\providecommand{\qedhere}{\hfill\qedsymbol}

\newtheoremstyle{claim}
	{\topsep}{\topsep}%
	{\itshape}
	{}
	{}
	{}
	{.5em}
	{{\bfseries\boldmath\thmname{#1} \thmnumber{#2}} \thmnote{(#3)}}

\newtheoremstyle{definition}
	{\topsep}{\topsep}%
	{}
	{}
	{}
	{}
	{.5em}
	{\textbf{\thmname{#1} \thmnumber{#2}} \thmnote{(#3)}}
	
\newtheoremstyle{algorithm}
	{\topsep}{\topsep}%
	{}
	{}
	{\bfseries\boldmath}
	{}
	{\newline}
	{\thmname{#1} \thmnumber{#2} \thmnote{(#3)}}

\mdfdefinestyle{boxed}{innertopmargin =6pt, splittopskip = \topskip, skipbelow=6pt, skipabove=12pt}

\mdfdefinestyle{emphframe}{linecolor=black, innertopmargin = 3pt, splittopskip = \topskip, skipbelow=6pt, skipabove=6pt}

\declaretheorem[style=claim,numberwithin=section]{theorem}
\declaretheorem[style=claim,sibling=theorem]{lemma}
\declaretheorem[style=claim,sibling=theorem]{corollary}
\declaretheorem[style=claim,sibling=theorem]{proposition}
\declaretheorem[style=definition,sibling=theorem]{definition}

\declaretheorem[style=definition,sibling=theorem,qed=$\Diamond$]{remark}
\declaretheorem[style=definition,sibling=theorem,qed=$\Diamond$]{example}
\declaretheorem[style=definition,sibling=theorem,qed=$\Diamond$]{iteration}
\declaretheorem[style=algorithm,sibling=theorem,%
	preheadhook={\begin{mdframed}[style=emphframe] \setcounter{mpfootnote}{\value{footnote}}},%
	postfoothook=\end{mdframed}]{experiment}
\declaretheorem[style=algorithm,sibling=theorem,%
	preheadhook={\begin{mdframed}[style=emphframe] \setcounter{mpfootnote}{\value{footnote}}},%
	postfoothook=\end{mdframed}]{algorithm}

\newcommand{\opleft}[1]{\mathopen{}\left#1}
\newcommand{\opright}[1]{\right#1\mathclose{}}
\newcommandx{\braces}[4]{%
\ifstrequal{#3}{normal}{#1#4#2}{%
\ifstrequal{#3}{auto}{\left#1#4\right#2}{%
\ifstrequal{#3}{opauto}{\opleft#1#4\opright#2}{%
#3#1#4#3#2}}}%
}
\newcommandx{\opannot}[3][3=\downarrow]{\stackrel{\mathclap{\substack{#1 \\ #3 \vspace{2pt}}}}{#2}}
\newcommandx{\lineannot}[3][3=\rightarrow]{\mathllap{\boxed{\text{\textsmaller{#1}}} #3} #2}
\newcommandx{\multilineannot}[4][4=\rightarrow]{\mathllap{\boxed{\parbox{#1}{\RaggedRight\textsmaller{#2}}} #4} #3}
 
\newcommand{\N}{\mathbb{N}} 
\newcommand{\Q}{\mathbb{Q}} 
\newcommand{\R}{\mathbb{R}} 

\newcommand{\suchthat}[1][normal]{\ifstrequal{#1}{normal}{\mid}{#1|}} 

\newcommand{\setcompl}[1]{#1^c} 
\newcommand{\cardinality}[1]{\abs{#1}} 
\newcommand{\union}{\cup} 
\newcommand{\disjunion}{\mathrel{\dot{\union}}} 
\newcommand{\intersec}{\cap} 

\newcommandx{\intvcl}[3][1=normal]{\braces{[}{]}{#1}{#2, #3}} 
\newcommandx{\intvop}[3][1=normal]{\braces{(}{)}{#1}{#2, #3}} 
\newcommandx{\intvclop}[3][1=normal]{\braces{[}{)}{#1}{#2, #3}} 
\newcommandx{\intvopcl}[3][1=normal]{\braces{(}{]}{#1}{#2, #3}} 
\DeclareMathOperator*{\argmin}{argmin} 
\DeclareMathOperator*{\argmax}{argmax} 
\DeclareMathOperator{\sign}{sign}
\newcommandx{\abs}[2][1=normal]{\braces{\lvert}{\rvert}{#1}{#2}} 
\newcommandx{\ceil}[2][1=normal]{\braces{\lceil}{\rceil}{#1}{#2}} 
\newcommandx{\floor}[2][1=normal]{\braces{\lfloor}{\rfloor}{#1}{#2}} 
\newcommandx{\round}[2][1=normal]{\braces{\llbracket}{\rrbracket}{#1}{#2}} 
\newcommandx{\der}[1]{D^{#1}} 
\newcommandx{\gradient}{\nabla} 
\newcommandx{\partder}[4][1={},4={}]{\frac{\partial^{#4} #2}{\partial #3^{#4}}\ifargdef{#1}{\Big|_{#1}}} 
\newcommandx{\integ}[4][1={},2={}]{\int_{#1}^{#2} #3 \, #4} 
\newcommandx{\asympffaster}[2][1=normal]{o\braces{(}{)}{#1}{#2}} 
\newcommandx{\asympfaster}[2][1=normal]{O\braces{(}{)}{#1}{#2}} 
\newcommandx{\asympeq}[2][1=normal]{\Theta\braces{(}{)}{#1}{#2}} 
\newcommandx{\asympsslower}[2][1=normal]{\omega\braces{(}{)}{#1}{#2}} 
\newcommandx{\asympslower}[2][1=normal]{\Omega\braces{(}{)}{#1}{#2}} 

\newcommand{\matr}[1]{\begin{bmatrix} #1 \end{bmatrix}} 
\newcommandx{\norm}[2][1=normal]{\braces{\|}{\|}{#1}{#2}} 
\renewcommandx{\sp}[3][1=normal]{\braces{\langle}{\rangle}{#1}{#2, #3}} 
\newcommandx{\End}[2][2={}]{\mathcal{L}\opleft( #1 \ifargdef{#2}{, #2} \opright)} 
\newcommand{\T}{\mathsf{T}} 
\newcommand{\psinv}[1]{#1^{\dagger}} 
\renewcommand{\vec}[1]{\boldsymbol{#1}} 

\newcommandx{\measure}[2][1=normal]{\operatorname{vol}\braces{(}{)}{#1}{#2}} 
\newcommand{\indset}[1]{\chi_{#1}} 
\DeclareMathOperator{\supp}{supp} 
\newcommandx{\Leb}[3][1={},3=normal]{L^{#2}\ifargdef{#1}{\braces{(}{)}{#3}{#1}}{}} 
\newcommandx{\Lebnorm}[4][1=normal,3={2},4={}]{\norm[#1]{#2}_{#3}} 
\renewcommandx{\l}[3][1={},3=normal]{\ell^{#2}\ifargdef{#1}{\braces{(}{)}{#3}{#1}}} 
\newcommandx{\lnorm}[4][1=normal,3={2},4={}]{\norm[#1]{#2}_{#3}} 
\newcommandx{\Smooth}[4][1={},3={},4=normal]{C_{#3}^{#2}\ifargdef{#1}{\braces{(}{)}{#4}{#1}}} 
\newcommandx{\Schwartz}[2][1={},2=normal]{\mathscr{S}\ifargdef{#1}{\braces{(}{)}{#2}{#1}}} 
\newcommandx{\Schwartzpoly}[2][1=normal]{\braces{\langle}{\rangle}{#1}{\abs[#1]{#2}} } 
\newcommandx{\Tempdistr}[2][1={},2=normal]{\mathscr{S}'\ifargdef{#1}{\braces{(}{)}{#2}{#1}}} 
\newcommandx{\distrinp}[3][1=normal]{\braces{\langle}{\rangle}{#1}{#2, #3}} 
\newcommandx{\ft}[3][1=default,2=auto]{
\ifstrequal{#1}{default}{\widehat{#3}}{
\ifstrequal{#1}{long}{{\braces{(}{)}{#2}{#3}}^{\wedge}}{}}} 
\newcommandx{\ift}[3][1=default,2=auto]{
\ifstrequal{#1}{default}{\check{#3}}{
\ifstrequal{#1}{long}{{\braces{(}{)}{#2}{#3}}^{\vee}}{}}} 

\DeclareMathOperator{\PolyLog}{PolyLog} 
\newcommand{\y}{\vec{y}} 
\newcommand{\A}{\vec{A}} 

\newcommand{\noise}{\vec{e}} 
\newcommand{\noiseparam}{\eta} 

\newcommand{\x}{\vec{x}} 
\newcommand{\grtr}{\vec{x}^\ast} 
\newcommand{\grtrsparse}{\bar{\vec{x}}^\ast} 
\newcommandx{\solu}[1][1={}]{\ifargdef[\hat{\vec{x}}]{#1}{\hat{#1}}} 
\renewcommand{\v}{\vec{v}} 
\newcommand{\vnull}{\vec{0}} 
\newcommand{\1}{\mathbf1} 
\newcommand{\restrict}{\big|} 

\newcommand{\contgrtr}{\mathcal{X}} 
\newcommand{\cont}[1]{#1^{\circ}}

\newcommand{\sset}{K} 

\newcommand{\aop}{\vec{\Psi}} 
\newcommand{\ssupp}[1][{}]{\mathcal{S}_{#1}} 
\newcommand{\ssuppc}[1][{}]{\setcompl{\mathcal{S}}_{#1}} 

\newcommand{\proj}[1]{\vec{P}_{#1}} 
\newcommand{\I}[1]{\vec{I}_{#1}} 

\newcommand{\probsuccess}{u} 

\newcommandx{\prob}[2][1={},2=normal]{\mathbb{P}\ifargdef{#1}{\braces{[}{]}{#2}{#1}}}

\newcommandx{\mean}[2][1={},2=normal]{\mathbb{E}\ifargdef{#1}{\braces{[}{]}{#2}{#1}}}
\newcommandx{\var}[2][1={},2=normal]{\mathbb{V}\ifargdef{#1}{\braces{[}{]}{#2}{#1}}}

\newcommand{\distributed}{\sim}
\newcommandx{\Normdistr}[3][1=normal]{\mathcal{N}\braces{(}{)}{#1}{#2, #3}} 
\newcommandx{\normsubg}[2][1=normal]{\norm[#1]{#2}_{\psi_2}} 
\newcommand{\gaussian}{\vec{g}} 

\newcommand{\median}{\operatorname{Median}}

\newcommandx{\anorm}[3][1=normal,3={\sset}]{\norm[#1]{#2}_{#3}} 
\newcommandx{\pospart}[2][1=auto]{\braces{[}{]}{#1}{#2}_+}

\newcommandx{\ball}[2][1={},2={}]{B_{#1}^{#2}} 
\renewcommand{\S}{\mathbb{S}} 

\newcommand{\meanwidth}[2][{}]{w_{#1}(#2)} 
\newcommand{\effdim}[2][{}]{w_{#1}^2(#2)} 
\newcommand{\conic}{\wedge} 
\newcommand{\cone}[1]{\operatorname{cone}(#1)} 
\newcommand{\descset}[1]{\mathcal{D}(#1)} 
\newcommand{\desccone}[1]{\mathcal{D}_{\conic}(#1)} 

\newcommand{\subd}[1]{\partial #1} 
\newcommandx{\subdiff}[2][2={}]{\partial #1\ifargdef{#2}{(#2)}{}} 
\newcommand{\dv}{\vec{w}} 
\newcommandx{\clip}[3][1=normal]{\operatorname{clip}\braces{(}{)}{#1}{#2; #3}} 

\newcommand{\F}{\mathcal{F}} 
\newcommand{\Fext}{\mathcal{\bar{F}}} 
\renewcommand{\H}{\vec{H}} 

\let\b\relax
\let\d\relax

\newcommand{\gssignals}[1]{\mathcal{G}_{#1}}

\newcommand{\s}{s}
\newcommand{\sext}{\bar{s}}
\newcommand{\SC}{\Delta}

\newcommand{\jump}{\nu}
\newcommand{\jumpext}{\xi}

\newcommand{\Klambda}{\mathcal{K}} 

\newcommand{\Supp}{\mathcal{\bar{S}}} 
\newcommand{\SuppEq}{\mathcal{\tilde{S}}} 
\newcommand{\SuppS}[2]{\Supp_{#1}^{(#2)}} 

\newcommand{\TV}{\nabla} 

\newcommand{\mo}{{-1}}

\newcommand{\b}[2]{\beta_{#1}^{(#2)}}
\newcommand{\d}[2]{d_{#1}^{(#2)}}
\newcommand{\dl}[2]{\overleftarrow{d}_{#1}^{(#2)}}
\newcommand{\dr}[2]{\overrightarrow{d}_{#1}^{(#2)}}
\newcommand{\e}[2]{e_{#1}^{(#2)}}
\newcommand{\h}[2]{\vec{h}_{#1}^{(#2)}}
\newcommand{\n}[2]{n_{#1}^{(#2)}}
\newcommand{\nl}[2]{\overleftarrow{n}_{#1}^{(#2)}}
\newcommand{\nr}[2]{\overrightarrow{n}_{#1}^{(#2)}}

\newcommand{\p}[2]{p_{#1}^{(#2)}}
\newcommand{\pl}[2]{\overleftarrow{p}_{#1}^{(#2)}}
\newcommand{\pr}[2]{\overrightarrow{p}_{#1}^{(#2)}}

\newcommand{\child}{\mathfrak{c}}
\DeclareMathOperator{\depth}{d}

\newcommand{\lchild}{\overleftarrow{\child}}
\newcommand{\rchild}{\overrightarrow{\child}}

\newcommand{\parent}{\mathfrak{p}}
\renewcommand{\Q}[2]{\mathcal{Q}_{#1}^{(#2)}}

\newcommand{\Ql}[2]{\overleftarrow{\mathcal{Q}}_{#1}^{(#2)}}
\newcommand{\Qr}[2]{\overrightarrow{\mathcal{Q}}_{#1}^{(#2)}}
\newcommand{\olarr}[1]{\overleftarrow{#1}}
\newcommand{\orarr}[1]{\overrightarrow{#1}}

\newcommand{\revision}[1]{#1}

\graphicspath{{images/}}

\begin{document}

\renewcommand*{\thefootnote}{\fnsymbol{footnote}}
\pagestyle{scrheadings}

\begin{center}
	\bfseries\larger[3]{Compressed Sensing with 1D Total Variation: \\ Breaking Sample Complexity Barriers \\ via Non-Uniform Recovery}
\end{center}

\vspace{1\baselineskip}
\begin{addmargin}[2em]{2em}
\begin{center}
\noindent{\normalsize\bfseries{Martin Genzel\footnote{Technische Universit\"at Berlin, Department of Mathematics, Berlin, Germany; E-Mail:~\href{mailto:genzel@math.tu-berlin.de}{\texttt{genzel@math.tu-berlin.de}}} \qquad Maximilian März\footnote{Technische Universit\"at Berlin, Department of Mathematics, Berlin, Germany; E-Mail:~\href{mailto:maerz@math.tu-berlin.de}{\texttt{maerz@math.tu-berlin.de}}} \qquad Robert Seidel\footnote{Technische Universit\"at Berlin, Institute of Software Engineering and Theoretical Computer Science, Berlin, Germany; E-Mail:~\href{mailto:robert.seidel@campus.tu-berlin.de}{\texttt{robert.seidel@campus.tu-berlin.de}}}}}
\end{center}

\vspace{1\baselineskip}
{\smaller
\noindent\textbf{Abstract.}
This paper investigates total variation minimization in one spatial dimension for the recovery of gradient-sparse signals from undersampled Gaussian measurements. Recently established bounds for the required sampling rate state that uniform recovery of all $s$-gradient-sparse signals in $\R^n$ is only possible with $m \gtrsim \sqrt{\s n} \cdot \PolyLog(n)$ measurements. 
Such a condition is especially prohibitive for high-dimensional problems, where $s$ is much smaller than $n$. However, previous empirical findings seem to indicate that this sampling rate does not reflect the typical behavior of total variation minimization.
The present work provides a rigorous analysis that breaks the $\sqrt{\s n}$-bottleneck for a large class of ``natural'' signals.
The main result shows that non-uniform recovery succeeds with high probability for $m \gtrsim \s \cdot \PolyLog(n)$ measurements if the jump discontinuities of the signal vector are sufficiently well separated.
In particular, this guarantee allows for signals arising from a discretization of piecewise constant functions defined on an interval.
The key ingredient of the proof is a novel upper bound for the associated conic Gaussian mean width, which is based on a signal-dependent, non-dyadic Haar wavelet transform. Furthermore, a natural extension to stable and robust recovery is addressed.

\vspace{.5\baselineskip}

\vspace{.5\baselineskip}
\noindent\textbf{MSC classes.} 42C40, 68Q25, 94A12, 94A20.

}
\end{addmargin}
\newcommand{\shortauthor}{Genzel, März, and Seidel: Compressed Sensing with 1D Total Variation}
%
\renewcommand*{\thefootnote}{\arabic{footnote}}
\setcounter{footnote}{0}


\thispagestyle{plain}

\section{Introduction}
\label{sec:intro}

The methodology of \emph{compressed sensing} allows for the reconstruction of an unknown signal from surprisingly few indirect and randomized measurements by exploiting the inherent structure of the signal. This field was pioneered by the seminal works of Cand\`{e}s, Donoho, Romberg, and Tao~\cite{candes2006cs,candes2006stable,donoho2006cs}, which have triggered a large amount of research in the past two decades; see also \cite{foucart2013cs} for more details on this subject. 

A standard setup in compressed sensing concerns the following inverse problem:  Assume that $\grtr \in \R^n$ denotes a signal vector of interest that is \emph{$\s$-sparse}, i.e., $\cardinality{\supp(\grtr)} \leq \s$. But instead of having direct access to $\grtr$, it is observed via a \emph{linear, non-adaptive measurement process}\footnote{For the sake of simplicity, potential distortions in the measurement process are ignored here, but we emphasize that all results of this work can be made robust against (adversarial) noise; see Subsection~\ref{subsec:results:stable} for more details.} 
\begin{equation}\label{eq:meas}
	\y = \A \grtr,
\end{equation}
where $\A \in \R^{m \times n}$ is a known matrix, typically referred to as the \emph{measurement matrix}. A remarkable insight of compressed sensing is that, under certain circumstances, it is still possible to retrieve $\grtr$ from the knowledge of $\y \in \R^m$ even when $m\ll n$. Indeed, if the measurement matrix $\A$ is drawn according to an appropriate random distribution, the recovery of $\grtr$ remains feasible with high probability as long as the number of measurements obeys $m \gtrsim \s \log(2n/\s)$, where the `$\gtrsim$'-notation hides a universal constant. For the success of this strategy, it is crucial to employ non-linear recovery methods that exploit the a priori knowledge that $\grtr$ is sparse. There exist numerous greedy methods and convex programs that are designed to accomplish this task efficiently. Arguably, the most popular variant is the so-called \emph{basis pursuit}~\cite{chen_atomic_1998}, which solves the convex problem
\begin{equation}\label{eq:intro:bp}\tag{$\text{BP}$}
	\min_{\x \in \R^n} \lnorm{\x}[1] \quad \text{subject to \quad $\A \x = \y$.}
\end{equation}
The key component of this minimization task is using the $\l{1}$-norm as objective function, which is known to promote sparsity in the solutions. 

Historically, the idea of promoting sparsity via the $\l{1}$-norm dates back even further than the initial works on compressed sensing~\cite{candes2006cs,candes2006stable,donoho2006cs} or on finding sparse representations with the basis pursuit~\cite{chen_atomic_1998}. For instance, similar techniques can be found in the work of \citeauthor{Beurling1938}~\cite{Beurling1938} or in the influential paper of \citeauthor{rudin_nonlinear_1992} (ROF)~\cite{rudin_nonlinear_1992}, which are both formulated in terms of \emph{total variation minimization}. In fact, the above assumption that the unknown signal $\grtr$ is sparse by itself is typically not satisfied in real-world applications. However, it has turned out that encouraging a small total variation norm often efficaciously reflects the inherent structure of the signal---an observation that holds particularly true for image signals. While the original ROF-model was formulated for image denoising, methods based on total variation regularization became state of the art for many other image reconstruction problems, for instance, see~\cite{Chambolle1997,Chambolle2004} or~\cite[Tab.~2.1]{benning_burger_2018}. Although probably not as popular as their counterparts in higher spatial dimensions, total variation methods in one spatial dimension find application in many practical scenarios as well, e.g., see \cite{Little2011,Little2010,Sandbichler2015,Wu2014,Perrone2016}. \revision{Furthermore, total variation in 1D has frequently been subject of mathematical research \cite{Condat2013,Selesnick2015,Selesnick2012,Mammen1997,Briani2011,Grasmair2007,gl19}}. 

In the discrete one-dimensional setting of this article, such a \emph{total variation (TV) model} translates into the assumption of \emph{gradient-sparsity}, i.e., we have that $\cardinality{\supp(\TV\grtr)} \leq \s$, where 
\begin{equation}
\TV  \coloneqq \matr{
	-1 & 1 & 0 & \dots & 0 \\
	0 & -1 & 1 & 0 & 0\\
	\vdots & & \ddots &\ddots & \vdots \\
	0 & \dots & 0 & -1 & 1} \in \R^{N \times n}
\end{equation}
denotes a \emph{discrete gradient operator}\footnote{Note that this specific difference operator is based on forward differences and assumes \emph{von Neumann boundary conditions}. This seems to be a particularly popular choice in the context of compressed sensing, which we do also follow in this work. Nevertheless, we expect that it is straightforward to extend our results to other common variants of gradient operators.} with $N \coloneqq n-1$. In other words, the signal $\grtr \in \R^n$ is assumed to be piecewise constant with at most $\s$ \emph{jump discontinuities}. In order to reconstruct $\grtr$ from a given measurement vector $\y = \A \grtr \in \R^m$, we consider the following modified version of the basis pursuit~\eqref{eq:intro:bp}, which we refer to as \emph{TV minimization (in 1D)}:
\begin{equation}\label{eq:intro:tv-1}\tag{$\text{TV-1}$}
	\min_{\x \in \R^n} \lnorm{\TV\x}[1] \quad \text{subject to \quad $\A \x = \y$.}
\end{equation}
The main objective of the present work is to study this minimization problem from a compressed sensing perspective by analyzing its \emph{sample complexity} for Gaussian measurements. In a nutshell, we intend to answer the following question:
\begin{highlight}
Assuming that $\A \in \R^{m\times n}$ is a standard Gaussian random matrix, under which conditions is it possible to recover an $\s$-gradient-sparse signal $\grtr \in \R^n$ via TV minimization~\eqref{eq:intro:tv-1} with the near-optimal rate of $m \gtrsim \s\cdot \PolyLog(n)$ measurements?
\end{highlight}

\subsection{Prior Art}
\label{subsec:intro:priorart}

TV minimization forms an active branch of research in compressed sensing, and we refer the interested reader to~\cite{krahmer_total_2017} for a comprehensive overview. 
The following discussion is confined to a brief historical outline of the subject as well as several results in the literature that are of particular relevance to the problem setup of this work; see also Table~\ref{tab:intro} for a compact summary.

\begin{table}
	\renewcommand{\arraystretch}{1.5}
	\begin{center}
		\begin{tabular}{|c||c|c|}
			\hline
			\diagbox[width=2.5cm,height=1cm]{$\A$}{$d$D} &  1D  & $\geq$2D \\
			\hline \hline
			\multirow{2}{*}{Gaussian} & \textbf{$\s \log^2 (n)$} (non-unif.) {\footnotesize[ours]} & $\s \cdot \PolyLog(n,\s)$ \\
			\cdashline{2-2}
			& $\sqrt{\s n} \cdot \log (n)$ (unif.) {\footnotesize \cite{cai_guarantees_2015}} & \footnotesize\cite{cai_guarantees_2015,Needell2013,Needell2013b} \\
			\hline
			\multirow{2}{*}{Fourier} & \multicolumn{2}{c|}{$\s \cdot \PolyLog(n,\s)$}  \\
			& \multicolumn{2}{c|}{\footnotesize \cite{candes2006cs,Poon2015,Krahmer2014}} \\
			\hline
		\end{tabular}
	\end{center}
	\caption{An overview of known asymptotic-order sampling rates for TV minimization in compressed sensing, ignoring universal and model-dependent constants.}
	\label{tab:intro}
\end{table}

Already one of the seminal works on compressed sensing by \citeauthor{candes2006cs} \cite{candes2006cs} considers the recovery of $\s$-gradient-sparse signals via \eqref{eq:intro:tv-1} from randomly subsampled Fourier measurements. They show that with high probability, $m \gtrsim \s \log (n)$ of such measurements are sufficient for exact recovery of $\grtr$---a guarantee that naturally extends to more spatial dimensions. However, we emphasize that this result is specifically tailored to Fourier measurements. Although extendable to stable recovery~\cite{Patel2012}, it does not allow for robustness against noise in the measurements, which is a crucial feature for applications. This issue was subsequently addressed by \citeauthor{Needell2013} in~\cite{Needell2013,Needell2013b}. These works establish robust and stable recovery for measurement matrices that satisfy a restricted isometry property (RIP) when composed with the orthonormal Haar wavelet transform; for instance, such a condition is fulfilled with $m\gtrsim \s\cdot \PolyLog (n,\s)$ measurements, if $\A$ is a (sub-)Gaussian matrix or a partial Fourier matrix with randomized column signs. This achievement is based on a connection between compressibility of Haar wavelet representations and bounded variation of a function. Unfortunately, such a property only holds true for two and more spatial dimensions, so that the results of~\cite{Needell2013,Needell2013b} cannot be extended to the one-dimensional case. 

To the best of our knowledge, these are the first guarantees in the literature indicating that sparse recovery by TV minimization in 1D might behave differently than in higher dimensions. We highlight already at this point that a certain connection between gradient-sparse signals and the Haar wavelet transform will also play a crucial role in our approach, yet it appears in a very different manner. 

\citeauthor{Krahmer2014}~\cite{Krahmer2014} and \citeauthor{Poon2015}~\cite{Poon2015} also cover TV-based recovery from subsampled Fourier measurements by relying on the Haar wavelet transform. Both works employ so-called variable-density sampling of the Fourier transform, where the sampling in the low frequencies is denser than in the high frequencies. Such a strategy enables the authors of~\cite{Krahmer2014} to apply the results of~\cite{Needell2013,Needell2013b} and thereby to show stable and robust recovery guarantees in 2D, which are based on the RIP. 
The work of \citeauthor{Poon2015}~\cite{Poon2015} contains similar results, which are non-uniform and allow for stable and robust recovery from $m\gtrsim \s \log (n)$ Fourier measurements in 1D. \citeauthor{Poon2015} also provides comparable theorems in 1D and 2D for a uniform sampling pattern of the (discrete) Fourier domain.
Of more interest to our approach is another result in~\cite{Poon2015} (see Thm.~2.6 therein) concerning the recovery of gradient-sparse signals in 1D that obey additional structural constraints: it shows that  $m \gtrsim \s~\cdot~\PolyLog(n,\s)$ random Fourier measurements up to a low-frequency threshold are sufficient for successful recovery of~$\grtr$; importantly, this frequency threshold is determined by the minimal distance of the jump discontinuities of $\grtr$.

The recovery method of~\eqref{eq:intro:tv-1} is formulated as minimization of an \emph{analysis-based prior}. This strategy builds upon the assumption that the signal of interest $\grtr$ is of low complexity after being transformed by a so-called \emph{analysis operator}, which is $\TV \in \R^{N \times n}$ in our case. In recent years, this conceptual idea has gained popularity and is sometimes also referred to as the \emph{cosparse analysis model}~\cite{nam_cosparse_2013}. Note that in contrast, a \emph{synthesis-based prior} assumes that $\grtr$ possesses a low-complexity representation in a \emph{dictionary} $\vec{D} \in \R^{n \times  d}$~\cite{Elad2006}, e.g., that $\grtr = \vec{D} \vec{z}$ with a sparse coefficient vector $\vec{z} \in \R^d$. Many works on analysis-based priors operate under the hypothesis that the underlying analysis operator forms a frame \cite{candes2011csdict,kabanava_analysis_2015,liu2012,foucart2014,Haltmeier2013,Alberti2019}, making related theoretical guarantees inapplicable to TV minimization. 
Other approaches are concerned with deriving general and abstract recovery conditions~\cite{nam_cosparse_2013,vaiter2013}, but these do not yield bounds for the number of measurements required for successful recovery.  
Finally, there has been an effort to describe the sampling rate of the $\l{1}$-analysis basis pursuit as tight as possible in a non-asymptotic sense~\cite{kabanava_robust_2015,Daei2018,Daei2019,genzel2017cosparsity}. To the best of our knowledge, the most precise prediction of the associated phase transition is provided in~\cite{genzel2017cosparsity}, where the incoherence structure of the analysis operator is explicitly taken into account. However, even this result still exhibits a gap between the predicted and the actual location of the phase transition for TV minimization in 1D, which becomes especially striking when the gradient-sparsity is small compared to the ambient dimension (cf.~\cite[Sec.~3.2]{genzel2017cosparsity}).  

Finally, of particular importance for our work are the findings of \citeauthor{cai_guarantees_2015} in~\cite{cai_guarantees_2015}. Indeed, at first sight, their main result seems to imply a negative answer to our initial question on TV-minimization in 1D with Gaussian measurements:
\begin{theorem}[\protect{\cite[Thm.~2.1]{cai_guarantees_2015}}]\label{thm:cai}
	Let $\A \in \R^{m \times n}$ be a standard Gaussian random matrix. Then we have the following:
	\begin{thmlist}
	\item\label{thm:cai:upper}
		There exist universal constants $C_1,C_2,C_3,C_4 > 0$ such that for
		\begin{equation}\label{eq:cai:meas}
			m \geq C_1\cdot \sqrt{s n} \cdot (\log (n) + C_2)
		\end{equation}
		the following holds true with probability at least ${1-C_3 e^{-C_4\sqrt{m}}}$: Every $s$-gradient-sparse vector $\grtr \in \R^n$ is exactly recovered via TV minimization \eqref{eq:intro:tv-1} with input $\y = \A \grtr \in \R^m$. 
	\item\label{thm:cai:lower}
		For every $\eta \in \intvop{0}{1}$, there exist constants $C_1,C_2>0$ depending on $\eta$ and a universal constant $C_3>0$ such that for $C_1 \leq \s < n/4-1$, the following statement holds true with probability at least $1-\eta$: There exist infinitely many $\grtr \in \R^n$ with $\cardinality{\supp (\TV \grtr)} = s$ such that TV minimization~\eqref{eq:intro:tv-1} does not recover $\grtr$ from $m\leq C_2 \cdot \sqrt{s n} - C_3$ measurements. This set particularly contains signals with \emph{(alternating) dense jumps}, which satisfy $\sign(\TV \grtr) = \{(-1)^1, (-1)^2, \dots, (-1)^{\s}, 0, \dots, 0\}$.
	\end{thmlist}
\end{theorem}
The conclusion from this theorem is as surprising as it is discouraging: It reveals that the threshold for successful recovery of $s$-gradient-sparse signals via \eqref{eq:intro:tv-1}  is essentially given by $\sqrt{s n}$-many Gaussian measurements. Remarkably, this rate does not resemble the standard criterion ${m \gtrsim \s \cdot \PolyLog (n,\s)}$. This underpins the special role of 1D~TV because the latter rate does apply to TV in two and more spatial dimensions, as shown in~\cite[Sec.~6]{cai_guarantees_2015}. In fact, the $\sqrt{s n}$-rate promoted by Theorem~\ref{thm:cai} is significantly worse, in particular, when the gradient-sparsity $s$ is relatively small compared to the ambient dimension $n$.

It is noteworthy that the previous result is formulated uniformly across all $s$-gradient-sparse signals. Indeed, the proof of Theorem~\ref{thm:cai}\ref{thm:cai:upper} is based on verifying a nullspace property that is suitably adapted to minimizing the $\l{1}$-gradient semi-norm. On the other hand, the lower bound of Theorem~\ref{thm:cai}\ref{thm:cai:lower} is established by considering signals that are somewhat ``unnatural'' in the context of TV minimization; these \emph{dense-jump signals} are also visualized further below in Figure~\ref{fig:intro:n30} and~\ref{fig:intro:n50}. Hence, the uniform $\sqrt{sn}$-rate of Theorem~\ref{thm:cai} just describes the worst-case performance on the class of all $s$-gradient-sparse signals. 
In other words, Theorem~\ref{thm:cai} states that there is no hope to lower the required sampling rate when asking for recovery of \emph{all} $s$-gradient-sparse signals simultaneously. Nevertheless, one might still wonder whether a meaningful restriction of the class of gradient-sparse signals allows for an improvement of the situation. In fact, the numerical experiments of~\cite[Sec.~5]{cai_guarantees_2015} and~\cite[Sec.~3.2]{genzel2017cosparsity} indicate that recovery of more natural piecewise constant signals succeeds with significantly fewer measurements.

\subsection{Our Contributions and Overview}
\label{subsec:intro:contrib}

The main contribution of this work consists in breaking the $\sqrt{sn}$-complexity barrier of Theorem~\ref{thm:cai}. Taking a non-uniform perspective, we show that a large class of piecewise constant signals is already recoverable from $m \gtrsim \s \cdot \PolyLog (n)$ (sub-)Gaussian measurements.\footnote{For the sake of clarity, we will only consider the case of Gaussian measurements in this article. However, our main results can be easily extended to the sub-Gaussian case, i.e., $\A$ has i.i.d.\ isotropic, sub-Gaussian rows; see \cite[Subsec.~6.1]{genzel2017cosparsity} for more details. Apart from that, we emphasize that, despite their practical limitations, (sub-)Gaussian measurement ensembles form a generic and widely accepted benchmark for the analysis of compressed sensing algorithms. } More specifically, we will address the following aspects:
\begin{deflist}
\item 
	Our main result (see Theorem~\ref{thm:results:exact} in Subsection~\ref{subsec:results:exact}) is a signal-dependent guarantee for exact recovery via \eqref{eq:intro:tv-1}.
	Informally speaking, we show that recovery of an $\s$-gradient-sparse signal $\grtr \in \R^n$ succeeds with high probability as long as $m \gtrsim \SC^{-1} \cdot \s \log^2(n)$, where $\SC \in \intvopcl{0}{1}$ is a \emph{separation constant} for $\grtr$, which controls the minimal distance of its jump discontinuities (see Definition~\ref{def:results:msc}).
	For signals with well-separated jump discontinuities, $\SC$ can be chosen of constant order (i.e., independently of $\s$ and $n$) so that Theorem~\ref{thm:results:exact} indeed yields the near-optimal sampling rate of $m \gtrsim \s \cdot \PolyLog (n)$.
	This class of ``natural'' signals particularly encompasses those arising from a discretization of piecewise constant functions defined on an interval; see Subsection~\ref{subsec:results:signals} and Corollary~\ref{cor:results:exact:cont} for further details.
	The numerical simulation of Figure~\ref{fig:intro} serves as an illustration of our results. 
\item 
	The proof of Theorem~\ref{thm:results:exact} builds upon the non-uniform methodology of~\cite{tropp2014convex,amelunxen2014edge,chandrasekaran2012geometry,stojnic2009gordon}, which in our case suggests studying the \emph{conic mean width} of $\lnorm{\TV(\cdot)}[1]$ at $\grtr$. The major technical achievement of our work is to provide an informative upper bound for this implicit complexity parameter (see Theorem~\ref{thm:results:mwbound} in Subsection~\ref{subsec:results:mw}).
	
	A brief roadmap for the proof of Theorem~\ref{thm:results:mwbound} can be found in Subsection~\ref{subsec:results:mw}, while the formal proof is presented in Subsection~\ref{subsec:proofs:mw}. In a nutshell, the key difficulty is a box-constrained least-squares problem which does not have a closed-form solution. However, an appropriate rotation with a \emph{non-dyadic} Haar wavelet transform will allow us to ``decouple'' the variables and thereby to construct a good approximate solution. An important step in this process is to interpret the gradient support of $\grtr$ as a binary tree, whose structure depends only on the position of the jump discontinuities; in particular, the better the jumps are separated, the more balanced is the resulting tree and the closer is the separation constant~$\SC$ to~$1$. 
	We believe that the development of such a proof strategy is of independent interest and might find application to other (analysis-based) convex recovery methods.
\item 
	By combining the results of~\cite{genzel2017cosparsity} with the previously discussed upper bound for the conic mean width, our results naturally extend to stable and robust recovery (see Theorem~\ref{thm:results:stable} in Subsection~\ref{subsec:results:stable}). 
\end{deflist}

\begin{figure}[ht!]
	\centering
	\begin{subfigure}[t]{0.4\textwidth}
		\centering
		\includegraphics[width=\textwidth]{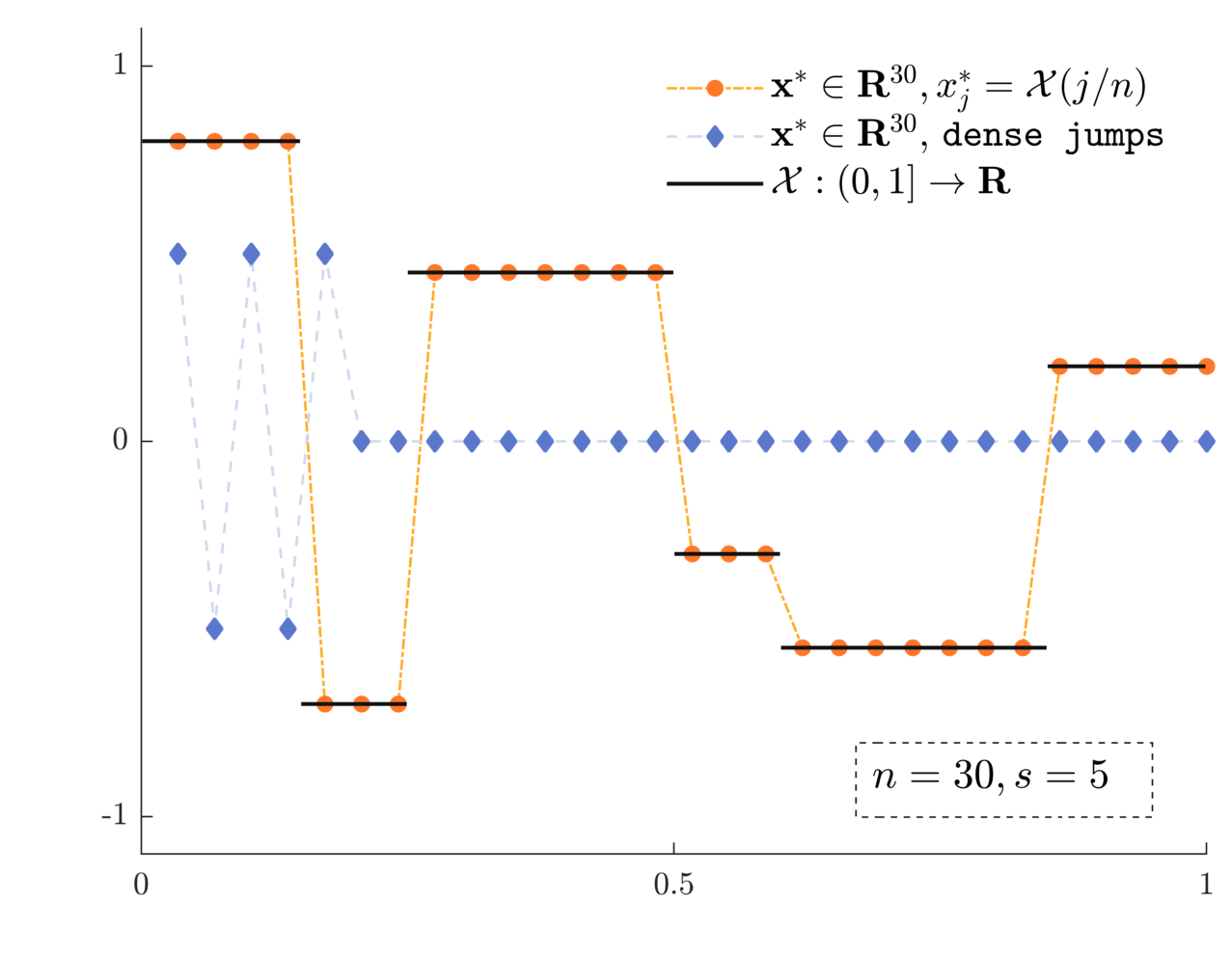}
		\caption{}
		\label{fig:intro:n30}
	\end{subfigure}%
	\qquad
	\begin{subfigure}[t]{0.4\textwidth}
		\centering
		\includegraphics[width=\textwidth]{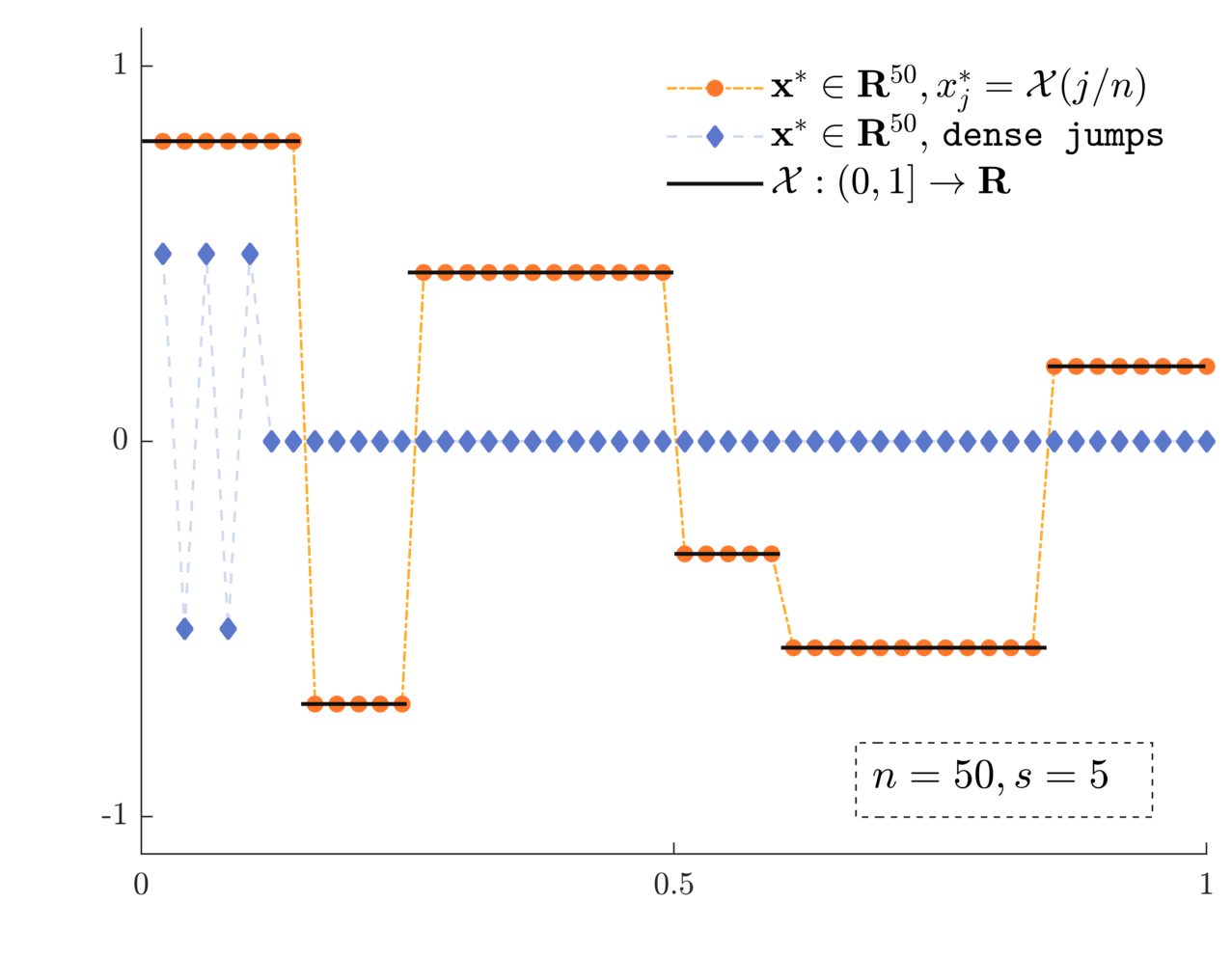}
		\caption{}
		\label{fig:intro:n50}
	\end{subfigure}
	\vspace{1em}
	\begin{subfigure}[t]{0.4\textwidth}
		\centering
		\includegraphics[width=\textwidth]{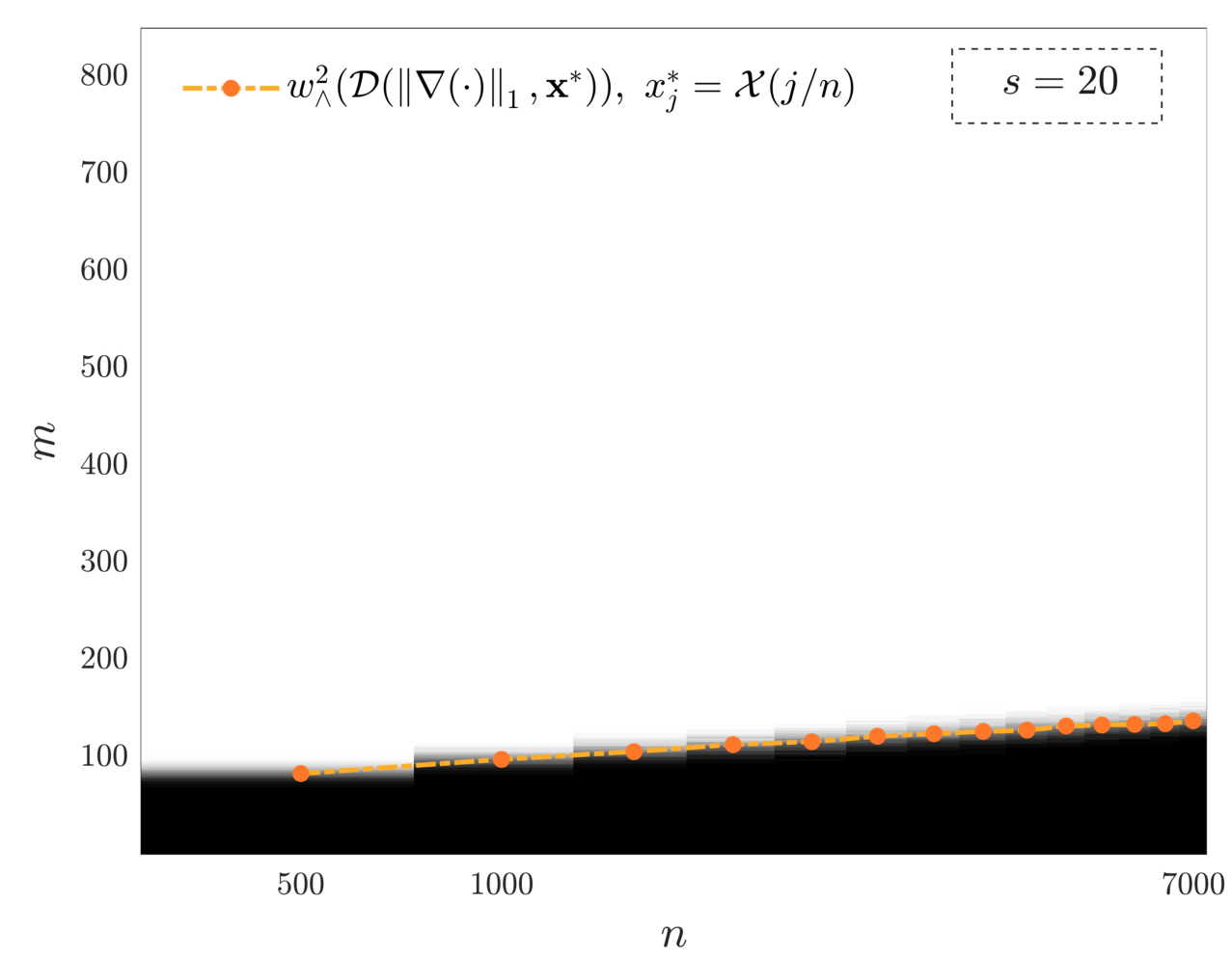}
		\caption{}
		\label{fig:intro:max}
	\end{subfigure}%
	\qquad
	\begin{subfigure}[t]{0.4\textwidth}
		\centering
		\includegraphics[width=\textwidth]{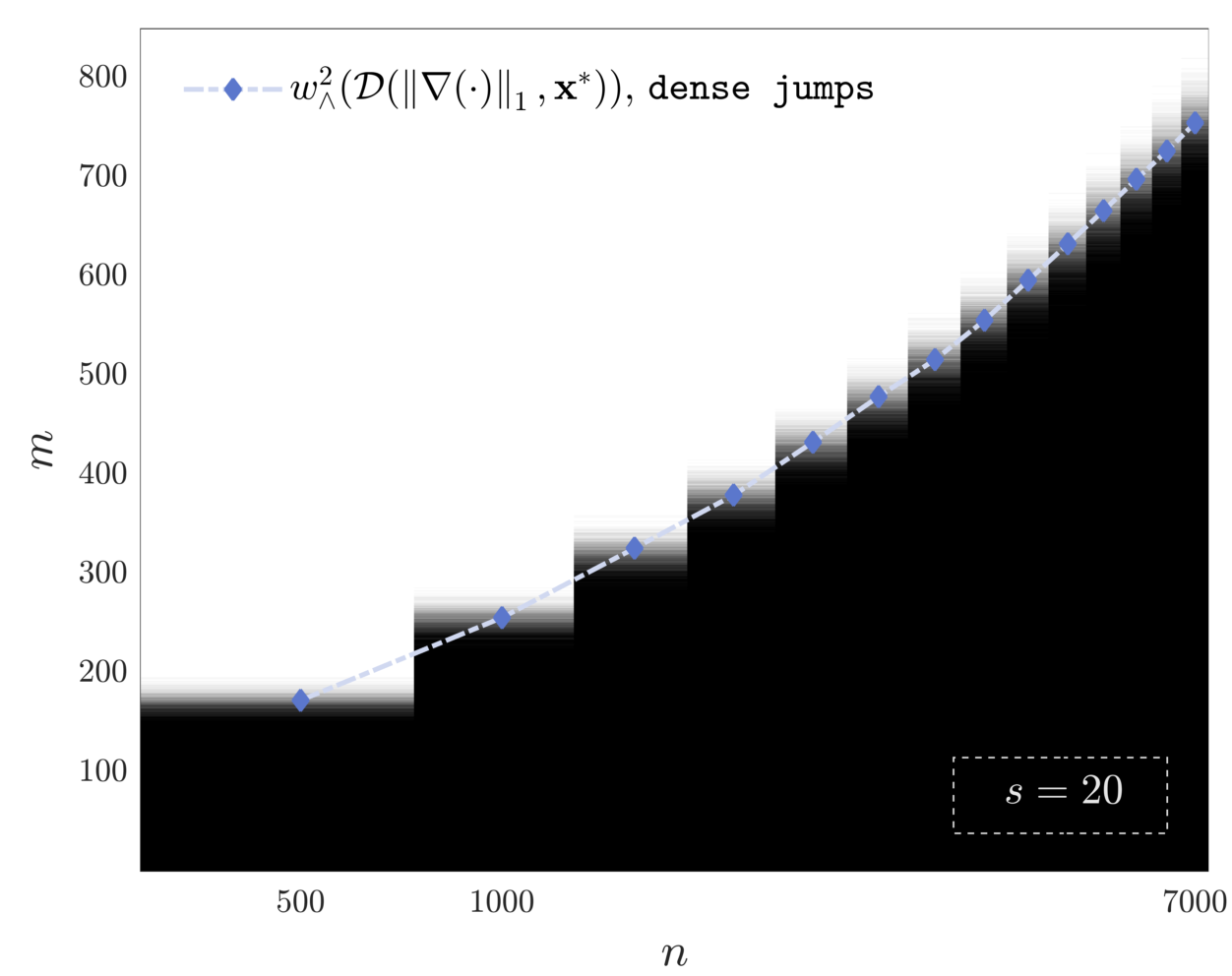}
		\caption{}
		\label{fig:intro:dense}
	\end{subfigure}%
	\caption{\textbf{Numerical simulation visualizing the recovery performance of \eqref{eq:intro:tv-1}.} Subfigure~\subref{fig:intro:n30} and~\subref{fig:intro:n50} show schematic examples of the signal classes that are considered in this experiment at different resolution levels ($n = 30$ and $n = 50$). The orange signal (with circle symbols) is defined as discretization of the piecewise constant function $\contgrtr \colon \intvopcl{0}{1} \to \R$ with $\s = 5$ jump discontinuities that is plotted in black. The blue plot (with diamond symbols) shows a so-called \emph{dense-jump signal} (see Theorem~\ref{thm:cai}\ref{thm:cai:lower}), which does not match the intuitive notion of a $5$-gradient-sparse signal; note that the spatial location of the jumps is chosen adaptively to the resolution level here, which does not correspond to a discretization of a piecewise constant function. For each signal class we have created phase transition plots: Subfigure~\subref{fig:intro:max} and~\subref{fig:intro:dense} display the empirical probability of successful recovery via TV minimization~\eqref{eq:intro:tv-1} for different pairs of ambient dimension $n$ and number of measurements $m$; note the horizontal axis uses a logarithmic scale. The corresponding grey tones reflect the observed probability of success, reaching from certain failure (black) to certain success (white). Additionally, we have estimated the conic (Gaussian) mean width of $\lnorm{\TV(\cdot)}[1]$ at $\grtr$ (denoted by $\effdim[\conic]{\descset{\lnorm{\TV(\cdot)}[1], \grtr}}$), which is known to precisely capture the phase transition (see Subsection~\ref{subsec:results:mw} for details). The result of Subfigure~\subref{fig:intro:dense} confirms that the class of dense-jump signals suffers from the $\sqrt{\s n}$-bottleneck of Theorem~\ref{thm:cai}. On the other hand, Subfigure~\subref{fig:intro:max} suggests that the discretized signals seem to break this bottleneck, indicating the common logarithmic growth in~$n$.}
	\label{fig:intro}
\end{figure}

From a conceptual point of view, our approach is closely related to the findings of~\cite{genzel2017cosparsity}:
There, a novel sampling-rate bound was proposed that is based on a generalized notion of sparsity, taking the support and the coherence structure of the underlying analysis operator into account.
This result defies the conventional wisdom that promotes the analysis sparsity as the crucial complexity parameter.
In the present work, we go one step further and examine this general phenomenon in the seemingly simple case of TV minimization: The homogeneous coherence structure of the gradient operator $\TV$ suggests that the situation is not significantly more complicated than for orthonormal bases.
However, our results reveal that the gradient support, and especially the (relative) location of discontinuities, is crucial. 
As a consequence, we are able to identify subclasses of the set of $s$-gradient-sparse signals for which the (asymptotic) sample complexity differs substantially. 

In general, we believe that such additional structural properties become important when going beyond orthonormal bases as sparsifying system. Nevertheless, to the best of our knowledge, this aspect has largely gone unnoticed in the literature.  
In this regard, it is worth mentioning that there are several approaches to compressed sensing that take additional support properties into account, such as \emph{joint/group sparsity}~\cite{Duarte2005,Yuan2006} or \emph{model-based compressed sensing}~\cite{baraniuk2010modelcs}. However, these approaches are based on adapting the regularizer of the recovery method to a refined sparsity model, while we stick to vanilla TV minimization in 1D. In~\cite{adcock2017}, it is also argued that sparsity on its own does not govern the reconstruction quality and that the structure of the coefficients needs to be considered. The latter work is nevertheless tailored to multilevel transforms, for which the sparsity appears in different levels. Furthermore, the need to address multilevel sparsity arises from the specific incoherence pattern with the underlying structured measurement setup. Such a constraint is certainly not relevant to the Gaussian model considered in our work. 

Finally, regarding TV-based recovery, recall that a notion similar to the above-mentioned \mbox{$\SC$-}sep\-ara\-bil\-ity also appears in~\cite{Poon2015}. However, this analysis addresses a specific recovery task from (low-frequency) Fourier samples, whereas we focus on Gaussian measurement matrices. Moreover, we note that the TV-Fourier combination allows for an important simplification of the problem, since the gradient operator commutes with the Fourier transform (differentiation forms a Fourier multiplier). This is also reflected by the fact that the $\sqrt{s n}$-bottleneck does not exist for 1D~TV~compressed sensing with Fourier measurements (see~\cite{candes2006cs,Poon2015}). 

\subsection{Notation}
\label{subsec:intro:notation}

The letter $C$ is reserved for constants, whose values could change from time to time, and we speak of a \emph{universal constant} if its value does not depend on any other involved parameter.
If an inequality holds up to a universal constant $C$, we usually write $A \lesssim B$ instead of $A \leq C \cdot B$; the notation $A \asymp B$ means that both $A \lesssim B$ and $B \lesssim A$ hold true.

For $d \in \N$, we set $[d] \coloneqq \{1, \dots, d\}$. The \emph{cardinality} of an index set $\mathcal{I} \subset [d]$ is denoted by $\cardinality{\mathcal{I}}$ and its \emph{set complement} is given by $\setcompl{\mathcal{I}} \coloneqq [d] \setminus \mathcal{I}$. Vectors and matrices are denoted by lower- and uppercase boldface letters, respectively. 
The \mbox{$j$-th} entry of a vector $\v \in \R^d$ is denoted by $[\v]_j$, or simply by $v_j$ if there is no danger of confusion.
Similarly, the restriction of $\v \in \R^d$ to an index set $\mathcal{I} \subset [d]$ is denoted by $[\v]_{\mathcal{I}} \in \R^{\cardinality{\mathcal{I}}}$ or simply $\v_{\mathcal{I}}$.
The \emph{zero vector} (in $\R^d$) is denoted by $\vnull$, while $\1_d \in \R^d$ is the \emph{all-ones vector} and $\I{d} \in \R^{d \times d}$ the \emph{identity matrix}.

The \emph{support} of $\v \in \R^d$ is defined as $\supp(\v) \coloneqq \{ j \in [d] \suchthat v_j \neq 0 \}$, and its cardinality is referred to as the \emph{sparsity} of $\v$.
For $1 \leq p \leq \infty$, we denote the \emph{$\l{p}$-norm} on $\R^d$ by $\lnorm{\cdot}[p]$, and the \emph{Euclidean unit sphere} is given by $\S^{d-1} \coloneqq \{ \v \in \R^d \suchthat \lnorm{\v} = 1 \}$. 
Furthermore, we write $\cone{\sset}$ for the \emph{conic hull} of a subset $\sset \subset \R^d$ and $\indset{\sset} \colon \R^d \to \{0, 1\}$ denotes the \emph{indicator function} (or \emph{step function}) of $\sset$.

The \emph{expected value} is denoted by $\mean{}_{\gaussian}[\cdot]$, where the subscript indicates that the expectation is computed with respect to a certain random variable/vector $\gaussian$. Moreover, we write $\gaussian \distributed \Normdistr{\vnull}{\I{d}}$ if $\gaussian$ is a \emph{standard Gaussian random vector} in $\R^d$.

Let $v \in \R$ and $s \geq 0$. The \emph{sign} of $v$ is denoted by $\sign(v)$, where $\sign(0) = 0$, and if $\sign(\cdot)$ is applied to a vector, this operation is understood entrywise. We also define the \emph{clip function} $\clip{v}{s} \coloneqq \sign(v) \cdot \min\{\abs{v}, s\}$ and the \emph{positive part} $\pospart{v} \coloneqq \max\{v, 0\}$.
By $\round{\cdot}$, $\ceil{\cdot}$, and $\floor{\cdot}$, we denote the \emph{rounding}, \emph{ceiling}, and \emph{floor function}, respectively.
Finally, $\PolyLog(\cdot)$ is the generic notation for a \emph{polylogarithmic function}.

\section{Main Results}
\label{sec:results}

In this section, we present the main results of this work. We begin with the underlying signal model in Subsection~\ref{subsec:results:signals}, which is motivated by the discretization of piecewise constants functions defined on a continuous domain. In this context, we also introduce the notion of $\SC$-separated signals (see Definition~\ref{def:results:msc}), which is the only technical requirement for understanding our recovery guarantees. Hence, equipped with this notion the hurried reader may jump directly to Subsection~\ref{subsec:results:exact}, where the main result of Theorem~\ref{thm:results:exact} and its implications are discussed.
The proof strategy of Theorem~\ref{thm:results:exact} is outlined subsequently in Subsection~\ref{subsec:results:mw}, including the basic geometric ideas as well as a brief roadmap for the proof of our mean-width bound in Theorem~\ref{thm:results:mwbound}.
An extension to stable and robust recovery is then derived in Subsection~\ref{subsec:results:stable}. 
Finally, we point out some limitations and possible refinements of our analysis in Subsection~\ref{subsec:results:limitations}.

\subsection{Gradient-Sparse Signals}
\label{subsec:results:signals}

As foreshadowed in the introduction, the $\sqrt{s n}$-bottleneck of Theorem~\ref{thm:cai} is due to the fact that this result addresses \emph{uniform} recovery across the class of all $\s$-gradient-sparse signals. In particular, the worst-case sampling rate is attained for the somewhat artificial dense-jump signals. While their gradients are certainly sparse, such signals do not match the intuitive notion of being piecewise constant. 

In contrast, the goal of this work is to demonstrate that the $\sqrt{s n}$-bottleneck can be broken for ``more natural'' piecewise constant signals. In order to make our approach precise, we need to introduce the concept of $\SC$-separation, which allows us to control the minimal distance of jump discontinuities and thereby to exclude the aforementioned pathological examples.
A very accessible path towards the definition of $\SC$-separation is to view a signal vector as discretization of a piecewise constant function defined on an interval, say $\intvopcl{0}{1}$. Therefore, we first specify $\cont\SC$-separation in the continuous setting:\footnote{The `$\circ$'-notation indicates that the corresponding object is defined in the continuous setting. This allows the reader for an easy differentiation from analogous parameters in the discrete setting.}
\begin{definition}[Separation constant -- continuous version]\label{def:results:contmsc}
	Let $\contgrtr \colon \intvopcl{0}{1} \to \R$ be a piecewise constant function such that\footnote{Note that the use of \emph{half-open} intervals $\intvopcl{0}{1}$ and $\intvopcl{\cont\jump_{i-1}}{\cont\jump_i}$ is not of particular importance here, but rather for the sake of convenience.}
	\begin{equation}\label{eq:results:contmsc:signal}
	\contgrtr(t) = \sum_{i = 1}^{\s+1} h_i \cdot \indset{\intvopcl{\cont\jump_{i-1}}{\cont\jump_i}}(t)
	\end{equation}
	for certain \emph{jump discontinuities} $\cont\jump_1, \dots, \cont\jump_{\s} \in \intvop{0}{1}$ with $0 \eqqcolon \cont\jump_0 < \cont\jump_1 < \dots < \cont\jump_{\s} < \cont\jump_{\s+1} \coloneqq 1$ and \emph{level set coefficients} $h_1, \dots, h_{\s+1} \in \R$ with $h_i \neq h_{i+1}$ for all $i = 1, \dots, \s$.
	We say that $\contgrtr$ is \emph{$\cont\SC$-separated} for some \emph{separation constant} $\cont\SC > 0$ if
	\begin{equation}\label{eq:results:contmscc}
	\min_{i \in [\s+1]}\abs{\cont\jump_i - \cont\jump_{i-1}} \geq \frac{\cont\SC}{\s+1}.
	\end{equation}
\end{definition}

Intuitively, a separation constant $\cont\SC$ for $\contgrtr$ measures the deviation of its jump discontinuities from an equidistant singularity pattern: The larger $\cont\SC \in \intvopcl{0}{1}$, the better the jumps are separated. Indeed, in the optimal case of equidistantly distributed singularities, we have that $\cont\jump_i = i/(\s+1)$ and therefore $\cont\SC = 1$ is a valid choice, independently of the total number of discontinuities $s$. In this situation, $\cont\SC$ is as large as possible, since the separation constant is always bounded by $1$:
\begin{equation}
\cont\SC \leq (\s + 1) \cdot \Big( \min_{i \in [\s+1]}\abs{\cont\jump_i - \cont\jump_{i-1}} \Big) \leq (\s + 1) \cdot \Big( \tfrac{1}{\s+1} \underbrace{\sum_{i = 1}^{\s+1} (\cont\jump_i - \cont\jump_{i-1}) }_{ = 1} \Big) = 1.
\end{equation}
On the other hand, the closer the minimal distance between the jumps of $\contgrtr$, the smaller becomes~$\cont\SC$. Consequently, the separation constant $\cont\SC$ reflects the ratio between the minimal jump distance of $\contgrtr$ and its equidistant counterpart with the same number of discontinuities.

Of particular importance for this work is the following discrete analog of the previous definition. It naturally arises when the continuous domain $\intvopcl{0}{1}$ is replaced by the discrete domain $\{1, \dots, n\}$:
\begin{definition}[Separation constant -- discrete version]\label{def:results:msc}
	Let $\grtr \in \R^n$ be a signal with $\s > 0$ \emph{jump discontinuities} such that $\supp(\TV\grtr) = \{\jump_1,\dots,\jump_{\s}\}$ with $0 \eqqcolon \jump_0 < \jump_1 < \dots < \jump_{\s} < \jump_{\s+1} \coloneqq n$.
	We say that $\grtr$ is \emph{$\SC$-separated} for some \emph{separation constant} $\SC > 0$ if
	\begin{equation}\label{eq:results:msc}
	\min_{i \in [\s+1]} \frac{\abs{\jump_i - \jump_{i-1}}}{n} \geq \frac{\SC}{\s+1}.
	\end{equation}
\end{definition}
Clearly, one may interpret the entries of $\grtr$ as discrete function values on a \emph{nodal/cell-centered} grid and each entry $[\TV\grtr]_j$ as a \emph{face-staggered} finite difference between the $j$-th and $(j+1)$-th node; see Figure~\ref{fig:proofs:mw:face_edges} for an illustration. Hence, each $\jump_i \in \supp(\TV\grtr)$ corresponds to a jump of height $[\TV\grtr]_{\jump_i}$ between the $i$-th and $(i+1)$-th constant segment of $\grtr$. In total, there are $s$ jumps and $s+1$ constant segments for a signal vector $\grtr$ with $\cardinality{\supp(\TV\grtr)} = \s$.

With this interpretation in mind, the intuition behind the separation constant for a gradient-sparse vector $\grtr$ is exactly the same as in the continuous setting: it measures the deviation of the jump discontinuities of $\grtr$ from an equidistant singularity pattern. In particular, it is not hard to see that the separation constant $\SC$ for $\grtr$ can always be chosen such that $(s+1)/n \leq \SC \leq 1$, where larger values of $\SC$ indicate that the gradient support is closer to being equidistant. 
Although such a notion of separation is natural, it still allows for degenerate cases in which $\SC$ depends on the ambient dimension $n$. For instance, in the worst-case scenario of dense-jump signals, we would have that $\SC = (\s+1)/n$.

The following proposition demonstrates how the continuous and discrete versions of separation constants can be connected by agreeing on a specific discretization procedure for $\contgrtr$. For the sake of simplicity, we focus on the obvious strategy of taking pointwise samples on a uniform grid, but other, more sophisticated discretization schemes would certainly lead to similar results.  
\begin{proposition}\label{prop:results:discrcont}
	Let $\contgrtr \colon \intvopcl{0}{1} \to \R$ be a piecewise constant function defined according to \eqref{eq:results:contmsc:signal}. For $n \geq \s+1 $, let $\grtr = (x_1^\ast, \dots, x_n^\ast) \in \R^n$ be defined by
	\begin{equation}\label{eq:results:discrcont:discrscheme}
		x_j^\ast \coloneqq \contgrtr\big( \tfrac{j}{n} \big), \quad j = 1, \dots, n.
	\end{equation}
	If $\contgrtr$ is $\cont\SC$-separated with $\cont\SC \geq (\s+1)/n$, then $\cardinality{\supp(\TV\grtr)} = \s$ and $\grtr$ is $\SC$-separated with $\SC > 0$ satisfying
	\begin{equation}\label{eq:results:discrcont:scbound}
		\SC > \cont\SC - \frac{\s+1}{n}.
	\end{equation}
	In particular, if $n \geq C(\s+1)/\cont\SC$ for any constant $C > 1$, we have that
	\begin{equation}
		\SC > (1 - \tfrac{1}{C}) \cont\SC.
	\end{equation}
\end{proposition} 

\begin{figure}
	\centering
	
	\def\tikz_n{10}
	\begin{tikzpicture}
		
		\definecolor{myo}{rgb}{1,0.46,0.15}
		\definecolor{myb}{rgb}{0.35,0.46,0.8}
		\definecolor{myb2}{rgb}{0.81,0.84,0.94}
		\definecolor{Sset}{gray}{0.5}
		\definecolor{myt}{rgb}{0.95,0.75,0.3}
		\definecolor{Sset2}{rgb}{0.35,0.46,0.8}
		
		\tikzstyle{signaldots}=[mark size=1.6pt,color=black,rotate=-90]
		\tikzstyle{jumpdots}=[mark size=2pt,color=gray,rotate=-90]
		\tikzstyle{jumpdotsS}=[mark size=2pt,color=myo,rotate=-90] 
		\tikzstyle{treeS}=[mark size=2pt,color=myo,rotate=-90]
		\tikzstyle{treeSc}=[mark size=2pt,color=gray,rotate=-90]
		
		\draw[-] (1,0) -- (\tikz_n,0);
		\draw[-,color=gray!50] (\tikz_n,0) -- (\tikz_n+0.5,0);
		\draw[-,color=gray!50] (0.5,0) -- (1,0);
		
		\foreach \x in {2,2,3,\tikz_n-1,\tikz_n,\tikz_n}
		\node[mark size=2pt,color=gray,rotate=-90] at (\x-0.5,0) {\pgfuseplotmark{triangle*}};
		\node[mark size=2pt,color=gray!50,rotate=-90] at (1-0.5,0) {\pgfuseplotmark{triangle*}};
		\node[mark size=2pt,color=gray!50,rotate=-90] at (\tikz_n+0.5,0) {\pgfuseplotmark{triangle*}};
		
		\foreach \x in {1,2,3,\tikz_n-1,\tikz_n}
		\node[signaldots] at (\x,0) {\pgfuseplotmark{square*}};
		
		\foreach \x in {1,2}
		\draw[shift={(\x,-2pt)},color=black] (0,0) node[below] {$\x$};
		
		\draw[shift={(3,-2pt)},color=black] (0,0) node[below] {\phantom{1}$\dots$};
		\draw[shift={(\tikz_n-1,-2pt)},color=black] (0,0) node[below] {$n-1$};
		\draw[shift={(\tikz_n,-2pt)},color=black] (0,0) node[below] {$\phantom{1}n$};
		
		\foreach \x in {1,2}
		\draw[shift={(\x+0.5,1pt)},color=gray] (0pt,1pt) node[above] {$\x$};
		\draw[shift={(0.5,1pt)},color=gray!50] (0pt,1pt) node[above] {$0$};
		
		\draw[shift={(\tikz_n-1.5,1pt)},color=gray] (0pt,1pt) node[above] {\phantom{1}$\dots$};
		\draw[shift={(\tikz_n-0.5,1pt)},color=gray] (0pt,1pt) node[above] {$N$};
		\draw[shift={(\tikz_n+0.5,1pt)},color=gray!50] (0pt,1pt) node[above] {$n$};
		
		\draw[shift={(-1,1.5pt)},color=gray] (0pt,1.5pt) node[above] {Faces:};
		\draw[shift={(-1,-4pt)},color=black] (0pt,3pt) node[below] {Nodes:};
		
	\end{tikzpicture}
	\caption{\textbf{Nodes and (staggered) faces of a signal in $\R^n$.} For technical reasons, we also visualize the ``ghost faces'' $0$ and $n$ (cf.~Definition~\ref{def:results:contmsc} and~\ref{def:results:msc}).}
	\label{fig:proofs:mw:face_edges}
\end{figure}
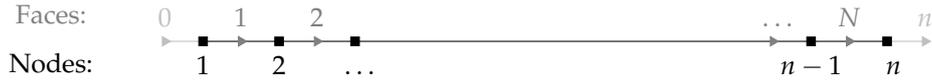

The proof of the previous result is deferred to Subsection~\ref{subsec:proofs:discrcont}. The statement of Proposition~\ref{prop:results:discrcont} shows that the discretization of a well-separated piecewise constant function is also well-sep\-arat\-ed, supposed that the grid resolution is fine enough, i.e., $n$ is sufficiently large. As motivated at the beginning of this section, we are primarily interested in scenarios where the ambient dimension is large and satisfies $n \gg \s$. In this case, the `in-particular'-part of Proposition~\ref{prop:results:discrcont} confirms that the separation constant $\SC$ for $\grtr$ can be chosen such that $\SC \approx \cont\SC$.
This relationship particularly allows us to speak of ``natural'' signal vectors when referring to discretizations of piecewise constant functions defined on an interval.
Indeed, if the grid is chosen fine enough, Proposition~\ref{prop:results:discrcont} ensures that the resulting jump discontinuities of $\grtr$ are well-separated and thereby excludes degenerate cases such as dense-jump signals.




We close our discussion with some additional remarks:
\begin{remark}
	\begin{rmklist}
	\item 
		Note that we assume that $\s > 0$ in Definition~\ref{def:results:msc}, which excludes the not particularly interesting scenario of a constant signal vector. Since one would have $\grtr \in \ker\TV$ in this special case, it needs to be treated slightly differently: according to \cite[Rmk.~2.6]{genzel2017cosparsity}, recovery via TV minimization~\eqref{eq:intro:tv-1} would succeed almost surely with $m\geq 1$, since $\ker\TV$ is a one-dimensional subspace of $\R^n$.
	\item
		We do not assume that $\SC$ is the largest possible constant such that the condition \eqref{eq:results:msc} is satisfied. However, one can certainly choose $\SC \coloneqq (\s+1) \cdot \min_{i \in [\s+1]} \abs{\jump_i - \jump_{i-1}} / n$ and then refer $\SC$ to as the \emph{maximal separation constant} for $\grtr$. The same convention is made for the condition~\eqref{eq:results:contmscc} in the continuous setting of Definition~\ref{def:results:contmsc}.
	\item 
		One could also show a reversed version of Proposition~\ref{prop:results:discrcont} by constructing a piecewise constant function $\contgrtr \colon \intvopcl{0}{1} \to \R$ from a gradient-sparse signal $\grtr \in \R^n$ in such a way that $\grtr$ is the discretization of $\contgrtr$. \qedhere
	\end{rmklist}
\end{remark}	

\subsection{TV Minimization for Exact Recovery}
\label{subsec:results:exact}

With the notion of $\SC$-separation from Definition~\ref{def:results:msc} at hand, we are now ready to state our main guarantee for noiseless signal recovery via TV minimization~\eqref{eq:intro:tv-1}. It is worth emphasizing that the following theorem does not rely on any particular signal generation process, such as the discretization described in Proposition~\ref{prop:results:discrcont}. In fact, this result applies to every gradient-sparse signal that satisfies the mild condition~\eqref{eq:results:exact:n}. 
\begin{theorem}[Exact recovery via TV minimization]\label{thm:results:exact}
	Let $\grtr \in \R^n$ be a $\SC$-separated signal with $\s>0$ jump discontinuities and 
	\begin{equation}\label{eq:results:exact:n}
	\SC \geq \frac{8 \s}{n}.
	\end{equation}
	Let $\probsuccess > 0$ and assume that $\A \in \R^{m \times n}$ is a standard Gaussian random matrix with
	\begin{equation}\label{eq:results:exact:meas}
	m \gtrsim \frac{1}{\SC} \cdot \s\log^2(n) + \probsuccess^2.
	\end{equation}
	Then with probability at least $1 - e^{-\probsuccess^2/2}$, TV minimization \eqref{eq:intro:tv-1} with noiseless input $\y = \A \grtr \in \R^m$ recovers $\grtr$ exactly.
\end{theorem}
The proof of Theorem~\ref{thm:results:exact} is a direct consequence of Proposition~\ref{prop:results:exact_geo} and Theorem~\ref{thm:results:mwbound}, which are presented subsequently in Subsection~\ref{subsec:results:mw}. Before discussing its consequences, we wish to point out that Theorem~\ref{thm:results:exact} provides a \emph{non-uniform} and \emph{signal-dependent} guarantee, in the sense that it concerns the successful recovery of a \emph{fixed} signal $\grtr$ (for a single random draw of $\A$) with a sampling rate that depends on $\grtr$ (in terms of $\SC$). 
This is in stark contrast to Theorem~\ref{thm:cai}, which addresses recovery of all $s$-gradient-sparse signals simultaneously.
We refer to Section~\ref{sec:conclusion} for a more detailed discussion of (non-)uniformity and signal dependence in the context of our results.

Most importantly, Theorem~\ref{thm:results:exact} allows us to break the $\sqrt{\s n}$-bottleneck of Theorem~\ref{thm:cai}. Indeed, the pivotal condition~\eqref{eq:results:exact:meas} indicates that exact recovery of $\grtr$ via TV minimization~\eqref{eq:intro:tv-1} is already possible with $m \gtrsim \SC^{-1} \cdot \s\log^2(n)$ measurements. However, in order to obtain the desired sampling rate of  $m \gtrsim \s \cdot \PolyLog (n)$, it is indispensable to control the size of the separation constant~$\SC \in  \intvopcl{0}{1}$.

Regarding the crucial role of $\SC$, we first note that the condition~\eqref{eq:results:exact:n} is \emph{not} a severe restriction: Recall that for every $\s$-gradient-sparse signal, the separation constant can be chosen such that $\SC \geq (\s + 1) / n$. In this light, the assumption of~\eqref{eq:results:exact:n} is only slightly stronger, where the suboptimal factor $8$ is just an artifact of our proof.
Alternatively, when rearranging \eqref{eq:results:exact:n} to $n \geq 8s/\SC$, it turns into a mild requirement on the ambient dimension (or resolution level).

However, condition~\eqref{eq:results:exact:n} does still not exclude pathological cases: for instance, the gradient support $\supp(\TV\grtr) = \left\{1,9,17,\dots,8s+1\right\}$ would be admissible, but the sampling-rate condition~\eqref{eq:results:exact:meas} would render to the trivial bound of $m \gtrsim n \log^2(n)$.
Consequently, Theorem~\ref{thm:results:exact} only makes a significant statement if $\SC$ is sufficiently large. In particular, we obtain the near-optimal rate of $m \gtrsim \s \cdot \PolyLog (n)$ if $\SC$ can be chosen independently of $n$ and $s$; see the discussions in Subsection~\ref{subsec:results:limitations} and Section~\ref{sec:conclusion} for an asymptotic point of view on this matter.
A typical example of such a situation is the discretization of a piecewise constant function $\contgrtr \colon \intvopcl{0}{1} \to \R$. Indeed, the following corollary is a straightforward combination of Theorem~\ref{thm:results:exact} and Proposition~\ref{prop:results:discrcont}, demonstrating that $m \gtrsim \s \log^2(n)$ measurements are sufficient for recovery when $\contgrtr$ is finely enough discretized.
\begin{corollary}[Recovery of discretized signals]\label{cor:results:exact:cont}
	Let $\contgrtr \colon \intvopcl{0}{1} \to \R$ be a piecewise constant function defined according to \eqref{eq:results:contmsc:signal} and assume that $\contgrtr$ is $\cont\SC$-separated.
	For $n \geq 16\s/\cont\SC$, let $\grtr \in \R^n$ be the equidistant (pointwise) discretization of $\contgrtr$, i.e.,
	\begin{equation}
	x_j^\ast = \contgrtr\big( \tfrac{j}{n} \big), \quad j = 1, \dots, n.
	\end{equation}
	Then with probability at least $1 - e^{-\probsuccess^2/2}$, TV minimization \eqref{eq:intro:tv-1} recovers $\grtr$ exactly from 
	\begin{equation}\label{eq:results:exact:cont:meas}
	m \gtrsim \frac{1}{\cont\SC} \cdot \s\log^2(n) + \probsuccess^2
	\end{equation}
	noiseless Gaussian measurements.
\end{corollary}
Similarly to the corresponding assumption \eqref{eq:results:exact:n} in Theorem~\ref{thm:results:exact}, the factor $16$ in the (still mild) condition $n \geq 16 \s /\cont\SC$ is just an artifact of the proof. Analogous statements can be expected for all other discretization schemes that allow for a comparable relation between the continuous and discrete setting as shown in Proposition~\ref{prop:results:discrcont}. 


\subsection{The Geometry of Non-Uniform Recovery}
\label{subsec:results:mw}

In this section, we discuss the general proof strategy leading to Theorem~\ref{thm:results:exact}, particularly focusing on the main ideas without technical details. 
The crucial ingredient is Theorem~\ref{thm:results:mwbound} below, which states a novel bound for the conic mean width of the semi-norm $\lnorm{\TV(\cdot)}[1]$.
As such, this section can be seen as a ``roadmap'' for Subsection~\ref{subsec:proofs:mw}, where the full proof of Theorem~\ref{thm:results:mwbound} is presented. 

Our approach towards a recovery guarantee for TV minimization~\eqref{eq:intro:tv-1} is based on a well-estab\-lished, abstract geometric framework. Its goal is to understand the interplay of the gradient-sparsity-promoting functional $f(\cdot) = \lnorm{\TV(\cdot)}[1]$ and the measurement matrix $\A$ in an (infinitesimal) neighborhood of the ground truth signal $\grtr$. We emphasize that such a non-uniform and signal-dependent analysis is especially useful for the problem of TV-based recovery in 1D, where a uniform approach inevitably leads to the $\sqrt{\s n}$-bottleneck (see Theorem~\ref{thm:cai}). 

To begin with the formalization, we require the following notion, which will play an important role in understanding the solution set of~\eqref{eq:intro:tv-1}; note that this definition makes sense for any convex function $f \colon \R^n \to \R$, but for our purposes, it is convenient to assume that $f(\cdot) = \lnorm{\TV(\cdot)}[1]$.
\begin{definition}[Descent set and descent cone]\label{def:results:descentcone}
	Let $f \colon \R^n \to \R$ be a convex function and let $\grtr~\in~\R^n$. The \emph{descent set} of $f$ at $\grtr$ is given by
	\begin{equation}
	\descset{f, \grtr} \coloneqq \{ \v \in \R^n \suchthat f(\grtr + \v) \leq f(\grtr) \}
	\end{equation}
	and its corresponding \emph{descent cone} is denoted by $\desccone{f, \grtr} \coloneqq \cone{\descset{f, \grtr}}$.
\end{definition}
While the shape of the descent cone certainly depends on the (level-set) geometry of $f$ around~$\grtr$, it does not involve the measurement matrix $\A$. 
In fact, our assumption that $\A$ is Gaussian allows us to establish such a connection by means of a single geometric parameter, namely the conic (Gaussian) mean width that is introduced next.
This quantity originates from classical results in geometric functional analysis and asymptotic convex geometry, e.g., see~\cite{milman1984mstar,gordon1988escape,giannopoulos2004asymptotic}, but it has also shown up in equivalent forms as \emph{Talagrand’s $\gamma_2$-functional} in stochastic processes~\cite{talagrand2014chaining} or as \emph{Gaussian complexity} in statistical learning theory~\cite{bartlett2003complexity}. 
While its benefits to signal reconstruction problems date back to~\cite{rudelson2008sparse,mendelson2007reconstruction}, most important for us is a more recent line of research that is concerned with \emph{non-uniform} guarantees for compressed sensing, for instance, see~\cite{stojnic2009gordon,chandrasekaran2012geometry,amelunxen2014edge,plan2015lasso,tropp2014convex,oymak2016sharpmse}.  
\begin{definition}[Mean width] \label{def:results:meanwidth}
	The \emph{(Gaussian) mean width} of a bounded subset $\sset \subset \R^n$ is defined as
	\begin{equation}
	\meanwidth{\sset} \coloneqq \mean{}_{\gaussian} [\sup_{\v \in \sset} \sp{\gaussian}{\v}],
	\end{equation}
	where $\gaussian \distributed \Normdistr{\vnull}{\I{n}}$.
	If $f \colon \R^n \to \R$ is a convex function and $\grtr \in \R^n$, we call
	\begin{equation}
	\meanwidth[\conic]{\descset{f, \grtr}} \coloneqq \meanwidth{\desccone{f, \grtr} \intersec \S^{n-1}}
	\end{equation}
	the \emph{conic mean width} of $f$ at $\grtr$.
\end{definition}

The following proposition reveals that, indeed, the conic mean width  essentially determines the required number of measurements for successful recovery via~\eqref{eq:intro:tv-1}. Note that we formulate this result for the case of TV minimization, but it also holds true for a general convex function instead of the particular choice $f(\cdot) = \lnorm{\TV(\cdot)}[1]$. The geometric idea behind Proposition~\ref{prop:results:exact_geo} is visualized in Figure~\ref{fig:exact_geo}.

\begin{figure}
	\centering
	\begin{tikzpicture}[scale=2]
	\coordinate (K1) at (-1,-.6);
	\coordinate (K2) at (0,-1.6);
	\coordinate (K5) at (-1.6,-1.6);
	\coordinate[below right=0cm and 0cm of K1] (X0);
	
	\coordinate (kerAnchor) at ($(X0)+(20:.2)$);
	\draw[thick] ($(X0)!-3.5!(kerAnchor)$) -- ($(X0)!3.5!(kerAnchor)$);
	\draw[thick,dashed] ($(X0)!-4.5!(kerAnchor)$) -- ($(X0)!4.5!(kerAnchor)$) node[xshift=.7cm] {$\ker\A$};
	
	\filldraw pic[top color=blue!30, bottom color=white, angle radius=3.5cm] {angle=K5--K1--K2};
	
	\node[draw,circle,minimum size=2.3cm] at (X0) {};
	\node[above left=.65 and .65 of X0] {$\S^{n-1}$};
	\begin{scope}
	\clip (K1) -- (K2) -- (K5) -- cycle;
	\node[draw,circle,minimum size=2.3cm,red,ultra thick] at (X0) {};
	\end{scope}
	
	\node at (barycentric cs:K1=0.3,K2=1.5,K5=1.5) {$\desccone{\lnorm{\TV(\cdot)}[1], \grtr}$};
	\node[blackdot] at (X0) {};
	
	\path (X0) -- ++(-60:.55cm) coordinate (anchorSphere);
	\node[above right=0 and 1 of anchorSphere] (anchorSpherelabel) {$\desccone{\lnorm{\TV(\cdot)}[1], \grtr}\intersec \S^{n-1}$};
	\path[<-,shorten <=3pt,>=stealth,bend right,red] (anchorSphere) edge (anchorSpherelabel);
	
	\node[blackdot,label={[label distance=-2pt]above :$\vnull$}] at (X0) {};
	\end{tikzpicture}
	\caption{\textbf{The convex geometry of recovery via~\eqref{eq:intro:tv-1}.} It is straightforward to see that $\grtr$ is the unique minimizer of~\eqref{eq:intro:tv-1} if and only if $\ker\A$ does not intersect the spherical subset $\desccone{f, \grtr} \intersec \S^{n-1}$. The probability of this event boils down to relating the dimension of the random subspace $\ker\A$ (which almost surely equals $n-m$) to the ``size'' of the descent cone (which is measured by the conic mean width). 
		According to Proposition~\ref{prop:results:exact_geo}, the subspace $\ker\A$ misses the spherical subset (red arc) with high probability if the condition~\eqref{eq:results:exact_geo:meas} is satisfied.}
	\label{fig:exact_geo}
\end{figure}
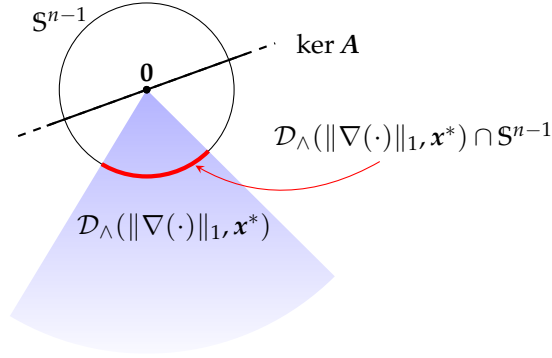

\begin{proposition}[\protect{\cite[Cor.~3.5]{tropp2014convex}}]\label{prop:results:exact_geo}
	Let $\grtr \in \R^n$ be an arbitrary signal vector. Let $\probsuccess > 0$ and assume that $\A \in \R^{m \times n}$ is a standard Gaussian random matrix with
	\begin{equation}\label{eq:results:exact_geo:meas}
		m > \big(\meanwidth[\conic]{\descset{\lnorm{\TV(\cdot)}[1], \grtr}} + \probsuccess\big)^2 + 1.
	\end{equation}
	Then with probability at least $1 - e^{-\probsuccess^2/2}$, TV minimization \eqref{eq:intro:tv-1} with noiseless input $\y = \A \grtr \in \R^m$ recovers $\grtr$ exactly.
\end{proposition}
Remarkably, the sample-size bound \eqref{eq:results:exact_geo:meas} of Proposition~\ref{prop:results:exact_geo} is essentially sharp, since recovery would fail with high probability if 
\begin{equation}\label{eq:results:exact_geo:lower}
	m\leq \effdim[\conic]{\descset{\lnorm{\TV(\cdot)}[1], \grtr}} - C\cdot\meanwidth[\conic]{\descset{\lnorm{\TV(\cdot)}[1], \grtr}},
\end{equation}
where $C>0$ is a universal constant; see~\cite[Rmk.~3.4]{tropp2014convex}. In other words, TV minimization~\eqref{eq:intro:tv-1} exhibits a \emph{phase transition} at $m \approx \effdim[\conic]{\descset{\lnorm{\TV(\cdot)}[1], \grtr}}$.

On the other hand, the definition of the conic mean with is quite implicit and therefore only provides an uninformative description of the required number of measurements. Except for a few simple cases (e.g., for $f(\cdot) = \lnorm{\cdot}[1]$ in standard compressed sensing), it is notoriously hard to find more informative bounds for $\meanwidth[\conic]{\descset{f, \grtr}}$ that are still sufficiently accurate. In this context, the work of \cite{genzel2017cosparsity} establishes a sophisticated, non-asymptotic upper bound for analysis-based priors of the form $\lnorm{\aop(\cdot)}[1]$, where $\aop \in \R^{N \times n}$ is intended to be a \emph{redundant} transformation with $N > n$. Although this approach even applies to the choice $\aop = \TV$, the resulting sampling rate does not break the $\sqrt{\s n}$-bottleneck of Theorem~\ref{thm:cai}. With that in mind, the technical centerpiece of this work is the following theorem, which states a surprisingly simple upper bound for $\meanwidth[\conic]{\descset{\lnorm{\TV(\cdot)}[1], \grtr}}$. Recall that a combination of this result with Proposition~\ref{prop:results:exact_geo} immediately yields our main result, Theorem~\ref{thm:results:exact}.
\begin{theorem}\label{thm:results:mwbound}
	Let $\grtr \in \R^n$ be a $\SC$-separated signal with $\s>0$ jump discontinuities. Assuming that $\SC \geq 8\s / n$, we have that
	\begin{equation}\label{eq:results:mwbound}
	\effdim[\conic]{\descset{\lnorm{\TV(\cdot)}[1], \grtr}} \lesssim \frac{1}{\SC} \cdot \s\log^2(n).
	\end{equation}
\end{theorem}
While Subsection~\ref{subsec:proofs:mw} is dedicated to the formal proof of the previous result, we use the remainder of the current subsection to give an overview of the main arguments.
The starting point for our proof is a now fairly standard polar bound that can be traced back to \citeauthor{stojnic2009gordon}~\cite{stojnic2009gordon}; see also~\cite{chandrasekaran2012geometry,amelunxen2014edge,tropp2014convex}. It allows us to show that
\begin{equation}
\label{eq:main:polar}
	\effdim[\conic]{\descset{\lnorm{\TV(\cdot)}[1], \grtr}} \leq \inf_{\tau > 0} \mean{}_{\gaussian} [ \inf_{\dv \in \F} \lnorm{\gaussian - \tau \H \TV^\T \dv}^2 ],
\end{equation}
where $\gaussian \distributed \Normdistr{\vnull}{\I{n}}$, $\F  \coloneqq  \{ \dv \in \R^N \suchthat \dv_{\ssupp} = [\sign(\TV \grtr)]_{\ssupp}, \lnorm{\dv_{\setcompl\ssupp}}[\infty] \leq 1, \ssupp = \supp (\TV\grtr) \}$, and $\H \in \R^{n \times n}$ is an arbitrary orthogonal matrix (see Step~\hyperref[eq:proof:step1]{1} in Subsection~\ref{subsec:proofs:mw}).\footnote{From now on, we omit `Subsection~\ref{subsec:proofs:mw}' when referring to Step~\hyperref[eq:proof:step1]{1}--\hyperref[eq:proof:step5]{5} therein.}

The inner optimization problem on the right-hand side of \eqref{eq:main:polar} forms a box-con\-strained least-squares problem, which does not possess a closed form solution. In fact, the most challenging part of our proof is the \emph{explicit} construction of a dual vector $\bar{\dv} \in \F$ that yields a good ansatz to this problem. 
In the standard form, i.e., when $\H = \I{n}$, this appears to be a hopeless endeavor due to the coupled dependencies of the variables that are caused by $\TV^\T$. 
Hence, a key step of our approach is to come up with an appropriate orthogonal transform $\H$ that allows us to ``decouple'' these dependencies. Interestingly, the Haar wavelet transform turns out to be a good candidate for this task (see Step~\hyperref[eq:proof:step2a]{2}); cf.~\cite{Needell2013,Needell2013b}, where Haar wavelets play a crucial role as well, yet in a very different manner. 
The simplification due to the Haar matrix is visualized in Figure~\ref{fig:roadmap:1}, where we have plotted the resulting matrix $\H\TV^\T \in \R^{n \times N}$. Along the partition of $\H\TV^\T$ into multiple scales, it is possible to loosely identify a binary tree-like structure whose root vertex is formed by the center of the second row (note that the first row consists of zeros only). Indeed, when permuting the columns of $\H\TV^\T$ according to the induced dyadic ordering, we obtain a lower triangular matrix, such as shown in Figure~\ref{fig:roadmap:2}. This particularly simple structure allows us to significantly facilitate the inner optimization problem on the right-hand side of \eqref{eq:main:polar}. By exploiting the induced multilevel structure, we are eventually able to come up with a meaningful choice of $\bar{\dv} \in \F$ (see Step~\hyperref[eq:proof:step3a]{3(a)}), which leads to a fairly general upper bound for the conic mean width (see Step~\hyperref[eq:proof:step3b]{3(b)}).

Unfortunately, the argument just outlined is only valid for ideal signals with equidistantly distributed jump discontinuities. Only in this case, the gradient support is consistent with the binary tree-like structure described above, which is in fact a crucial feature, since every $\dv \in \F$ is fixed on the gradient support. Hence, another key component of our proof is the adaption to an arbitrary singularity pattern (see Step~\hyperref[eq:proof:step2a]{2}). To this end, we take a signal-dependent point of view by first representing the gradient support of $\grtr$ as a binary tree, which is then completed with the remaining off-support elements (see Step~\hyperref[eq:proof:step2a]{2(a)}). Note that the resulting tree is not necessarily \emph{perfect}, i.e., not all interior vertices have two children or the leaves might be at different levels.
Based on this construction, we then design an appropriate signal-dependent transform $\H$, which is known as \emph{non-dyadic Haar matrix}~\cite{gupta_non-dyadic_2010} (see Step~\hyperref[eq:proof:step2b]{2(b)}). 

Finally, in Step~\hyperref[eq:proof:step4]{4}, the balance of the constructed binary tree is related to the separation constant~$\SC$ for the underlying signal $\grtr$. While this step requires a certain technical effort, it allows for a simple conclusion: the closer the jumps of $\grtr$ are to an equidistant pattern, the more balanced is the binary tree. A combination of Step~\hyperref[eq:proof:step3b]{3(b)} and Step~\hyperref[eq:proof:step4]{4} eventually leads to the desired mean-width bound of Theorem~\ref{thm:results:mwbound} (see Step~\hyperref[eq:proof:step5]{5}).

\begin{figure}
	\centering
	\begin{subfigure}[t]{0.45\textwidth}
		\centering
		\includegraphics[width=1\textwidth]{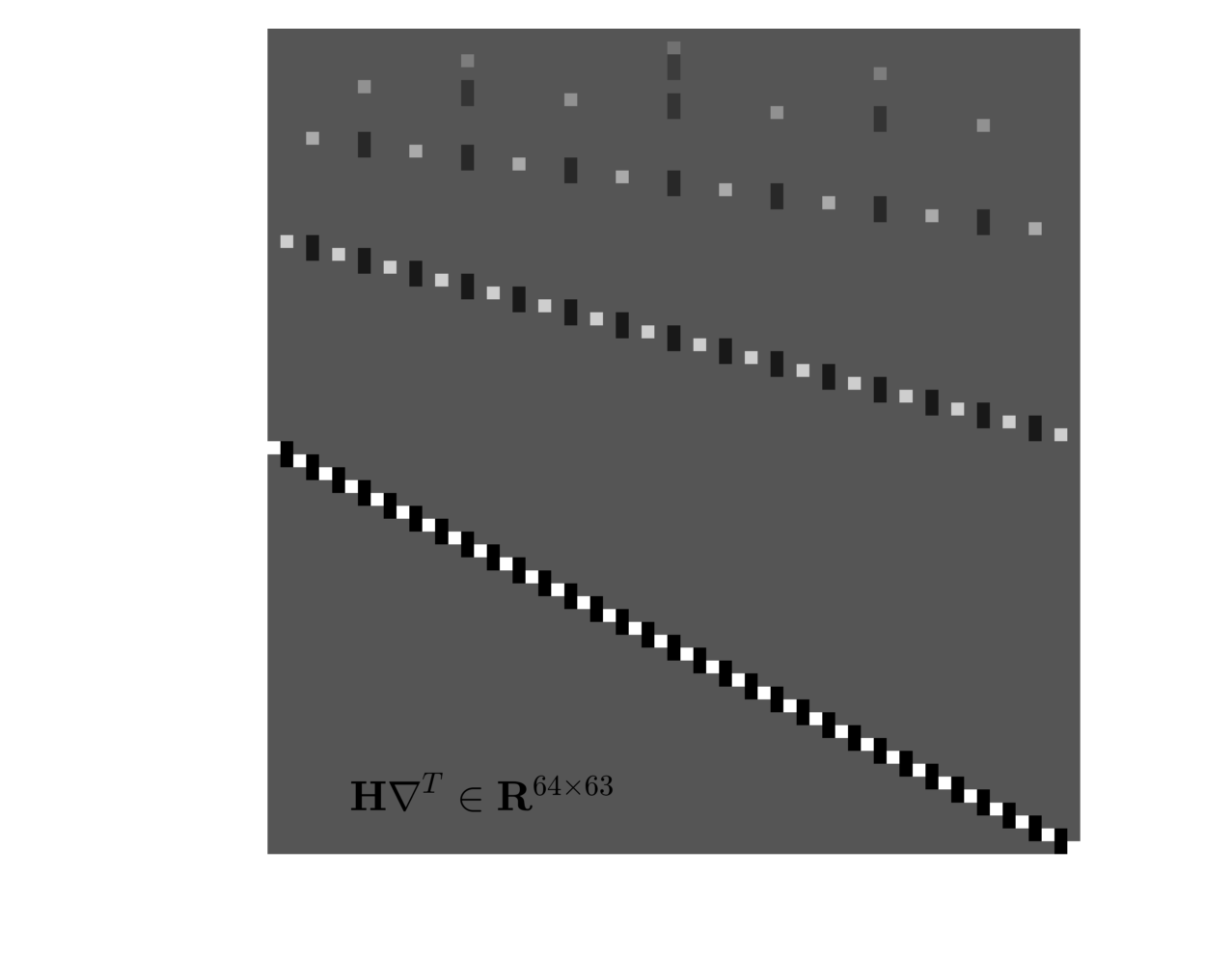}
		\caption{}
		\label{fig:roadmap:1}
	\end{subfigure}%
	\qquad
	\begin{subfigure}[t]{0.45\textwidth}
		\centering
		\includegraphics[width=1\textwidth]{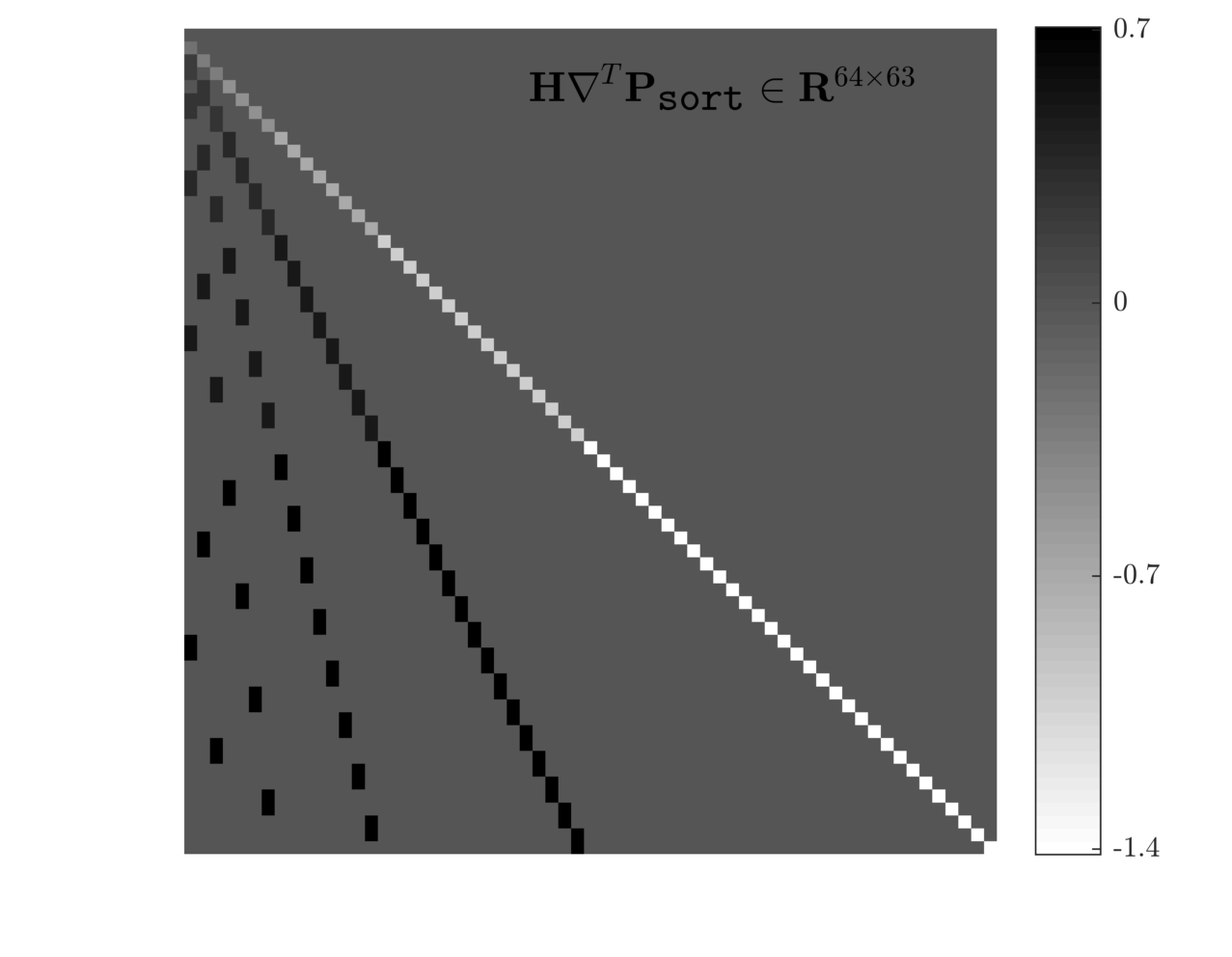}
		\caption{}
		\label{fig:roadmap:2}
	\end{subfigure}%
	\caption{\textbf{Visualization of the transformation matrix $\H\TV^\T$.} Subfigure~\subref{fig:roadmap:1} shows the matrix $\H\TV^\T \in \R^{n \times N}$ in dimension $n=64$, where $\H  \in \R^{n \times n}$ is defined as a Haar wavelet basis with $6$ decomposition levels. In Subfigure~\subref{fig:roadmap:2}, the columns of $\H\TV^\T$ are permuted appropriately.}
	\label{fig:roadmap}
\end{figure}

\subsection{TV Minimization for Stable and Robust Recovery}
\label{subsec:results:stable}

The mean-width bound for $\s$-gradient-sparse signals in Theorem~\ref{thm:results:mwbound} allows for a natural extension to stable and robust recovery.  We first briefly describe these two phenomena and then derive a generalization of Theorem~\ref{thm:results:exact} that takes them into account.

\emph{Robustness} refers to the observation that signal recovery is not too sensitive to measurement noise. In this work, we consider the (standard) setup of \emph{adversarial noise}, i.e., we assume that the measurements are given by
\begin{equation}\label{eq:results:measnoisy}
\y = \A \grtr + \noise,
\end{equation}
where $\noise \in \R^m$ models (possibly deterministic) noise, satisfying $\lnorm{\noise} \leq \noiseparam$ for some $\noiseparam \geq 0$.
In order to keep record of such perturbations, the vanilla TV minimization problem \eqref{eq:intro:tv-1} is adapted as follows:
\begin{equation}\label{eq:results:stable:tvdenois}\tag{$\text{TV-1$_\noiseparam$}$}
	\min_{\x \in \R^n} \lnorm{\TV\x}[1] \quad \text{subject to \quad $\lnorm{\A \x - \y} \leq \noiseparam$.}
\end{equation}
We note that the $\l{2}$-constraint in \eqref{eq:results:stable:tvdenois} ensures that the ground truth signal $\grtr$ remains a feasible point of the convex program.

\emph{Stability} addresses the recovery of signals that are only approximately gradient-sparse, i.e., $\TV\grtr$ possesses only a few dominant coefficients but could have many other small coefficients. In such a case, there is clearly no hope for a perfect recovery of $\grtr$ from significantly undersampled measurements. This is reflected by the fact that the assumptions of Theorem~\ref{thm:results:exact} only depend on the number of non-zero coefficients of $\TV\grtr$ but not on their magnitude. In particular, if $\TV\grtr$ is too densely populated, then \eqref{eq:results:exact:meas} could render a vacuous bound, or even worse, the separation condition \eqref{eq:results:exact:n} might be violated.
However, if $\grtr$ is close to a vector $\grtrsparse$ that is exactly gradient-sparse, it can be expected that \eqref{eq:results:stable:tvdenois}, or \eqref{eq:intro:tv-1} in the noiseless case, is stable under such small model inaccuracies and approximate recovery still succeeds with very few measurements.
The following guarantee for stable and robust recovery makes this claim precise. It is an instance of a more general approach to stable recovery, which was developed in \cite[Subsec.~6.1]{genzel2017cosparsity}.
\begin{proposition}[\protect{\cite[Cor.~6.6]{genzel2017cosparsity}}]\label{prop:results:stable:general}
	Let $\grtr \in \R^n$ be an arbitrary signal vector. Moreover, fix a vector $\grtrsparse \in \R^n$ with $\lnorm{\TV\grtr}[1] = \lnorm{\TV\grtrsparse}[1]$. For $\probsuccess > 0$ and $R > 0$, we assume that $\A \in \R^{m \times n}$ is a standard Gaussian random matrix with
	\begin{equation}\label{eq:results:stable:general:meas}
	m > m_0 \coloneqq \Big(\tfrac{R + 1}{R} \cdot \big[\meanwidth[\conic]{\descset{\lnorm{\TV(\cdot)}[1], \grtrsparse}} + 1\big] + \probsuccess \Big)^2 + 1.
	\end{equation}
	Then with probability at least $1 - e^{-\probsuccess^2/2}$, every minimizer $\solu \in \R^n$ of \eqref{eq:results:stable:tvdenois} with noisy input \eqref{eq:results:measnoisy} satisfies
	\begin{equation}\label{eq:results:stable:general:bound}
	\lnorm{\solu - \grtr} \leq R \lnorm{\grtr - \grtrsparse} + \frac{2\noiseparam}{\pospart{\sqrt{m-1} - \sqrt{m_0-1}}}.
	\end{equation}
	In the noiseless case, i.e., $\noiseparam = 0$, we particularly have that $\lnorm{\solu - \grtr} \leq R \lnorm{\grtr - \grtrsparse}$.
\end{proposition}
The previous result extends Proposition~\ref{prop:results:exact_geo} by the following intuitive trade-off concerning stability: instead of evaluating the conic mean width at the actual signal vector $\grtr$, we rather evaluate it at a well-chosen \emph{surrogate vector} $\grtrsparse \in \R^n$. If $\grtrsparse$ is exactly gradient-sparse, this can lead to a significantly weaker requirement on the number of measurements $m$ in \eqref{eq:results:stable:general:meas}. The price to pay is an additional error term in the bound of \eqref{eq:results:stable:general:bound}, which scales with the Euclidean distance between $\grtr$ and $\grtrsparse$. In other words, Proposition~\ref{prop:results:stable:general} enables a ``barter trade'' between the number of required measurements and the desired recovery accuracy; in this context, $R > 0$ serves as an additional oversampling parameter, which allows us to further balance this trade-off.
Furthermore, regarding robustness to measurement noise, we note that the dependence on the noise parameter in \eqref{eq:results:stable:general:bound} is standard, e.g., see~\cite[Cor.~3.5]{tropp2014convex}.

Nevertheless, the statement of Proposition~\ref{prop:results:stable:general} still remains rather uninformative as long as the surrogate vector $\grtrsparse$ is left unspecified.
According to our bound for the conic mean width in Theorem~\ref{thm:results:mwbound}, it is a natural strategy to select $\grtrsparse$ as an exactly gradient-sparse vector that is close to $\grtr$.
Indeed, the following stable and robust recovery guarantee is an application of Proposition~\ref{prop:results:stable:general}, stating an error bound that explicitly depends on the $\s$ most dominant jump discontinuities of $\grtr$.
As the definition of the corresponding $\grtrsparse$ involves some additional technicalities, we refer to the proof in Subsection~\ref{subsec:proofs:stable} for more details.
\begin{theorem}[Stable and robust recovery via TV minimization]\label{thm:results:stable}
	Let $\grtr \in \R^n$ be an arbitrary signal vector. Moreover, fix a gradient-sparsity level $\s > 0$ and assume that $\cardinality{\supp(\TV\grtr)} \geq \s$. We introduce the following notation:
	\begin{thmproperties}
		\item
		Let $\ssupp \subset [N] = [n-1]$ be a (possibly non-unique) index set with $\cardinality{\ssupp} = \s$ such that $[\TV\grtr]_{\ssupp} \in \R^\s$ contains the $\s$ largest entries of $\TV\grtr$ in magnitude, i.e., it corresponds to a \emph{best $\s$-term approximation} to~$\TV\grtr$ (with respect to the $\l{1}$-norm).
		\item 
		Let $\proj{\ssupp} \in \R^{N \times N}$ and $\proj{\setcompl\ssupp} = \I{N} - \proj{\ssupp} \in \R^{N \times N}$ be the coordinate projections onto $\ssupp$ and $\setcompl\ssupp$, respectively.
		\item 
		Let $\psinv{\TV} \in \R^{n \times N}$ be the pseudo-inverse of $\TV$, which satisfies $\TV \psinv{\TV} = \I{N}$ and $\psinv{\TV} \TV = \I{n} - \tfrac{1}{n}\1_n \1_n^\T$.
		\item 
		Let $\SC > 0$ be a separation constant for any signal vector with gradient support $\ssupp$.\footnote{Note that the definition of the separation constant (see Definition~\ref{def:results:msc}) only depends on the gradient support of a signal vector but not on its actual coefficients.}
	\end{thmproperties}
	We assume that $\SC \geq 8\s / n$, and for $\probsuccess > 0$, let $\A \in \R^{m \times n}$ be a standard Gaussian random matrix with
	\begin{equation}\label{eq:results:stable:meas}
	m \gtrsim \frac{1}{\SC} \cdot \s\log^2(n) + \probsuccess^2.
	\end{equation}
	Then with probability at least $1 - e^{-\probsuccess^2/2}$, every minimizer $\solu \in \R^n$ of \eqref{eq:results:stable:tvdenois} with noisy input \eqref{eq:results:measnoisy} satisfies
	\begin{equation}\label{eq:results:stable:bound}
	\lnorm{\solu - \grtr} \lesssim \tau(\grtr) \lnorm[\big]{\psinv{\TV} \proj{\setcompl\ssupp}\TV\grtr} + \big(\tau(\grtr) - 1 \big) \lnorm[\big]{\grtr - \sp{\tfrac{1}{n}\1_n}{\grtr} \1_n} + \frac{\noiseparam}{\sqrt{m}},
	\end{equation}
	where $\tau(\grtr) \coloneqq \lnorm{\TV\grtr}[1] / \lnorm{\proj{\ssupp}\TV\grtr}[1] \geq 1$.
\end{theorem}
The significance of the error estimate in \eqref{eq:results:stable:bound} depends on the ratio $\tau(\grtr)$, which measures how well $\grtr$ can be ``compressed'' by an $\s$-gradient-sparse signal. 
Indeed, the more the gradient coefficients $\TV\grtr$ concentrate on $\ssupp$, the closer $\tau(\grtr)$ is to $1$. 

Despite the dependence on $\tau(\grtr)$, the first error term in \eqref{eq:results:stable:bound} is dominated by $\lnorm{\psinv{\TV} \proj{\setcompl\ssupp}\TV\grtr}$, which captures the size of the remaining gradient coefficients $\proj{\setcompl\ssupp}\TV\grtr$. 
The second error term, on the other hand, depends on $\lnorm{\grtr - \sp{\tfrac{1}{n}\1_n}{\grtr} \1_n}$, which can be seen as the energy of the centered version of $\grtr$.\footnote{The inner product $\sp{\tfrac{1}{n}\1_n}{\grtr}$ simply computes the arithmetic mean of $\grtr$, which is then subtracted entrywise.} While this quantity is independent of $\ssupp$, its contribution to the total recovery error gets smaller as $\tau(\grtr)$ approaches $1$. Nevertheless, we suspect that the presence of this term is an artifact of our proof, which is caused by the technical assumption $\lnorm{\TV\grtr}[1] = \lnorm{\TV\grtrsparse}[1]$ in Proposition~\ref{prop:results:stable:general}.

Although the bound of Theorem~\ref{thm:results:stable} is likely to be non-tight, it exhibits the typical features of stable recovery:
the size of the first two error terms in \eqref{eq:results:stable:bound} depends continuously (but non-linearly) on the entries of $\grtr$, and if $\grtr$ is exactly $\s$-gradient-sparse, they both vanish.

We close our discussion with some additional remarks on Theorem~\ref{thm:results:stable}:
\begin{remark}
	\begin{rmklist}
		\item \label{rmk:results:stable:optimal}
		Theorem~\ref{thm:results:stable} relies on a natural choice of $\grtrsparse$ in Proposition~\ref{prop:results:stable:general}, which leads to explicit error terms in \eqref{eq:results:stable:bound}.
		However, one could easily obtain a tighter, but less informative, error bound from Proposition~\ref{prop:results:stable:general} by optimizing the right-hand side of \eqref{eq:results:stable:general:bound} over $\grtrsparse$. Indeed, for a fixed separation constant $\SC > 0$, the following error bound holds true under the same hypotheses as in Theorem~\ref{thm:results:stable}:
		\begin{equation}
		\lnorm{\solu - \grtr} \lesssim \inf \Big\{ \lnorm{\grtr - \grtrsparse} \suchthat \substack{\text{$\grtrsparse \in \R^n$ $\s$-gradient-sparse and} \\ \text{$\SC$-separated with $\lnorm{\TV\grtr}[1] = \lnorm{\TV\grtrsparse}[1]$} } \Big\} + \frac{\noiseparam}{\sqrt{m}}.
		\end{equation}
		Hence, the recovery accuracy is essentially determined by the error of the best $\s$-gradient-sparse approximation to $\grtr$ (with respect to the $\l{2}$-norm) with some additional constraints.
		
		\item 
		There is an additional trade-off which has not been taken into account in Theorem~\ref{thm:results:stable}, but which could be of importance when choosing the underlying surrogate vector $\grtrsparse$: In Theorem~\ref{thm:results:stable}, the set $\ssupp$ is selected according to the $\s$ largest entries of $\TV \grtr$ in magnitude. However, in view of the sampling rate promoted by \eqref{eq:results:stable:meas}, it might be beneficial to select $\ssupp$ in such a way that the resulting separation constant $\SC$ gets enlarged. By accepting a possibly less accurate reconstruction, this can lead to a significant decrease in the required number of measurements. Such a refinement might be of particular importance in an asymptotic-order regime, e.g., when $\grtr$ arises from the discretization of a function that is more complicated than a piecewise constant signal.
		\item \label{rmk:results:stable:literature}
		We wish to point out that the above approach to stable recovery is rather different from common strategies in the compressed sensing literature. While we rely on the \emph{Euclidean geometry} of the descent set $\descset{\lnorm{\TV(\cdot)}[1], \grtr}$ around $\grtr$, most standard stability results build upon \emph{uniform} recovery conditions, such as variants of the stable and robust nullspace property, e.g., see~\cite{Needell2013,Needell2013b,cai_guarantees_2015}.
		These findings suggest that, instead of \eqref{eq:results:stable:bound}, we may expect a bound of the form
		\begin{equation}\label{eq:results:stable:literature:bound}
		\lnorm{\solu - \grtr} \lesssim \frac{\lnorm{\TV\grtr - [\TV\grtr]_{\ssupp}}[1]}{\sqrt{s}} + \frac{\noiseparam}{\sqrt{m}}.
		\end{equation}
		Such a bound appears to be simpler than \eqref{eq:results:stable:bound} and it is consistent with the standard compressibility theory (for sparsity in orthonormal bases), e.g., see~\cite[Thm.~4.22]{foucart2013cs}.
		However, it is not entirely clear to us whether such a statement can be achieved within the signal-dependent setup of this paper, and if so, we suspect that different proof techniques would be required.
		We also point out that the worst-case (uniform) analysis of~\cite{cai_guarantees_2015} requires significantly more Gaussian measurements than \eqref{eq:results:stable:meas} in order to achieve a stability guarantee that is similar to \eqref{eq:results:stable:literature:bound}. \qedhere
	\end{rmklist}\label{rmk:results:stable}
\end{remark}

\subsection{Limitations and Possible Refinements}
\label{subsec:results:limitations}

In this part, we discuss some limitations of our recovery results that are particularly related to the concept of $\SC$-separation.
According to the statement of Corollary~\ref{cor:results:exact:cont}, the required sampling rate scales logarithmically with $n$ when considering natural (discretized) signals.
While such a behavior is desirable and significantly better than what can be achieved in the worst case (cf.~Theorem~\ref{thm:cai}), the dependence on the gradient-sparsity $\s$ can still be suboptimal in certain cases.
In fact, there exist classes of gradient-sparse signals whose separation constants cannot be controlled independently of the number of jump discontinuities.
Let us consider two illustrative examples:
\begin{example}
	\begin{rmklist}
		\item\label{ex:results:limitations:rndjumps}
		\textbf{Random jumps.} Let $\contgrtr \colon \intvopcl{0}{1} \to \R$ be a piecewise continuous function defined according to \eqref{eq:results:contmsc:signal} with $\s$ jump discontinuities that are selected independently at random from the uniform distribution on $\intvop{0}{1}$.
		One can show that with high probability, the maximal separation constant $\cont\SC > 0$ for $\contgrtr$ satisfies $\cont\SC \asymp \s^{-1}$, so that the sampling-rate bound \eqref{eq:results:exact:cont:meas} turns into
		\begin{equation}\label{eq:results:limitations:rndjumps:meas}
		m \gtrsim \s^2\log^2(n) + \probsuccess^2.
		\end{equation}
		\item\label{ex:results:limitations:expjumps}
		\textbf{Densifying jumps.} It is possible to construct even worse (deterministic) examples than in part~\ref{ex:results:limitations:rndjumps}. For instance, let $\contgrtr \colon \intvopcl{0}{1} \to \R$ be a piecewise continuous function defined according to \eqref{eq:results:contmsc:signal} such that
		\begin{equation}
		\cont\jump_i = 1 - \frac{1}{2^i}, \quad i = 1, \dots, \s.
		\end{equation}
		Then the maximal separation constant for $\contgrtr$ is $\cont\SC = 2^{-\s}(\s+1)$, so that the sampling-rate bound \eqref{eq:results:exact:cont:meas} turns into
		\begin{equation}\label{eq:results:limitations:expjumps:meas}
		m \gtrsim 2^\s\log^2(n) + \probsuccess^2.
		\end{equation}
		This example of ``exponentially densifying'' jumps pushes the applicability of Corollary~\ref{cor:results:exact:cont} to the limits, in the sense that the required resolution level scales exponentially with $\s$, i.e., $n \geq 16\s/\cont\SC \gtrsim 2^\s$. \qedhere
	\end{rmklist}\label{ex:results:limitations}
\end{example}
We emphasize that these two examples are still instances of natural signals, which do not correspond to the worst-case scenarios for discrete signals (see Subsection~\ref{subsec:results:signals}); in particular, the sampling rates in \eqref{eq:results:limitations:rndjumps:meas} and \eqref{eq:results:limitations:expjumps:meas} scale logarithmically with $n$.
However, we suspect that the (asymptotic) dependence on $\s$ is suboptimal in both cases; in fact, it is not clear to us whether a simple asymptotic-order bound is meaningful in the pathological situation of densifying jumps.
More generally, Example~\ref{ex:results:limitations} demonstrates the limitations of the $\SC$-separation property: The size of the separation constant is determined by the distance of the two closest jump discontinuities. Thus, it is very sensitive to outliers and does not capture the ``average'' distribution of the jump discontinuities, which might be much more benign in the above cases.
In other words, $\SC$-separation is a \emph{local} feature of a signal and does not reflect the \emph{global} structure of its gradient support.

On the other hand, our proof strategy in Subsection~\ref{subsec:proofs:mw} leaves certain space for refinements that could allow for a sampling-rate bound that reflects the global structure of the gradient support.
To be more specific, we observe that the bounds for $\Lambda^{L_0}$ and $\Lambda^{\natural}$ in \eqref{eq:proofs:mw:pivoterr_sums} can be easily improved as follows:
\begin{align}
\sum_{(\ell,i) \in \Lambda^{L_0}} \e{i}{\ell}
&\leq \cardinality{\Lambda^{L_0}} + (2 \tau c_{L_0})^2 \sum_{(\ell,i) \in \Lambda^{L_0}} (\b{i}{\ell})^2, \\*
\sum_{(\ell,i) \in \Lambda^\natural} \e{i}{\ell}
&\leq \cardinality{\Lambda^\natural} + (\tau  c_{L_0})^2 \sum_{(\ell,i) \in \Lambda^\natural} (\b{i}{\ell})^2. \label{eq:results:limitations:pivoterr_sums}
\end{align}
Following this path, it remains to control the ``average'' balance parameters $\sum_{(\ell,i) \in \Lambda^{L_0}} (\b{i}{\ell})^2$ and $\sum_{(\ell,i) \in \Lambda^\natural} (\b{i}{\ell})^2$ instead of $\b{\max}{\Lambda^{L_0}}$ and $\b{\max}{\Lambda^\natural}$, respectively.
This would particularly require a substantial adaption of Step~\hyperref[eq:proof:step4]{4} in Subsection~\ref{subsec:proofs:mw}, coming along with different technical difficulties.
A precise analysis would result in various conditions that depend on the average separation of the jump discontinuities of $\grtr$, rather than the distance of the two closest jumps.
In a prototypical form, such a condition would read as follows:
\begin{equation}\label{eq:results:limitations:msc_refined}
\tfrac{1}{\s+1} \sum_{i = 1}^{\s+1} \frac{n}{\abs{\jump_i - \jump_{i-1}}} \leq \frac{\s+1}{\tilde\SC},
\end{equation}
where $\tilde\SC > 0$ plays the role of a separation parameter, similarly to the separation constant $\SC$ in Definition~\ref{def:results:msc}.
In particular, it is not hard to see that the gradient support of a $\SC$-separated signal satisfies \eqref{eq:results:limitations:msc_refined} with $\tilde\SC \geq \SC$, so that this condition can be seen as a relaxation of the~$\SC$-separation.
However, in order to achieve (exact) recovery with $m \gtrsim \tilde{\SC}^{-1} \cdot \s \cdot \PolyLog(n)$ measurements, we expect that \eqref{eq:results:limitations:msc_refined} needs to be complemented by more technical assumptions, based on our multi-level tree representation of the gradient support.
Working out the details in that respect is by far not straightforward and would go beyond the scope of this paper. 


\revision{
\begin{remark}
 In the context of signal denoising, refinements closely related to~\eqref{eq:results:limitations:msc_refined} have been shown by Guntuboyina et al.~\cite{guntuboyina20}. Their work provides comprehensive risk bounds for (higher order) TV-denoising (also referred to as trend filtering), improving previous results in the literature, e.g., see~\cite{dalalyan17}. It turns out that~\eqref{eq:results:limitations:msc_refined} can be additionally sharpened by summing only over constant pieces of $\grtr$, for which the sign of $\TV \grtr$ changes. 
\end{remark}
}

\section{Conclusion and Outlook}
\label{sec:conclusion}

In this work, we have addressed the task of recovering piecewise constant signals from compressed Gaussian measurements via TV minimization~\eqref{eq:intro:tv-1}. From a conceptual point of view, our key achievement can be summarized as follows: 
Already in the ``simple'' case of TV minimization in 1D, the usage of a \emph{signal-dependent} complexity measure becomes indispensable.
Indeed, according to Theorem~\ref{thm:cai}, uniform recovery of all $\s$-gradient-sparse signals in $\R^n$ necessarily leads to the so-called $\sqrt{\s n}$-bottleneck---a pessimistic sampling rate that corresponds to the worst-case behavior on this signal class. 
The non-uniform approach carried out in Theorem~\ref{thm:results:exact} allows us to break this complexity barrier, showing that recovery of a $\SC$-separated signal is possible with $m\gtrsim \SC^{-1} \cdot \s \log^2(n)$ measurements. 

This sampling-rate bound is commonly regarded as an \emph{asymptotic-order} statement: such bounds are still non-asymptotic in the sense that the model parameters $m$, $n$, and $\s$ are finite; however, they may involve unknown constants and do not yield a literal description of the phase transition of \eqref{eq:intro:tv-1} (cf.~Proposition~\ref{prop:results:exact_geo} and \eqref{eq:results:exact_geo:lower}). Hence, the recovery guarantees presented in this work rather obtain their meaning by revealing the interplay of the model parameters in high dimensions.
With that in mind, recall that our main result of Theorem~\ref{thm:results:exact} becomes especially relevant when the separation constant $\SC$ is only mildly depending on $\s$ and $n$. More formally, for a fixed value of $\SC \in \intvopcl{0}{1}$, one might consider the following class of signals:
\begin{equation}
\gssignals{\SC} \coloneqq \big\{ \grtr \in \R^n \suchthat n \in \N, \text{$\grtr$ is $\SC$-separated with $\cardinality{\supp(\TV\grtr)} = \s > 0$ and $\SC \geq 8\s/n$} \big\}.
\end{equation}
Then, Theorem~\ref{thm:results:exact} states that every signal $\grtr \in \gssignals{\SC}$ can be recovered with high probability as long as $m \geq C \cdot \s \log^2(n)$, where the constant $C>0$ is independent of all other variable parameters. Consequently, $\gssignals{\SC}$ forms a subclass of those gradient-sparse signals that would enjoy the desired sampling rate of $m \gtrsim \s \cdot \PolyLog (n)$; in particular, it encompasses discretizations of piecewise constant functions with sufficiently well-separated jumps. At the same time, pathological examples like dense jumps or the densifying jumps of Example~\ref{ex:results:limitations}\ref{ex:results:limitations:expjumps} are excluded. 


With the set $\gssignals{\SC}$ at hand, we can now also make the distinction between signal-dependent and non-uniform recovery more precise: \emph{Signal dependence} refers to the fact that there exist subclasses of gradient-sparse vectors that differ significantly in their associated sampling rate, e.g., dense-jump signals versus equidistant-jump signals (in $\gssignals{1}$). \emph{Non-uniformity} means that successful recovery of a fixed signal is highly probable for a random draw of the measurement matrix $\A$. In contrast, uniform guarantees are concerned with the probability that an entire class of signals can be recovered for a random draw of $\A$. Hence, a non-uniform result allows (at least in principle) that the small exceptional set of measurement matrices for which recovery fails may depend on the considered signal.
In this light, our work provides rigorous evidence that signal dependence is an actual phenomenon for TV minimization in 1D. On the other hand, although quite natural, the non-uniformity of our result might just be an artifact of our general proof strategy. We suspect that a more advanced proof could lead to a guarantee similar to Theorem~\ref{thm:results:exact} that holds \emph{uniformly} over an appropriate subclass of gradient-sparse vectors, such as $\gssignals{\SC}$. Apart from that, we emphasize that a signal-dependent theory is not required for ordinary sparse recovery from Gaussian measurements: although the optimal number of measurements for uniform and non-uniform recovery differ slightly, they both take the form $m \gtrsim \s\cdot \PolyLog (n)$ for $s$-sparse vectors, cf.~\cite[Notes~of~Chap.~9]{foucart2013cs}.

In our view, the signal-dependent bounds established in this work are only the ``tip of the iceberg'' of a much more general phenomenon in compressed sensing. 
Indeed, the findings of~\cite{genzel2017cosparsity} indicate that a similar observation applies to $\l{1}$-analysis minimization in general.
Compared to~\cite{genzel2017cosparsity}, we do not aim at non-asymptotic descriptions of the phase transition and restrict ourselves to the seemingly simple case of $\TV$ as analysis operator; but in fact, these simplifications allow for an informative asymptotic-order bound that reflects the intuitive signal-dependent behavior of~\eqref{eq:intro:tv-1}.
In that sense, the present work provides a qualitative version of the conclusions of~\cite{genzel2017cosparsity}.
More generally, we believe that for a thorough understanding of many popular recovery methods, it is crucial to go beyond generic signal models such as plain sparsity. 
It appears beneficial to identify problem-dependent signal classes, whose elements satisfy specific structural properties that are associated with the success or failure of a recovery method under study. While our work and~\cite{genzel2017cosparsity} show that this approach is fruitful for analysis-based compressed sensing, we expect that it is also essential for the synthesis formulation~\cite{boyer2020synthesis} as well as more recent approaches, such as denoising-based regularization~\cite{bouman2013,oymak2016sharpmse}.
From a technical perspective, the non-uniform geometric framework outlined in Subsection~\ref{subsec:results:mw} has proven very useful in this context, but certainly, there might be other effective strategies. Having said this, we conclude our discussion by posing the following question: Does the signal-dependent construction of the orthogonal transform $\H \in \R^{n \times n}$ in the polar bound~\eqref{eq:main:polar} constitute only a special case of a general methodology that enables estimates of the conic mean width of more complicated regularizers?

%

\section{Proofs}
\label{sec:proofs}

This part is dedicated to the proofs of the results from Section~\ref{sec:results}.

\subsection{Proof of Proposition~\ref{prop:results:discrcont}}
\label{subsec:proofs:discrcont}

	For every $i \in [\s]$, there exists $j_i \in [N]$ such that
	\begin{equation}\label{eq:proofs:discrcont:jumpdiscr}
	\frac{j_i}{n} \leq \cont\jump_i < \frac{j_i+1}{n}.
	\end{equation}
	The resulting indices satisfy $j_1 < \dots < j_\s$, due to the $\cont\SC$-separation of $\contgrtr$ and the assumption $\cont\SC \geq (\s+1)/n$; indeed, otherwise, there would exist $i \in [\s]$ such that
	\begin{equation}
	\abs{\cont\jump_{i} - \cont\jump_{i-1}} < \frac{1}{n} \leq \frac{\cont\SC}{\s+1}.
	\end{equation}
	Consequently, we have that $x_{j_i}^\ast \neq x_{j_i+1}^\ast$ for all $i \in [\s]$ and $x_{j}^\ast = x_{j+1}^\ast$ for all $j \in [N] \setminus \{j_1, \dots, j_\s\}$. Therefore, $\supp(\TV\grtr) = \{\jump_1, \dots, \jump_\s\} = \{j_1, \dots, j_\s\}$.
	Finally, with $j_0 \coloneqq 0$ and $j_{\s+1} \coloneqq n$, we obtain
	\begin{equation}
	\min_{i \in [\s+1]} \frac{\abs{\jump_i - \jump_{i-1}}}{n} = \min_{i \in [\s+1]} \frac{j_i - j_{i-1}}{n} \stackrel{\eqref{eq:proofs:discrcont:jumpdiscr}}{>} \cont\jump_i - \frac{1}{n} - \cont\jump_{i-1} \geq \frac{\cont\SC}{\s+1} - \frac{1}{n} = \frac{\cont\SC + \frac{\s+1}{n}}{\s+1},
	\end{equation}
	which implies that $\SC \coloneqq \cont\SC + \frac{\s+1}{n}$ is a valid separation constant for $\grtr$. \qed

\subsection{Proof of Theorem \ref{thm:results:mwbound}}
\label{subsec:proofs:mw}

Throughout the proof, we agree on the following assumptions and notational conventions:
\begin{itemize}
	\item
	Let $\grtr \in \R^n$ be a $\SC$-separated signal with $\s>0$ jump discontinuities. The gradient support of $\grtr$ is denoted by $\ssupp \coloneqq \supp(\TV\grtr)$, satisfying $\cardinality{\ssupp} = \s$. 
	Moreover, according to the hypothesis of Theorem~\ref{thm:results:mwbound}, we assume that $\SC \geq 8\s / n$.
	Following the visualization of Figure~\ref{fig:proofs:mw:face_edges}, we call the $n$ entries of $\grtr \in \R^n$ the \emph{nodes}, and the $N = n-1$ entries of $\TV \grtr \in \R^N$ the \emph{faces} of~$\grtr$.  
	This distinction will prove particularly useful when representing the gradient support of~$\grtr$ as a binary tree in Step~\hyperref[eq:proof:step2a]{2(a)}.
	\item 
	Let $\Supp \subset [N]$ be any superset of the gradient support $\ssupp$, i.e., $\ssupp \subset \Supp$, and set $\sext \coloneqq \cardinality{\Supp}$.
	The purpose of $\Supp$ is primarily of technical nature, but it will be handy for the construction of binary tree structures that are almost perfect (see Step~\hyperref[eq:proof:step2a]{2(a)} and Step~\hyperref[eq:proof:step4]{4}).
	Intuitively, one can interpret $\Supp$ as an artificially extended gradient support, where $\jump \in \Supp \setminus \ssupp$ corresponds to a ``ghost jump'' of height $0$ at the $\jump$-th node.
	The set $\Supp$ is left unspecified up to Step~\hyperref[eq:proof:step4]{4}, where it will be selected such that $\sext \asymp \s / \SC$. Until then, it might be helpful to simply assume that $\Supp$ is equal to $\ssupp$ (and therefore $\sext = \s$). Finally, we set $L_0  \coloneqq  \ceil{ \log_2(\sext+1) } \in \N$.
\end{itemize}

\revision{
As the proof of Theorem~\ref{thm:results:mwbound} turns out to be somewhat lengthy and intricate, we have decided to split it into several substeps; see the roadmap in Subsection~\ref{subsec:results:mw} for a non-technical overview. In particular, the construction of the signal-dependent tree in Step~\hyperref[eq:proof:step2a]{2(a)} requires some technical overhead that might obfuscate the general strategy. 
Hence, in each substep, we also consider the special case of a \emph{dyadic signal} as a recurring example. 
This will assist in conveying the intuitive proof idea, as the involved quantities simplify considerably.
A dyadic signal $\grtr \in \R^n$ has equidistant jump discontinuities with  $\s = 2^{L_0} - 1$ and $n = 2^L$ for some $L_0 \leq L \in \N$ (see Figure~\ref{fig:trees:1} for an example). In this case, the extended gradient support can be chosen as  $\Supp = \ssupp = \supp(\TV\grtr) = \{ i \cdot 2^{L - L_0} \suchthat i = 1, \dots, \s \}$. 
}

\subsubsection{Step 1: A Polar Bound for the Conic Mean Width}
\label{eq:proof:step1}

We begin with a well-known polarity argument, relating the conic mean width of an arbitrary convex function to the expected difference of a standard Gaussian random vector to its subdifferential:
\begin{proposition}[\protect{\cite[Prop.~4.1]{amelunxen2014edge}}] \label{prop:proofs:mw:polargeneral}
	Let $f \colon \R^n \to \R$ be a convex function and let $\grtr \in \R^n$. Then, the \emph{subdifferential} of $f$ at $\grtr$, given by 
	\begin{equation}
	\subd f(\grtr) \coloneqq \big\{ \v \in \R^n \suchthat f(\x) \geq f(\grtr) + \sp{\v}{ \x-\grtr} \text{ for all } \x \in \R^n \big\},
	\end{equation}
	is well-defined. If $\subd f(\grtr)$ is non-empty, compact, and does not contain the origin, we have that
	\begin{equation} \label{eq:proofs:mw:polargeneral:bound}
	\effdim[\conic]{\descset{f,\grtr}} \leq \inf_{\tau > 0} \mean{}_{\gaussian} [ \inf_{\v \in \subd f(\grtr)} \lnorm{\gaussian - \tau \v }^2 ],
	\end{equation}
	where $\gaussian \distributed \Normdistr{\vnull}{\I{n}}$.
\end{proposition}
We point out that the bound of \eqref{eq:proofs:mw:polargeneral:bound} can be loose in general.
But in our case of interest, where $f(\cdot) = \lnorm{\TV (\cdot)}[1]$, it has recently been shown that \eqref{eq:proofs:mw:polargeneral:bound} is tight up to an additive constant~\cite{zhang_precise_2016}.
This observation is particularly useful in conjunction with the following lemma, which provides an explicit expression for the subdifferential of $\lnorm{\TV (\cdot)}[1]$.
\begin{lemma}[\protect{\cite[Lem.~6.11]{genzel2017cosparsity}}]\label{lem:proofs:mw:subd}
	We have that
	\begin{equation}\label{eq:proofs:mw:subd}
	\subd{\lnorm{\TV (\cdot)}[1]}(\grtr) = \TV^\T \cdot \underbrace{\big( \sign(\TV\grtr)  + \{  \dv \in \R^N \suchthat \supp(\dv) \subset \ssuppc, \lnorm{\dv}[\infty] \leq 1 \} \big)}_{\eqqcolon \F}.
	\end{equation}
	In particular, $\emptyset \neq \subd{\lnorm{\TV (\cdot)}[1]}(\grtr) \subset \TV^\T \intvcl{-1}{1}^N$ and $\vnull \not\in \subd{\lnorm{\TV (\cdot)}[1]}(\grtr)$.
\end{lemma}
Since the extended gradient support $\Supp$ satisfies $\Supp \supset \ssupp$ and $[\sign(\TV\grtr)]_\jump = 0 \in \intvcl{-1}{1}$ for all $\jump \in \Supp\setminus \ssupp$, we observe that
\begin{equation} \label{eq:proofs:mw:feasibleset}
\Fext  \coloneqq  \{ \dv \in \R^N \suchthat \dv_\Supp = [\sign(\TV \grtr)]_\Supp, \lnorm{\dv_{\setcompl\Supp}}[\infty] \leq 1 \} \subset \F,
\end{equation}
and therefore $\TV^\T \Fext \subset \subd{\lnorm{\TV (\cdot)}[1]}(\grtr)$. A combination of this fact with Proposition~\ref{prop:proofs:mw:polargeneral} and Lemma~\ref{lem:proofs:mw:subd} leads to the following bound for the conic mean width, which forms the basis for the proof of Theorem~\ref{thm:results:mwbound}; note that the matrix $\H \in \R^{n \times n}$ allows for an orthogonal change of variables, which will play a crucial role in the subsequent proof steps.
\begin{lemma} \label{lem:proofs:mw:unitarypolar}
	Let $\H \in \R^{n \times n}$ be an orthogonal matrix. Then, we have
	\begin{equation} \label{eq:proofs:mw:unitarypolar}
	\effdim[\conic]{\descset{\lnorm{\TV(\cdot)}[1], \grtr}} \leq \inf_{\tau > 0} \mean{}_{\gaussian} [ \inf_{\dv \in \Fext} \lnorm{\gaussian - \tau \H \TV^\T \dv}^2 ],
	\end{equation}
	where $\gaussian \distributed \Normdistr{\vnull}{\I{n}}$ and $\Fext$ is given by \eqref{eq:proofs:mw:feasibleset}.
\end{lemma}
\begin{proof}
	From Proposition~\ref{prop:proofs:mw:polargeneral}, Lemma~\ref{lem:proofs:mw:subd}, and~\eqref{eq:proofs:mw:feasibleset}, it follows that 
	\begin{align}
	\effdim[\conic]{\descset{\lnorm{\TV(\cdot)}[1], \grtr}}
	&\leq \inf_{\tau > 0} \mean{}_{\gaussian} [ \inf_{\substack{\dv \in \R^N, \\ \dv_{\ssupp} = \vnull, \lnorm{\dv_{\ssuppc}}[\infty] \leq 1}} \lnorm{\gaussian - \tau \TV^\T \sign(\TV \grtr) - \tau \TV^\T \dv}^2 ] \\*
	&= \inf_{\tau > 0} \mean{}_{\gaussian}[ \inf_{\dv \in \F} \lnorm{ \gaussian - \tau \TV^\T \dv}^2 ] \\
	&\leq \inf_{\tau > 0} \mean{}_{\gaussian}[ \inf_{\dv \in \Fext} \lnorm{ \gaussian - \tau \TV^\T \dv}^2 ] \\
	&= \inf_{\tau > 0} \mean{}_{\gaussian}[ \inf_{\dv \in \Fext} \lnorm{ \H \gaussian - \tau \H \TV^\T \dv}^2 ]\\
	&= \inf_{\tau > 0} \mean{}_{\gaussian}[ \inf_{\dv \in \Fext} \lnorm{ \gaussian - \tau \H \TV^\T \dv}^2 ],
	\end{align}
	where we have used in the fourth line that $\H$ is orthogonal, and in the last line that standard Gaussian random vectors are rotationally invariant.
\end{proof}

\subsubsection{Step 2(a): Signal-Dependent Trees}
\label{eq:proof:step2a}

In order to continue with the simple polar bound \eqref{eq:proofs:mw:unitarypolar} from Lemma~\ref{lem:proofs:mw:unitarypolar}, we need to handle a least-squares problem with box-constraint set $\Fext$.
At first sight, this might seem less challenging than calculating the conic mean width explicitly, but unfortunately, the inner optimization problem on the right-hand side of \eqref{eq:proofs:mw:unitarypolar},
\begin{equation}\label{eq:proofs:mw:boxconstr}
\inf_{\dv \in \Fext} \lnorm{\gaussian - \tau \H \TV^\T \dv}^2,
\end{equation}
has no closed-form solution ($\gaussian$ and $\tau$ are both fixed here).
For the simplest choice $\H = \I{n}$, it turns out that the variables in $\dv$ are strongly coupled by the action of $\TV^\T$, so that it is unclear how to \emph{explicitly} construct a good approximate solution to \eqref{eq:proofs:mw:boxconstr}.
As already sketched in roadmap subsequent to Theorem~\ref{thm:results:mwbound}, the key step of our proof is to design a matrix $\H \in \R^{n \times n}$ that allows us to decouple these dependencies and thereby to derive a meaningful upper bound for \eqref{eq:proofs:mw:boxconstr}.

Such a construction of $\H$ requires some technical effort. The aim of the current substep is to arrange the extended gradient support set $\Supp$ according to a specific binary tree.
The multi-level structure of the resulting tree then leads to the desired transform matrix $\H$ in Step~\hyperref[eq:proof:step2b]{2(b)}, which corresponds to a signal-dependent (non-dyadic) Haar wavelet basis.

First of all, we need several general definitions for our tree representation of $\Supp$. While the used notation might not be standard in the literature, it is consistent with the common terminology in graph theory.
\begin{definition}[Vertex] \label{def:proofs:mw:vertex}
	We call a tuple $(\ell,i) \in \{ (k, j) \suchthat k \in \N,\ j \in [2^{k-1}] \}$ a \emph{vertex} at \emph{level} $\ell$, and $(1,1)$ is referred to as the \emph{root vertex}.
	If $(\ell, i)$ is a non-root vertex, we denote by
	\begin{equation}
	\parent(\ell, i)  \coloneqq  \begin{cases}
	(\ell-1, \tfrac{i+1}{2}), &\text{if $i$ odd}, \\
	(\ell-1, \tfrac{i}{2}), & \text{if $i$ even},
	\end{cases}
	\end{equation}
	the \emph{parent} of $(\ell,i)$. Whenever $(k,j) = \parent(\ell,i)$ is well-defined, we call $\lchild(k,j) \coloneqq (\ell,i)$ the \emph{left child} of $(k,j)$ if $i$ is odd and $\rchild(k,j) \coloneqq (\ell,i)$ the \emph{right child} of $(k,j)$ if $i$ is even. As shorthand notation, we also define $\child(k,j)  \coloneqq  \{\lchild(k,j), \rchild(k,j)\}$.
	Finally, for $\ell \in \N$, we denote the \emph{$\ell$-th level set} by $\lambda_{\ell} \coloneqq \{ (\ell, i) \suchthat  i\in[2^{\ell-1}] \}$.
\end{definition}
From these definitions, it is clear that every parent and every left (right) child are also vertices. Every vertex has exactly one left (right) child and every non-root vertex has exactly one parent.
\begin{definition}[Binary tree]
	We call a set of vertices $\Lambda$ a \emph{binary tree} (or simply a \emph{tree}) if it contains the root vertex and $\parent(\ell,i) \in \Lambda$ holds true for every non-root vertex $(\ell,i) \in \Lambda$. The \emph{depth} of a tree~$\Lambda$ is defined as $\depth(\Lambda)  \coloneqq  \max \{\ell \suchthat (\ell,i) \in \Lambda \} \in \N \union \{+\infty\}$, and by
	\begin{equation}
	\ceil{ (\ell,i) }  \coloneqq  \bigg\{ (k,j) \in \Lambda \suchthat (k,j) = (\underbrace{\parent \circ \dots \circ \parent}_{\text{$d$ times}})(\ell,i) \text{ for some } d \in \N  \bigg\}
	\end{equation}
	we define the \emph{ancestors} of a vertex $(\ell,i) \in \Lambda$.
\end{definition}
We note that a tree $\Lambda$ of depth $L  \coloneqq  \depth(\Lambda) < \infty$ always satisfies the inequality $\cardinality{\Lambda} \leq 2^L-1$.

Next, we iteratively generate two signal-dependent trees $\Lambda$ and $\Lambda^{L_0} \subset \Lambda$ such that each vertex of~$\Lambda$ and~$\Lambda^{L_0}$ corresponds to (exactly) one element of $[N]$ and $\Supp$, respectively.
\revision{It is useful to bear in mind that the elements of the introduced subsets of $[N]$ correspond to faces of $\grtr$; see also Figure~\ref{fig:proofs:mw:face_edges}.
The intuitive idea behind Iteration~\ref{def:proofs:mw:signaltree} below is a simple bipartition scheme: given a (consecutive) subset of faces $\SuppS{i}{\ell} \subset [N]$ at level $\ell$, bisect it at a pivot face $\p{i}{\ell} \in \SuppS{i}{\ell}$ and proceed recursively with the left and right section. The resulting pivot faces form the desired tree $\Lambda$; see Figure~\ref{fig:trees} for an illustration by means of two examples.}

\revision{
\begin{example}[Dyadic signals I]
\label{ex:dyadic_1}
This construction is best understood in the special case of a dyadic signal $\grtr \in \R^{n}$ with $n = 2^L$. Here, $\Lambda$ is a perfect binary tree of depth $\depth(\Lambda) = L$, which corresponds to the faces of $\grtr$.
Furthermore, $\Lambda^{L_0} \subset \Lambda$ is simply the (perfect) subtree of the first  $L_0 =  \log_2(\sext+1)$ levels, corresponding to the gradient support $\Supp = \ssupp$; see Figure~\ref{fig:trees:1} for an illustration. In this case, the \emph{pivot faces} of Iteration~\ref{def:proofs:mw:signaltree} can be explicitly represented by $\p{i}{\ell} = (i-1)\cdot 2^{L-\ell-1} + 2^{L-\ell}$ and we have that $\cardinality{\SuppS{i}{\ell}} = 2^{L-\ell+1} -1$. 
 
For arbitrary gradient-sparse signals, the balance of $\Lambda$ reflects how well the jump discontinuities $\Supp$ of $\grtr$ are separated: the more $\Supp$ deviates from being equidistant, the less $\Lambda$ is balanced; see Figure~\ref{fig:trees:2} for an example.
\end{example}
}


\begin{figure}
 \centering
 \begin{subfigure}[t]{1\textwidth}
 \centering
 \begin{tikzpicture}[thick,scale=0.75, every node/.style={scale=0.7}]
  
  \definecolor{myo}{rgb}{1,0.46,0.15}
  \definecolor{myb}{rgb}{0.35,0.46,0.8}
  \definecolor{myb2}{rgb}{0.81,0.84,0.94}
  \definecolor{Sset}{gray}{0.5}
  \definecolor{myt}{rgb}{0.95,0.75,0.3}
  \definecolor{Sset2}{rgb}{0.35,0.46,0.8}
  
  \tikzstyle{signaldots}=[mark size=1.6pt,color=myb,rotate=-90]
  \tikzstyle{jumpdots}=[mark size=2pt,color=gray,rotate=-90]
  \tikzstyle{jumpdotsS}=[mark size=2pt,color=myo,rotate=-90] 
  \tikzstyle{treeS}=[mark size=2pt,color=myo,rotate=-90]
  \tikzstyle{treeSc}=[mark size=2pt,color=gray,rotate=-90]
  
  \draw[step=1.0,gray,ultra thin,dashed] (1,2) grid (15,6);
  \draw[ultra thin,dashed,gray!50] (0,6) -- (1,6);
  \draw[ultra thin,dashed,gray!50] (13,6) -- (16,6);
  \draw[thin] (1,6) -- (15,6);
  
  \draw[Sset2,->] (2,3) -- (2,5.9);
  \node[color=Sset2] at (1.7,5.5) {\small $\p{1}{3}$};
  
  \foreach \i in {1,...,4} {
          \node[mark size=1.6pt,color=Sset2,rotate=-90] at (\i-0.5,-1.2) {\pgfuseplotmark{square*}};;
    }
  \draw[color=Sset2,dashed,ultra thin] (0.5,-1.2) --(3.5,-1.2);
  \node[color=Sset2] at (3.7,-0.85) { $\Q{1}{3}$};

  \node[jumpdots] at (1,6) {\pgfuseplotmark{triangle*}};
  \node[color=gray] at (1,6.3) {1};
  \node[jumpdots] at (2,6) {\pgfuseplotmark{triangle*}};
  \node[color=gray] at (2,6.3) {2};
  \node[jumpdots] at (3,6) {\pgfuseplotmark{triangle*}};
  \node[color=gray] at (3,6.3) {3};
  \node[jumpdots] at (5,6) {\pgfuseplotmark{triangle*}};
  \node[color=gray] at (5,6.3) {5};
  \node[jumpdots] at (6,6) {\pgfuseplotmark{triangle*}};
  \node[color=gray] at (6,6.3) {6};
  \node[jumpdots] at (7,6) {\pgfuseplotmark{triangle*}};
  \node[color=gray] at (7,6.3) {7};
  \node[jumpdots] at (9,6) {\pgfuseplotmark{triangle*}};
  \node[color=gray] at (9,6.3) {9};
  \node[jumpdots] at (10,6) {\pgfuseplotmark{triangle*}};
  \node[color=gray] at (10,6.3) {10};
  \node[jumpdots] at (11,6) {\pgfuseplotmark{triangle*}};
  \node[color=gray] at (11,6.3) {11};
  \node[jumpdots] at (13,6) {\pgfuseplotmark{triangle*}};
  \node[color=gray] at (13,6.3) {13};
  \node[jumpdots] at (14,6) {\pgfuseplotmark{triangle*}};
  \node[color=gray] at (14,6.3) {14};
  \node[jumpdots] at (15,6) {\pgfuseplotmark{triangle*}};
  \node[color=gray] at (15,6.3) {$15= N$};
  
  \node[jumpdotsS] at (4,6) {\pgfuseplotmark{triangle*}};
  \node[color=gray] at (4,6.3) {4}; 
  \node[jumpdotsS] at (8,6) {\pgfuseplotmark{triangle*}};
  \node[color=gray] at (8,6.3) {8}; 
  \node[jumpdotsS] at (12,6) {\pgfuseplotmark{triangle*}};
  \node[color=gray] at (12,6.3) {12}; 
  
  \node[mark size=2pt,color=gray!50,rotate=-90] at (0,6) {\pgfuseplotmark{triangle*}};
  \node[color=gray!50] at (0,6.3) {0}; 
  \node[mark size=2pt,color=gray!50,rotate=-90] at (16,6) {\pgfuseplotmark{triangle*}};
  \node[color=gray!50] at (16,6.3) {16};

  \draw (4,4) -- (8,5);
  \draw (12,4) -- (8,5);
  \draw (12,4) -- (14,3);
  \draw (12,4) -- (10,3);
  \draw (4,4) -- (2,3);
  \draw (4,4) -- (6,3);
  \draw (2,3) -- (1,2);
  \draw (2,3) -- (3,2);
  \draw (6,3) -- (7,2);
  \draw (6,3) -- (5,2);
  \draw (10,3) -- (9,2);
  \draw (10,3) -- (11,2);
  \draw (14,3) -- (15,2);
  \draw (14,3) -- (13,2);

  \node[treeS] at (8,5) {\pgfuseplotmark{*}};
  \node[color=gray] at (8.4,5.2) {\small $(1,1)$}; 
  \node[treeS] at (4,4) {\pgfuseplotmark{*}};
  \node[color=gray] at (3.6,4.2) {\small $(2,1)$}; 
  \node[treeS] at (12,4) {\pgfuseplotmark{*}};
  \node[color=gray] at (12.4,4.2) {\small $(2,2)$};   
  \node[color=gray] at (14,3) {\pgfuseplotmark{*}};
  \node[color=gray] at (14.4,3.2) {\small $(3,4)$}; 
  \node[treeSc] at (2,3) {\pgfuseplotmark{*}};
  \node[color=gray] at (1.6,3.2) {\small $(3,1)$};
  \node[treeSc] at (6,3) {\pgfuseplotmark{*}};
  \node[color=gray] at (6.4,3.2) {\small $(3,2)$};
  \node[treeSc] at (1,2) {\pgfuseplotmark{*}};
  \node[color=gray] at (0.6,2.2) {\small $(4,1)$};
  \node[treeSc] at (3,2) {\pgfuseplotmark{*}};
  \node[color=gray] at (3.4,2.2) {\small $(4,2)$};
  \node[treeSc] at (5,2) {\pgfuseplotmark{*}};
  \node[color=gray] at (4.6,2.2) {\small $(4,3)$};
  \node[treeSc] at (7,2) {\pgfuseplotmark{*}};
  \node[color=gray] at (7.4,2.2) {\small $(4,4)$};
  \node[treeSc] at (11,2) {\pgfuseplotmark{*}};
  \node[color=gray] at (11.4,2.2) {\small $(4,6)$}; 
  \node[treeSc] at (13,2) {\pgfuseplotmark{*}};
  \node[color=gray] at (12.6,2.2) {\small $(4,7)$}; 
  \node[treeSc] at (15,2) {\pgfuseplotmark{*}};
  \node[color=gray] at (14.4,2.15) {\small $(4,8)$}; 
  \node[treeSc] at (10,3) {\pgfuseplotmark{*}};
  \node[color=gray] at (9.6,3.2) {\small $(3,3)$};
  \node[treeSc] at (9,2) {\pgfuseplotmark{*}};
  \node[color=gray] at (8.6,2.2) {\small $(4,5)$}; 
  
  \draw[color=myb2,dashed] (0,8.5) -- (4,8.5);
  \draw[color=myb2,dashed] (4,8.5) -- (4,7.5);
  \draw[color=myb2,dashed] (4,7.5) -- (8,7.5);
  \draw[color=myb2,dashed] (8,7.5) -- (8,6.8);
  \draw[color=myb2,dashed] (8,6.8) -- (12,6.8);
  \draw[color=myb2,dashed] (12,6.8) -- (12,8.9);
  \draw[color=myb2,dashed] (12,8.9) -- (16,8.9);
  \node[signaldots] at (0.5,8.5) {\pgfuseplotmark{square*}};
  \node[signaldots] at (1.5,8.5) {\pgfuseplotmark{square*}};
  \node[signaldots] at (2.5,8.5) {\pgfuseplotmark{square*}};
  \node[signaldots] at (3.5,8.5) {\pgfuseplotmark{square*}};
  \node[signaldots] at (4.5,7.5) {\pgfuseplotmark{square*}};
  \node[signaldots] at (5.5,7.5) {\pgfuseplotmark{square*}};
  \node[signaldots] at (6.5,7.5) {\pgfuseplotmark{square*}};
  \node[signaldots] at (7.5,7.5) {\pgfuseplotmark{square*}};
  \node[signaldots] at (8.5,6.8) {\pgfuseplotmark{square*}};
  \node[signaldots] at (9.5,6.8) {\pgfuseplotmark{square*}};
  \node[signaldots] at (10.5,6.8) {\pgfuseplotmark{square*}};
  \node[signaldots] at (11.5,6.8) {\pgfuseplotmark{square*}};
  \node[signaldots] at (12.5,8.9) {\pgfuseplotmark{square*}};
  \node[signaldots] at (13.5,8.9) {\pgfuseplotmark{square*}};
  \node[signaldots] at (14.5,8.9) {\pgfuseplotmark{square*}};
  \node[signaldots] at (15.5,8.9) {\pgfuseplotmark{square*}};
  \node[color=myb] at (0.5,8.2) {\small $x^\ast_1$};
  \node[color=myb] at (1.5,8.2) {\small $x^\ast_2$};
  \node[color=myb] at (2.5,8.2) {\small \dots};
  \node[color=myb] at (15.5,8.6) {\small $x^\ast_{16}$};
  
  \node[color=gray] at (16,5.05) {\small $\ell = 1$};
  \node[color=gray] at (16,4.05) {\small $\ell = 2$};
  \node[color=gray] at (16,3.05) {\small $\ell = 3$};
  \node[color=gray] at (16,2.05) {\small $\ell = 4$};

  \node[color=gray] at (16,0.95) {\small $\ell = 1$};
  \node[color=gray] at (16,-0.05) {\small $\ell = 2$};
  \node[color=gray] at (16,-1.05) {\small $\ell = 3$};
  \node[color=gray] at (16,-2.05) {\small $\ell = 4$};

  \foreach \i in {1,...,15} {
          \node[mark size=2pt,color=Sset,rotate=-90] at (\i,1) {\pgfuseplotmark{triangle*}};;
    }
  \draw[color=Sset,dashed,ultra thin] (1,1) --(15,1);  
  \node[mark size=2pt,color=myo,rotate=-90] at (8,1) {\pgfuseplotmark{triangle*}};
  \node[mark size=2pt,color=myo,rotate=-90] at (12,1) {\pgfuseplotmark{triangle*}};
  \node[mark size=2pt,color=myo,rotate=-90] at (4,1) {\pgfuseplotmark{triangle*}};
  \node[color=myo] at (8,1.4) { $\SuppS{1}{1}$};
  
  \foreach \i in {1,...,7} {
          \node[mark size=2pt,color=Sset,rotate=-90] at (\i,0) {\pgfuseplotmark{triangle*}};;
    }
  \draw[color=Sset,dashed,ultra thin] (1,0) --(7,0);
  \node[mark size=2pt,color=myo,rotate=-90] at (4,0) {\pgfuseplotmark{triangle*}};
  \node[color=myo] at (4,0.4) { $\SuppS{1}{2}$};
  
  \foreach \i in {9,...,15} {
          \node[mark size=2pt,color=Sset,rotate=-90] at (\i,0) {\pgfuseplotmark{triangle*}};;
    }
  \draw[color=Sset,dashed,ultra thin] (9,0) --(15,0);
  \node[mark size=2pt,color=myo,rotate=-90] at (12,0) {\pgfuseplotmark{triangle*}};
  \node[color=myo] at (12,0.4) { $\SuppS{2}{2}$};
  
  \foreach \i in {1,...,3} {
          \node[mark size=2pt,color=Sset,rotate=-90] at (\i,-1) {\pgfuseplotmark{triangle*}};;
    }
  \draw[color=Sset,dashed,ultra thin] (1,-1) --(3,-1); 
  \node[color=Sset] at (2,-0.6) { $\SuppS{1}{3}$};
  
   \foreach \i in {5,...,7} {
          \node[mark size=2pt,color=Sset,rotate=-90] at (\i,-1) {\pgfuseplotmark{triangle*}};;
    }
  \draw[color=Sset,dashed,ultra thin] (5,-1) --(7,-1);
  \node[color=Sset] at (6,-0.6) { $\SuppS{2}{3}$};
  
  \foreach \i in {9,...,11} {
          \node[mark size=2pt,color=Sset,rotate=-90] at (\i,-1) {\pgfuseplotmark{triangle*}};;
    }
  \draw[color=Sset,dashed,ultra thin] (9,-1) --(11,-1);
  \node[color=gray] at (10,-0.6) { $\SuppS{3}{3}$};
  
   \foreach \i in {13,...,15} {
          \node[mark size=2pt,color=Sset,rotate=-90] at (\i,-1) {\pgfuseplotmark{triangle*}};;
    }
  \draw[color=Sset,dashed,ultra thin] (13,-1) --(15,-1);
  \node[color=Sset] at (14,-0.6) { $\SuppS{4}{3}$};
  
  \node[mark size=2pt,color=Sset,rotate=-90] at (1,-2) {\pgfuseplotmark{triangle*}};
  \node[color=Sset] at (1,-1.6) { $\SuppS{1}{4}$};
  
  \node[mark size=2pt,color=Sset,rotate=-90] at (3,-2) {\pgfuseplotmark{triangle*}};
  \node[color=Sset] at (3,-1.6) { $\SuppS{2}{4}$};

  \node[mark size=2pt,color=Sset,rotate=-90] at (5,-2) {\pgfuseplotmark{triangle*}};
  \node[color=Sset] at (5,-1.6) { $\SuppS{3}{4}$};
  
  \node[mark size=2pt,color=Sset,rotate=-90] at (7,-2) {\pgfuseplotmark{triangle*}};
  \node[color=Sset] at (7,-1.6) { $\SuppS{4}{4}$};
  
  \node[mark size=2pt,color=Sset,rotate=-90] at (13,-2) {\pgfuseplotmark{triangle*}};
  \node[color=Sset] at (13,-1.6) { $\SuppS{7}{4}$};
  
  \node[mark size=2pt,color=Sset,rotate=-90] at (15,-2) {\pgfuseplotmark{triangle*}};
  \node[color=Sset] at (15,-1.6) { $\SuppS{8}{4}$};
  
  \node[mark size=2pt,color=Sset,rotate=-90] at (9,-2) {\pgfuseplotmark{triangle*}};
  \node[color=Sset] at (9,-1.6) { $\SuppS{5}{4}$};
    
  \node[mark size=2pt,color=Sset,rotate=-90] at (11,-2) {\pgfuseplotmark{triangle*}};
  \node[color=Sset] at (11,-1.6) { $\SuppS{6}{4}$};
  
  \draw[black!30!green,->,thin,dashed] plot [smooth] coordinates {(6,3) (6.4,4.0) (8,5)};
  \node[color=black!30!green] at (7,4.0) {\small $\pr{2}{3}$};
  
  \draw[black!30!green,->,thin,dashed] plot [smooth] coordinates {(13,2) (13.0,3.0) (12,4)};
  \node[color=black!30!green] at (12.7,2.7) {\small $\pl{7}{4}$};

    %
    %
    %
    %
    %
  \end{tikzpicture}
   \caption{}
   \label{fig:trees:1}
   \end{subfigure}
  
  \vspace{1em}
  \begin{subfigure}[t]{1\textwidth}
  \centering
  \begin{tikzpicture}[thick,scale=0.75, every node/.style={scale=0.75}]
  
  \definecolor{myo}{rgb}{1,0.46,0.15}
  \definecolor{myb}{rgb}{0.35,0.46,0.8}
  \definecolor{myb2}{rgb}{0.81,0.84,0.94}
  \definecolor{Sset}{gray}{0.5}
  \definecolor{myt}{rgb}{0.95,0.75,0.3}
  \definecolor{Sset2}{rgb}{0.35,0.46,0.8}
  
  \tikzstyle{signaldots}=[mark size=1.6pt,color=myb,rotate=-90]
  \tikzstyle{jumpdots}=[mark size=2pt,color=gray,rotate=-90]
  \tikzstyle{jumpdotsS}=[mark size=2pt,color=myo,rotate=-90] 
  \tikzstyle{treeS}=[mark size=2pt,color=myo,rotate=-90]
  \tikzstyle{treeSc}=[mark size=2pt,color=gray,rotate=-90]
  
  \draw[step=1.0,gray,ultra thin,dashed] (1,1) grid (15,6);
  \draw[ultra thin,dashed,gray!50] (0,6) -- (1,6);
  \draw[ultra thin,dashed,gray!50] (15,6) -- (16,6);
  \draw[thin] (1,6) -- (15,6);
  
  \draw[Sset2,->] (2,3) -- (2,5.9);
  \node[color=Sset2] at (1.7,5.5) {\small $\p{1}{3}$};
  
  \foreach \i in {1,...,4} {
          \node[mark size=1.6pt,color=Sset2,rotate=-90] at (\i-0.5,-2.2) {\pgfuseplotmark{square*}};;
    }
  \draw[color=Sset2,dashed,ultra thin] (0.5,-2.2) --(3.5,-2.2);
  \node[color=Sset2] at (3.7,-1.85) { $\Q{1}{3}$};

  \node[jumpdots] at (1,6) {\pgfuseplotmark{triangle*}};
  \node[color=gray] at (1,6.3) {1};
  \node[jumpdots] at (2,6) {\pgfuseplotmark{triangle*}};
  \node[color=gray] at (2,6.3) {2};
  \node[jumpdots] at (3,6) {\pgfuseplotmark{triangle*}};
  \node[color=gray] at (3,6.3) {3};
  \node[jumpdots] at (5,6) {\pgfuseplotmark{triangle*}};
  \node[color=gray] at (5,6.3) {5};
  \node[jumpdots] at (6,6) {\pgfuseplotmark{triangle*}};
  \node[color=gray] at (6,6.3) {6};
  \node[jumpdots] at (9,6) {\pgfuseplotmark{triangle*}};
  \node[color=gray] at (9,6.3) {9};
  \node[jumpdots] at (10,6) {\pgfuseplotmark{triangle*}};
  \node[color=gray] at (10,6.3) {10};
  \node[jumpdots] at (11,6) {\pgfuseplotmark{triangle*}};
  \node[color=gray] at (11,6.3) {11};
  \node[jumpdots] at (13,6) {\pgfuseplotmark{triangle*}};
  \node[color=gray] at (13,6.3) {13};
  \node[jumpdots] at (14,6) {\pgfuseplotmark{triangle*}};
  \node[color=gray] at (14,6.3) {14};
  \node[jumpdots] at (15,6) {\pgfuseplotmark{triangle*}};
  \node[color=gray] at (15,6.3) {$15= N$};
  
  \node[jumpdotsS] at (4,6) {\pgfuseplotmark{triangle*}};
  \node[color=gray] at (4,6.3) {4}; 
  \node[jumpdotsS] at (7,6) {\pgfuseplotmark{triangle*}};
  \node[color=gray] at (7,6.3) {7}; 
  \node[jumpdotsS] at (8,6) {\pgfuseplotmark{triangle*}};
  \node[color=gray] at (8,6.3) {8}; 
  \node[jumpdotsS] at (12,6) {\pgfuseplotmark{triangle*}};
  \node[color=gray] at (12,6.3) {12}; 
  
  \node[mark size=2pt,color=gray!50,rotate=-90] at (0,6) {\pgfuseplotmark{triangle*}};
  \node[color=gray!50] at (0,6.3) {0}; 
  \node[mark size=2pt,color=gray!50,rotate=-90] at (16,6) {\pgfuseplotmark{triangle*}};
  \node[color=gray!50] at (16,6.3) {16};

  \draw (4,4) -- (7,5);
  \draw (8,4) -- (7,5);
  \draw (8,4) -- (12,3);
  \draw (4,4) -- (2,3);
  \draw (4,4) -- (5,3);
  \draw (2,3) -- (1,2);
  \draw (2,3) -- (3,2);
  \draw (5,3) -- (6,2);
  \draw (12,3) -- (10,2);
  \draw (12,3) -- (14,2);
  \draw (10,2) -- (9,1);
  \draw (10,2) -- (11,1);
  \draw (13,1) -- (14,2);
  \draw (15,1) -- (14,2);
  
  \node[treeS] at (7,5) {\pgfuseplotmark{*}};
  \node[color=gray] at (7.4,5.2) {\small $(1,1)$}; 
  \node[treeS] at (4,4) {\pgfuseplotmark{*}};
  \node[color=gray] at (3.6,4.2) {\small $(2,1)$}; 
  \node[treeS] at (8,4) {\pgfuseplotmark{*}};
  \node[color=gray] at (8.4,4.2) {\small $(2,2)$};   
  \node[treeS] at (12,3) {\pgfuseplotmark{*}};
  \node[color=gray] at (12.4,3.2) {\small $(3,4)$}; 
  \node[treeSc] at (2,3) {\pgfuseplotmark{*}};
  \node[color=gray] at (1.6,3.2) {\small $(3,1)$};
  \node[treeSc] at (5,3) {\pgfuseplotmark{*}};
  \node[color=gray] at (5.4,3.2) {\small $(3,2)$};
  \node[treeSc] at (1,2) {\pgfuseplotmark{*}};
  \node[color=gray] at (0.6,2.2) {\small $(4,1)$};
  \node[treeSc] at (3,2) {\pgfuseplotmark{*}};
  \node[color=gray] at (3.4,2.2) {\small $(4,2)$};
  \node[treeSc] at (6,2) {\pgfuseplotmark{*}};
  \node[color=gray] at (6.4,2.2) {\small $(4,4)$};
  \node[treeSc] at (10,2) {\pgfuseplotmark{*}};
  \node[color=gray] at (9.6,2.2) {\small $(4,7)$}; 
  \node[treeSc] at (14,2) {\pgfuseplotmark{*}};
  \node[color=gray] at (14.4,2.2) {\small $(4,8)$}; 
  \node[treeSc] at (13,1) {\pgfuseplotmark{*}};
  \node[color=gray] at (12.5,1.2) {\small $(5,15)$};
  \node[treeSc] at (15,1) {\pgfuseplotmark{*}};
  \node[color=gray] at (14.25,1.15) {\small $(5,16)$};
  \node[treeSc] at (9,1) {\pgfuseplotmark{*}};
  \node[color=gray] at (8.5,1.2) {\small $(5,13)$};
  \node[treeSc] at (11,1) {\pgfuseplotmark{*}};
  \node[color=gray] at (11.5,1.2) {\small $(5,14)$};
  
  \draw[color=myb2,dashed] (0,8.5) -- (4,8.5);
  \draw[color=myb2,dashed] (4,8.5) -- (4,7.5);
  \draw[color=myb2,dashed] (4,7.5) -- (7,7.5);
  \draw[color=myb2,dashed] (7,7.5) -- (7,6.8);
  \draw[color=myb2,dashed] (7,6.8) -- (8,6.8);
  \draw[color=myb2,dashed] (8,6.8) -- (8,7.9);
  \draw[color=myb2,dashed] (8,7.9) -- (12,7.9);
  \draw[color=myb2,dashed] (12,7.9) -- (12,8.9);
  \draw[color=myb2,dashed] (12,8.9) -- (16,8.9);
  \node[signaldots] at (0.5,8.5) {\pgfuseplotmark{square*}};
  \node[signaldots] at (1.5,8.5) {\pgfuseplotmark{square*}};
  \node[signaldots] at (2.5,8.5) {\pgfuseplotmark{square*}};
  \node[signaldots] at (3.5,8.5) {\pgfuseplotmark{square*}};
  \node[signaldots] at (4.5,7.5) {\pgfuseplotmark{square*}};
  \node[signaldots] at (5.5,7.5) {\pgfuseplotmark{square*}};
  \node[signaldots] at (6.5,7.5) {\pgfuseplotmark{square*}};
  \node[signaldots] at (7.5,6.8) {\pgfuseplotmark{square*}};
  \node[signaldots] at (8.5,7.9) {\pgfuseplotmark{square*}};
  \node[signaldots] at (9.5,7.9) {\pgfuseplotmark{square*}};
  \node[signaldots] at (10.5,7.9) {\pgfuseplotmark{square*}};
  \node[signaldots] at (11.5,7.9) {\pgfuseplotmark{square*}};
  \node[signaldots] at (12.5,8.9) {\pgfuseplotmark{square*}};
  \node[signaldots] at (13.5,8.9) {\pgfuseplotmark{square*}};
  \node[signaldots] at (14.5,8.9) {\pgfuseplotmark{square*}};
  \node[signaldots] at (15.5,8.9) {\pgfuseplotmark{square*}};
  \node[color=myb] at (0.5,8.2) {\small $x^\ast_1$};
  \node[color=myb] at (1.5,8.2) {\small $x^\ast_2$};
  \node[color=myb] at (2.5,8.2) {\small \dots};
  \node[color=myb] at (15.5,8.6) {\small $x^\ast_{16}$};
  
  \node[color=gray] at (16,5.05) {\small $\ell = 1$};
  \node[color=gray] at (16,4.05) {\small $\ell = 2$};
  \node[color=gray] at (16,3.05) {\small $\ell = 3$};
  \node[color=gray] at (16,2.05) {\small $\ell = 4$};
  \node[color=gray] at (16,1.05) {\small $\ell = 5$};

  \node[color=gray] at (16,0.05) {\small $\ell = 1$};
  \node[color=gray] at (16,-1.05) {\small $\ell = 2$};
  \node[color=gray] at (16,-2.05) {\small $\ell = 3$};
  \node[color=gray] at (16,-3.05) {\small $\ell = 4$};
  \node[color=gray] at (16,-4.05) {\small $\ell = 5$};

  \foreach \i in {1,...,15} {
          \node[mark size=2pt,color=Sset,rotate=-90] at (\i,0) {\pgfuseplotmark{triangle*}};;
    }
  \draw[color=Sset,dashed,ultra thin] (1,0) --(15,0);  
  \node[mark size=2pt,color=myo,rotate=-90] at (7,0) {\pgfuseplotmark{triangle*}};
  \node[mark size=2pt,color=myo,rotate=-90] at (8,0) {\pgfuseplotmark{triangle*}};
  \node[mark size=2pt,color=myo,rotate=-90] at (4,0) {\pgfuseplotmark{triangle*}};
  \node[mark size=2pt,color=myo,rotate=-90] at (12,0) {\pgfuseplotmark{triangle*}};
  \node[color=myo] at (8,0.4) { $\SuppS{1}{1}$};
  
  \foreach \i in {1,...,6} {
          \node[mark size=2pt,color=Sset,rotate=-90] at (\i,-1) {\pgfuseplotmark{triangle*}};;
    }
  \draw[color=Sset,dashed,ultra thin] (1,-1) --(6,-1);
  \node[mark size=2pt,color=myo,rotate=-90] at (4,-1) {\pgfuseplotmark{triangle*}};
  \node[color=myo] at (3.5,-0.6) { $\SuppS{1}{2}$};
  
  \foreach \i in {8,...,15} {
          \node[mark size=2pt,color=Sset,rotate=-90] at (\i,-1) {\pgfuseplotmark{triangle*}};;
    }
  \draw[color=Sset,dashed,ultra thin] (8,-1) --(15,-1);
  \node[mark size=2pt,color=myo,rotate=-90] at (8,-1) {\pgfuseplotmark{triangle*}};
  \node[mark size=2pt,color=myo,rotate=-90] at (12,-1) {\pgfuseplotmark{triangle*}};
  \node[color=myo] at (11.5,-0.6) { $\SuppS{2}{2}$};
  
  \foreach \i in {1,...,3} {
          \node[mark size=2pt,color=Sset,rotate=-90] at (\i,-2) {\pgfuseplotmark{triangle*}};;
    }
  \draw[color=Sset,dashed,ultra thin] (1,-2) --(3,-2); 
  \node[color=Sset] at (2,-1.6) { $\SuppS{1}{3}$};
  
   \foreach \i in {5,...,6} {
          \node[mark size=2pt,color=Sset,rotate=-90] at (\i,-2) {\pgfuseplotmark{triangle*}};;
    }
  \draw[color=Sset,dashed,ultra thin] (5,-2) --(6,-2);
  \node[color=Sset] at (5.5,-1.6) { $\SuppS{2}{3}$};
  
  \foreach \i in {9,...,15} {
          \node[mark size=2pt,color=Sset,rotate=-90] at (\i,-2) {\pgfuseplotmark{triangle*}};;
    }
  \draw[color=Sset,dashed,ultra thin] (9,-2) --(15,-2);
  \node[mark size=2pt,color=myo,rotate=-90] at (12,-2) {\pgfuseplotmark{triangle*}};
  \node[color=myo] at (12,-1.6) { $\SuppS{4}{3}$};
  
  \node[mark size=2pt,color=Sset,rotate=-90] at (1,-3) {\pgfuseplotmark{triangle*}};
  \node[color=Sset] at (1,-2.6) { $\SuppS{1}{4}$};
  
  \node[mark size=2pt,color=Sset,rotate=-90] at (3,-3) {\pgfuseplotmark{triangle*}};
  \node[color=Sset] at (3,-2.6) { $\SuppS{2}{4}$};

  \node[mark size=2pt,color=Sset,rotate=-90] at (6,-3) {\pgfuseplotmark{triangle*}};
  \node[color=Sset] at (6,-2.6) { $\SuppS{4}{4}$};
  
  \foreach \i in {9,...,11} {
          \node[mark size=2pt,color=Sset,rotate=-90] at (\i,-3) {\pgfuseplotmark{triangle*}};;
    }
  \draw[color=Sset,dashed,ultra thin] (9,-3) --(11,-3);
  \node[color=Sset] at (10,-2.6) { $\SuppS{7}{4}$};
  
  \foreach \i in {13,...,15} {
          \node[mark size=2pt,color=Sset,rotate=-90] at (\i,-3) {\pgfuseplotmark{triangle*}};;
    }
  \draw[color=Sset,dashed,ultra thin] (13,-3) --(15,-3);  
  \node[color=Sset] at (14,-2.6) { $\SuppS{8}{4}$};
  
  \node[mark size=2pt,color=Sset,rotate=-90] at (9,-4) {\pgfuseplotmark{triangle*}};
  \node[color=Sset] at (9,-3.6) { $\SuppS{13}{5}$};
    
  \node[mark size=2pt,color=Sset,rotate=-90] at (11,-4) {\pgfuseplotmark{triangle*}};
  \node[color=Sset] at (11,-3.6) { $\SuppS{14}{5}$};
  
  \node[mark size=2pt,color=Sset,rotate=-90] at (13,-4) {\pgfuseplotmark{triangle*}};
  \node[color=Sset] at (13,-3.6) { $\SuppS{15}{5}$};
  
  \node[mark size=2pt,color=Sset,rotate=-90] at (15,-4) {\pgfuseplotmark{triangle*}};
  \node[color=Sset] at (15,-3.6) { $\SuppS{16}{5}$};
  
  \draw[black!30!green,->,thin,dashed] plot [smooth] coordinates {(6,2) (6.4,4.0) (7,5)};
  \node[color=black!30!green] at (7,4.0) {\small $\pr{4}{4}$};
  
  \draw[black!30!green,->,thin,dashed] plot [smooth] coordinates {(13,1) (13.0,2.0) (12,3)};
  \node[color=black!30!green] at (12.7,1.7) {\small $\pl{15}{5}$};

  \draw[color=myb2,dashed] (0.7,-4) -- (1.3,-4);
  \node[signaldots] at (1,-4) {\pgfuseplotmark{square*}};
  \node[color=myb] at (2,-4) {$\grtr \in \R^n$};
  
  \node[mark size=2pt,color=myo,rotate=-90] at (1,-4.5) {\pgfuseplotmark{triangle*}};
  \node[color=myo] at (2,-4.5) {$\Supp \subset [N]$};
  
  \node[mark size=2pt,color=gray,rotate=-90] at (3,-4) {\pgfuseplotmark{triangle*}};
  \node[color=gray] at (4,-4) {$\setcompl{\Supp}\subset [N]$};
  
  \node[mark size=2pt,color=myo,rotate=-90] at (3,-4.5) {\pgfuseplotmark{*}};
  \node[color=myo] at (3.7,-4.5) {$\Lambda^{L_0}$};
  
  \node[mark size=2pt,color=myo,rotate=-90] at (5.1,-4) {\pgfuseplotmark{*}};
  \node[mark size=2pt,color=gray,rotate=-90] at (5.35,-4) {\pgfuseplotmark{*}};
  \node at (5.8,-4) {$\Lambda$};
  
  \draw[thin] (0.5,-5) rectangle (6.2,-3.5);
  \end{tikzpicture}
  \caption{}
  \label{fig:trees:2}
  \end{subfigure}
  \caption{\revision{\textbf{Examples of signal-dependent trees.} Subfigure \subref{fig:trees:1} and  \subref{fig:trees:2} illustrate the binary tree construction of Iteration~\ref{def:proofs:mw:signaltree} for two different signals $\grtr \in \R^{16}$ with $3$ and $4$ jump discontinuities, respectively (note that we have $\Supp = \ssupp$ in both cases). The first signal is dyadic (i.e., the jumps are equidistant), resulting in a perfect tree; see Example~\ref{ex:dyadic_1}. At the bottom of each subfigure, we have visualized all subsets $\SuppS{i}{\ell} \subset [N]$, while for the sake of clarity, we only show the set $\Q{1}{3} \subset [n]$ as an example (see Definition~\ref{def:proofs:mw:signaltree:Q}). In fact, displaying every set $\Q{i}{\ell}$ would be redundant because they are closely related to their counterparts $\SuppS{i}{\ell}$; see Proposition~\ref{prop:proofs:mw:tree_properties:Q}\ref{prop:proofs:mw:tree_properties:Q_S}.}}
  \label{fig:trees}
\end{figure}
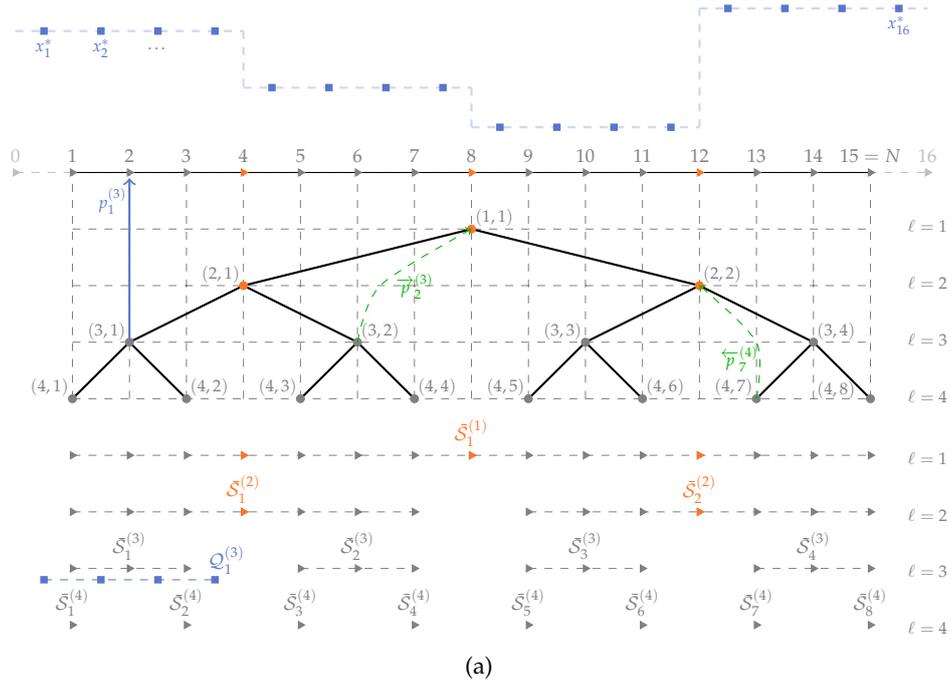
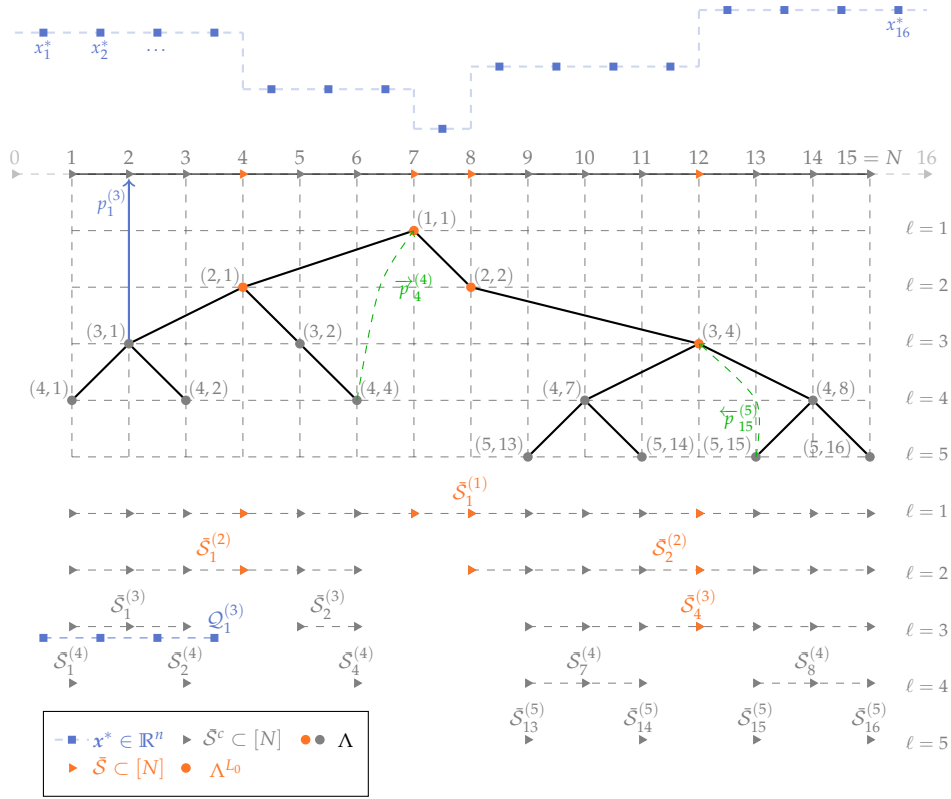

\begin{iteration}[Signal-dependent tree]\label{def:proofs:mw:signaltree}
	Let $\SuppS{1}{1}  \coloneqq  [N]$ denote the set of all faces of $\grtr$. We iteratively define for $\ell=1,2,\dots$ and $i \in [2^{\ell-1}]$ the following:
	If $\SuppS{i}{\ell} \neq \emptyset$, then select the \emph{pivot face} at vertex~$(\ell,i)$ as\footnote{By $\median(\cdot)$ we denote the \emph{lower median} of finite non-empty sets, i.e., if $M = \{v_1, \dots, v_d\}$ with $v_1 < \dots < v_d$, then $\median(M) = v_{\floor{(d+1)/2}}$.}
	\begin{equation} \label{eq:proofs:mw:signaltree:privot}
	\p{i}{\ell} \coloneqq \begin{cases}
	\median(\SuppS{i}{\ell} \intersec \Supp), & \text{ if } \SuppS{i}{\ell} \intersec \Supp \neq \emptyset, \\
	\median(\SuppS{i}{\ell}), & \text{ otherwise}.
	\end{cases}
	\end{equation}
	In order to invoke the next iteration step, we set
	\begin{align}
	\SuppS{j}{k}  \coloneqq  \big\{ \jump \in \SuppS{i}{\ell} \suchthat \jump < \p{i}{\ell} \big\} & \quad \text{ if } (k,j) = \lchild(\ell,i), \text{ and} \\*
	\SuppS{j}{k}  \coloneqq   \big\{ \jump \in \SuppS{i}{\ell} \suchthat \jump > \p{i}{\ell} \big\} & \quad \text{ if } (k,j) = \rchild(\ell,i).
	\end{align}
	The iteration ends at some level $L \in \N$, if $\SuppS{i}{L+1} = \emptyset$ for all $(L+1,i) \in \lambda_{L+1}$. We finally define $\Lambda  \coloneqq  \{ (\ell,i) \text{ is a vertex} \suchthat \SuppS{i}{\ell} \neq \emptyset \}$ and $\Lambda^{L_0}  \coloneqq  \{(\ell,i) \in \Lambda \suchthat \p{i}{\ell} \in \Supp \}$. In particular, $\Lambda$ and $\Lambda^{L_0}$ are binary trees with finite depths $L  \coloneqq  \depth(\Lambda)$ and $L_0 = \depth(\Lambda^{L_0})$, respectively.
\end{iteration}

We now show several important properties of the signal-dependent tree constructed in Iteration~\ref{def:proofs:mw:signaltree}, most importantly, that $\Lambda$ and $\Lambda^{L_0}$ can be indeed identified with the faces $[N]$ and the gradient support $\Supp$, respectively.
\begin{proposition} \label{prop:proofs:mw:tree_properties}
	Let $\Lambda$ and $\Lambda^{L_0}$ be the trees constructed in Iteration~\ref{def:proofs:mw:signaltree}. Then the following holds true:
	\begin{thmlist}
		\item\label{prop:proofs:mw:tree_properties:bijective}
		The map $p \colon \Lambda \to [N],\ (\ell,i) \mapsto \p{i}{\ell}$ is a bijection, i.e., every vertex of $\Lambda$ is mapped to exactly one of the $N$ faces of $\grtr$.
		Similarly, the restricted map $p\restrict_{\Lambda^{L_0}} \colon \Lambda^{L_0} \to \Supp$ is also a bijection.
		
		\item\label{prop:proofs:mw:tree_properties:supp_control}
		For every $(\ell,i) \in \Lambda$, we have that $\cardinality{\SuppS{i}{\ell} \intersec \Supp} \leq \floor[\big]{ \sext / 2^{\ell-1} }$.
		
		\item\label{prop:proofs:mw:tree_properties:supporttree}
		There exists $i \in [2^{L_0-1}]$ such that $\SuppS{i}{L_0} \intersec \Supp \neq \emptyset$, but $\SuppS{i}{L_0+1} \intersec \Supp = \emptyset$ for all $i \in [2^{L_0}]$.
		
		\item\label{prop:proofs:mw:tree_properties:offsupp_control}
		For every $i \in [2^{L_0}]$, we have that $\cardinality{\SuppS{i}{L_0+1}} \leq N-\sext$.
		
		\item\label{prop:proofs:mw:tree_properties:depth}
		Iteration~\ref{def:proofs:mw:signaltree} terminates after at most $\ceil{\log_2(\sext+1)} + \ceil{\log_2(n-\sext)}$ steps.
	\end{thmlist}
\end{proposition}
\begin{proof}
	\begin{prooflist}
		\item
		By construction, every element of $[N]$ is picked exactly once as pivot element. As the tree keeps record of all pivot elements, the claim follows directly.
		
		\item
		First, we note that $\cardinality{\SuppS{1}{1} \intersec \Supp} = \cardinality{\Supp} = \sext$. Now, let $(\ell,i) \in \Lambda$ be such that $\SuppS{i}{\ell} \intersec \Supp \neq \emptyset$. Then $\p{i}{\ell}$ is the median of $\SuppS{i}{\ell} \intersec \Supp$, and by construction, we have that $\cardinality{\SuppS{j}{k} \intersec \Supp} \leq \cardinality{\SuppS{i}{\ell} \intersec \Supp} / 2$ for $(k,j) \in \child(\ell,i)$. The claim now follows by induction.
		
		\item
		Towards a contradiction, assume that $\SuppS{i}{L_0} \intersec \Supp = \emptyset$ for all $i \in [2^{L_0-1}]$. Then, $\depth(\Lambda^{L_0}) \leq L_0-1$, and by the definition of $L_0$, we obtain $\cardinality{\Lambda^{L_0}} \leq 2^{L_0-1} - 1 < 2^{\log_2(\sext+1)} - 1 = \sext$. This contradicts the fact that $p\restrict_{\Lambda^{L_0}} \colon \Lambda^{L_0} \to \Supp$ is a bijection with $\cardinality{\Lambda^{L_0}} = \cardinality{\Supp} = \sext$.
		
		To see the second claim, we note that $\cardinality{\SuppS{i}{L_0+1} \intersec \Supp} \leq \floor{ \sext / 2^{\ceil{ \log_2(\sext+1)}} } < 1$ due to \ref{prop:proofs:mw:tree_properties:supp_control}, implying that the set is empty for every $i \in [2^{L_0}]$.
		
		\item
		Since $\SuppS{i}{L_0+1} \intersec \Supp = \emptyset$ due to \ref{prop:proofs:mw:tree_properties:supporttree}, we have $\SuppS{i}{L_0+1} \disjunion \Supp \subset [N]$ and therefore $\cardinality{\SuppS{i}{L_0+1}} + \cardinality{\Supp} \leq N$, which yields the claim.
		
		\item
		From \ref{prop:proofs:mw:tree_properties:offsupp_control}, we have an upper bound for $\cardinality{\SuppS{i}{L_0+1}}$. The claim now follows from the fact that we partition these sets along their median; cf.~\ref{prop:proofs:mw:tree_properties:supp_control}. \qedhere
	\end{prooflist}
\end{proof}

To simplify the presentation of the subsequent proof steps, we introduce further notation that is related to the above tree construction.
The following definition makes precise what is meant by neighboring faces, namely the closest previously selected pivot faces.
Based on that, we define specific subsets of nodes, denoted by $\Q{i}{\ell} \subset [n]$. As we will see later in Proposition~\ref{prop:proofs:mw:tree_properties:Q}\ref{prop:proofs:mw:tree_properties:Q_S}, each~$\Q{i}{\ell}$ can be seen as the nodal counterpart of $\SuppS{i}{\ell} \subset [N]$; see again Figure~\ref{fig:trees} for an illustration.

\begin{definition}\label{def:proofs:mw:signaltree:Q}
	Using the notation of Iteration~\ref{def:proofs:mw:signaltree}, we define the (\emph{left} and \emph{right}) \emph{neighbor faces} at vertex $(\ell,i) \in \Lambda$ by
	\begin{align} 
	\pl{i}{\ell} & \coloneqq  \max \Big(\big\{ \p{j}{k} \suchthat (k,j) \in \ceil{(\ell,i)}, \ \p{j}{k} < \p{i}{\ell} \big\} \union \{0\}\Big), \text{ and} \\*
	\pr{i}{\ell} & \coloneqq  \min \Big(\big\{ \p{j}{k} \suchthat (k,j) \in \ceil{(\ell,i)}, \ \p{j}{k} > \p{i}{\ell} \big\} \union \{n\}\Big). \label{eq:proofs:mw:signaltree:neighbors}
	\end{align}
	We denote by $\Q{i}{\ell} \subset [n]$ the nodes that are enclosed by $\pl{i}{\ell}$ and $\pr{i}{\ell}$, i.e.,
	\begin{equation}
		\Q{i}{\ell}  \coloneqq  \big\{ \pl{i}{\ell}+1,\dots,\pr{i}{\ell} \big\} \neq \emptyset.
	\end{equation}
	Moreover, by splitting $\Q{i}{\ell}$ at the pivot face $\p{i}{\ell}$, we define the subsets
	\begin{equation}
	\Ql{i}{\ell}  \coloneqq  \big\{ \pl{i}{\ell}+1,\dots,\p{i}{\ell} \big\} \neq \emptyset
	\end{equation}
	and
	\begin{equation}
	\Qr{i}{\ell}  \coloneqq  \big\{ \p{i}{\ell} +1 , \dots, \pr{i}{\ell}\big\} \neq \emptyset,
	\end{equation}
	so that $\Q{i}{\ell}  =  \Ql{i}{\ell} \disjunion \Qr{i}{\ell}$ for $(\ell,i) \in \Lambda$. 
\end{definition}

\revision{
\begin{example}[Dyadic signals II]
 In the special case of a dyadic signal, the quantities of Definition~\ref{def:proofs:mw:signaltree:Q} can again be explicitly determined. The left and right neighbor faces are given by $\pl{i}{\ell} = (i-1) \cdot 2^{L-\ell+1}$ and $\pr{i}{\ell} = i \cdot 2^{L-\ell+1}$. Furthermore, we have that $\cardinality{\Ql{i}{\ell}} = \cardinality{\Qr{i}{\ell}} = 2^{L-\ell}$. 
\end{example}
}

The following proposition relates $\Q{i}{\ell} \subset [n]$ to $\SuppS{i}{\ell} \subset [N]$ and states several cardinality relationships. In particular, it reveals that the children sets $\Ql{i}{\ell}$ and $\Qr{i}{\ell}$ are consistent with the recursive tree construction of Iteration~\ref{def:proofs:mw:signaltree}.
\begin{proposition} \label{prop:proofs:mw:tree_properties:Q}
	Let $\Lambda$ be the tree constructed in Iteration~\ref{def:proofs:mw:signaltree}. Then the following holds true:
	\begin{thmlist}
		
		\item\label{prop:proofs:mw:tree_properties:Q_S}
		We have that $\SuppS{i}{\ell} = \{\pl{i}{\ell}+1, \dots, \pr{i}{\ell}-1\} = \Q{i}{\ell} \setminus \{\pr{i}{\ell}\}$. 
		In particular, the following cardinality relationships hold true:
		\begin{align}
		\cardinality{\SuppS{i}{\ell}} + 1 &= \cardinality{\Q{i}{\ell}} = \cardinality{\Ql{i}{\ell}} + \cardinality{\Qr{i}{\ell}}, \\*
		\cardinality{\SuppS{j}{k}} + 1 &= \cardinality{\Ql{i}{\ell}} \quad  \text{if $(k,j) = \lchild(\ell,i)$,} \\*
		\cardinality{\SuppS{j}{k}} + 1 &= \cardinality{\Qr{i}{\ell}} \quad  \text{if $(k,j) = \rchild(\ell,i)$.}
		\end{align}
		Moreover, if $(k,j) = \lchild(\ell,i)$ and $\SuppS{j}{k} \intersec \Supp = \emptyset$, we have
		\begin{equation}\label{eq:proofs:mw:tree_properties:Q_S:card_iter}
		\frac{\cardinality{\Ql{i}{\ell}} - 1}{2} \leq \cardinality{\Ql{j}{k}}, \cardinality{\Qr{j}{k}} \leq \frac{\cardinality{\Ql{i}{\ell}} + 1}{2}.
		\end{equation}
		An analogous statement holds true if $(k,j) = \rchild(\ell,i)$.
		
		\item\label{prop:proofs:mw:tree_properties:recursive_inclusion}
		Let $(\ell,i), (k,j) \in \Lambda$ be such that $(\ell,i) \neq (k,j)$ and $k \leq \ell$. Then exactly one of the following statements holds true: either (i)  $\Q{i}{\ell} \intersec \Q{j}{k} = \emptyset$, or (ii) $\Q{i}{\ell} \subset \Ql{j}{k}$, or (iii) $\Q{i}{\ell} \subset \Qr{j}{k}$.
	\end{thmlist}
\end{proposition}
\begin{proof}
\begin{prooflist}
	\item
	The relation $\SuppS{i}{\ell} = \{\pl{i}{\ell}+1, \dots, \pr{i}{\ell}-1\}$ is easily shown by an induction, making use of the fact that $\SuppS{i}{\ell}$ does not contain any other pivot element than $\p{i}{\ell}$ which was selected in the $\ell$-th iteration. The cardinality relationships are consequences of this representation of $\SuppS{i}{\ell}$.
	The estimate of \eqref{eq:proofs:mw:tree_properties:Q_S:card_iter} follows from
	\begin{equation}
	\cardinality{\Ql{j}{k}} + \cardinality{\Qr{j}{k}} = \cardinality{\Q{j}{k}} = \cardinality{\SuppS{j}{k}} + 1 = \cardinality{\Ql{i}{\ell}}
	\end{equation}
	and that $\abs[\big]{\cardinality{\Ql{j}{k}} - \cardinality{\Qr{j}{k}}} \leq 1$, which is in turn due to
	\begin{equation}
	\p{j}{k} = \median\big(\SuppS{j}{k}\big) = \median\big(\big\{\pl{j}{k}+1, \dots, \pr{j}{k}-1\big\}\big).
	\end{equation}
	
	\item
	First, we assume that $\p{i}{\ell} < \p{j}{k}$. Let us now consider the case that $(k,j) \not\in \ceil{(\ell,i)}$: Then, the set of \emph{common ancestors} of the vertices $(\ell,i)$ and $(k,j)$ is given by $\mathfrak{C}  \coloneqq  \ceil{(\ell,i)} \intersec \ceil{(k,j)}$ and let $(\ell', i')  \coloneqq  \argmax_{(\tilde\ell,\tilde{i}) \in \mathfrak{C}} \tilde\ell$ be the \emph{last common ancestor} of these vertices. Using this notation, we observe that
	\begin{equation}
	\pr{i}{\ell} \leq \p{i'}{\ell'} \leq \pl{j}{k} < \pl{j}{k}+1,
	\end{equation}
	which implies the claim of (i). If $(k,j) \in \ceil{(\ell,i)}$, it follows from \eqref{eq:proofs:mw:signaltree:neighbors} that $\pl{i}{\ell} \geq \pl{j}{k}$ and $\pr{i}{\ell} \leq \p{j}{k}$. In particular, we have $\Q{i}{\ell} \subset \Ql{j}{k}$ by definition, which implies the claim of (ii).
	
	If $\p{i}{\ell} > \p{j}{k}$, the claims of (i) and (iii) follow by a similar line of argument, respectively. The claims of (i)--(iii) are obviously mutually exclusive. \qedhere
\end{prooflist}
\end{proof}

\subsubsection{Step 2(b): Non-Dyadic Haar Matrix}
\label{eq:proof:step2b}

Based on the tree construction of the previous substep, we now design an orthogonal matrix $\H \in \R^{n \times n}$ that has a favorable structure with regard to the box-constrained optimization problem in \eqref{eq:proofs:mw:boxconstr}. For this purpose, we introduce the constants
\begin{equation}\label{eq:proofs:mw:d-parameters}
\d{i}{\ell}  \coloneqq  \sqrt{\frac{\nl{i}{\ell} + \nr{i}{\ell}}{\nl{i}{\ell} \cdot \nr{i}{\ell}}}, \qquad
\dl{i}{\ell}  \coloneqq  \frac{\nr{i}{\ell}}{\nl{i}{\ell} + \nr{i}{\ell}}, \qquad
\dr{i}{\ell}  \coloneqq  \frac{\nl{i}{\ell}}{\nl{i}{\ell} + \nr{i}{\ell}},
\end{equation}
for $(\ell,i) \in \Lambda$, where $\nl{i}{\ell}  \coloneqq  \cardinality{\Ql{i}{\ell}} \geq 1$ and $\nr{i}{\ell}  \coloneqq  \cardinality{ \Qr{i}{\ell} } \geq 1$. In particular, note that $\dl{i}{\ell}, \dr{i}{\ell} \in \intvop{0}{1}$ and $\dl{i}{\ell} + \dr{i}{\ell} = 1$.

Moreover, we set $p((0,0)) \coloneqq \p{0}{0}  \coloneqq  n$. Then, $p \colon \Lambda \union \{(0,0)\} \to [n]$ is a bijection, and we may use $\Lambda \union \{(0,0)\}$ as new indexing set for the rows of $\H$, i.e., the $\p{i}{\ell}$-th row of $\H$ is denoted by $\h{i}{\ell}$ for all $(\ell,i) \in \Lambda \union \{(0,0)\}$.

The following construction of a \emph{non-dyadic Haar matrix} is based on the work of \citeauthor{gupta_non-dyadic_2010} \ \cite{gupta_non-dyadic_2010}. Of particular relevance to our problem setup is the fact that $\H \TV^\T$ has only a few non-zero entries that are explicitly known. More specifically, we will exploit the tree-like structure of $\H \TV^\T$ later in Step~\hyperref[eq:proof:step3a]{3}, in order to decouple the dependencies in \eqref{eq:proofs:mw:boxconstr} and thereby to iteratively construct an appropriate ansatz.

\revision{
\begin{example}[Dyadic signals III]
\label{ex:dyadic_3}
 The current proof step simplifies considerably in the case of a dyadic signal. Indeed, the previously introduced constants are then given by $\d{i}{\ell} = \sqrt{2^{\ell-L+1}}$ and $\dl{i}{\ell} = \dr{i}{\ell} = \tfrac{1}{2}$. In particular, the non-dyadic Haar matrix in Proposition~\ref{prop:proofs:mw:haarofbinary} below becomes the standard Haar matrix. We refer to Figure~\ref{fig:roadmap} for a visualization of the tree structure induced by $\H \TV^\T$ in the dyadic case. 
\end{example}
}

\begin{proposition} \label{prop:proofs:mw:haarofbinary}
	Let $\Lambda$ be the signal-dependent tree constructed according to Iteration~\ref{def:proofs:mw:signaltree}. Then, there exists an orthogonal matrix $\H \in \R^{n \times n}$ that satisfies the following: 
	
	\begin{thmlist}
		\item \label{prop:proofs:mw:haarofbinary:zero} We have that $\h{0}{0} \TV^\T = \vnull^\T$.
		
		\item \label{prop:proofs:mw:haarofbinary:tree} For every $(\ell, i) \in \Lambda$, we have that
		\begin{equation} \label{eq:proofs:mw:haarofbinary:htv_p}
		[\h{i}{\ell} \TV^\T]_{\p{i}{\ell}} = - \d{i}{\ell}.
		\end{equation}
		Furthermore, if $\pl{i}{\ell} > 0$, then 
		\begin{equation} \label{eq:proofs:mw:haarofbinary:htv_pl}
		[\h{i}{\ell} \TV^\T]_{\pl{i}{\ell}} = \d{i}{\ell} \cdot \dl{i}{\ell},
		\end{equation}
		and if $\pr{i}{\ell} < n$, then
		\begin{equation} \label{eq:proofs:mw:haarofbinary:htv_pr}
		[\h{i}{\ell} \TV^\T]_{\pr{i}{\ell}} = \d{i}{\ell} \cdot \dr{i}{\ell}.
		\end{equation}
		Finally, if $\jump \not \in \big\{ \p{i}{\ell}, \pl{i}{\ell}, \pr{i}{\ell} \big\}$, we have that
		\begin{equation}
		[\h{i}{\ell} \TV^\T]_\jump = 0.
		\end{equation}
	\end{thmlist}
\end{proposition}
\begin{proof}
	We define the row of $\H$ according to the index set above. More precisely, let
	\begin{equation}
	\h{0}{0} \coloneqq \frac{1}{\sqrt{n}}\1_n ,
	\end{equation}
	and for $(\ell,i) \in \Lambda$, let
	\begin{equation} \label{eq:proofs:mw:haarofbinary:defH}
	[\h{i}{\ell}]_\jump \coloneqq \d{i}{\ell} \cdot \begin{cases}
	+ \dl{i}{\ell}, & \text{if } \jump \in \Ql{i}{\ell},\\
	- \dr{i}{\ell}, & \text{if } \jump \in \Qr{i}{\ell}, \\
	0, & \text{otherwise}.
	\end{cases}
	\end{equation}
	
	The claim of \ref{prop:proofs:mw:haarofbinary:zero} is obvious, since $\h{0}{0}$ is a constant vector. To prove the claims of \ref{prop:proofs:mw:haarofbinary:tree}, we note that $\h{i}{\ell} \TV^\T = (\TV (\h{i}{\ell})^\T)^\T$, i.e., the finite difference operator applied to the rows of $\H$. Recalling the definition of $\Ql{i}{\ell}$ and $\Qr{i}{\ell}$, it is not hard to see that the construction of \eqref{eq:proofs:mw:haarofbinary:defH} implies $\supp(\h{i}{\ell} \TV^\T) \subset \{\p{i}{\ell}, \pl{i}{\ell}, \pr{i}{\ell}\}$. The jump discontinuity at $\p{i}{\ell}$ amounts to
	\begin{equation}
	[\h{i}{\ell} \TV^\T]_{\p{i}{\ell}} = \d{i}{\ell} \cdot (- \dr{i}{\ell} - \dl{i}{\ell}) = -\d{i}{\ell},
	\end{equation}
	which shows \eqref{eq:proofs:mw:haarofbinary:htv_p}. The proof of \eqref{eq:proofs:mw:haarofbinary:htv_pl} and \eqref{eq:proofs:mw:haarofbinary:htv_pr} works similarly.
	
	It remains to show that $\H$ is orthogonal. To this end, let $(\ell,i), (k,j) \in \Lambda$. If $(\ell,i) = (k,j)$, we have
	\begin{align}
		\sp{\h{i}{\ell}}{\h{j}{k}}
		&= \big(\d{i}{\ell}\big)^2 \cdot \Big[\nl{i}{\ell} \big(\dl{i}{\ell}\big)^2 + \nr{i}{\ell} \big(\dr{i}{\ell}\big)^2\Big] \\*
		&= \frac{\nl{i}{\ell} + \nr{i}{\ell}}{\nl{i}{\ell} \cdot \nr{i}{\ell}} \cdot
		\frac{\nl{i}{\ell} \big(\nr{i}{\ell}\big)^2 + \nr{i}{\ell} \big(\nl{i}{\ell}\big)^2}{\big(\nl{i}{\ell} + \nr{i}{\ell}\big)^2} = 1.
	\end{align}
	Finally, let $(\ell,i) \neq (k,j)$ and assume without loss of generality that $k \leq \ell$. Then Proposition~\ref{prop:proofs:mw:tree_properties:Q}\ref{prop:proofs:mw:tree_properties:recursive_inclusion} yields that either $\supp(\h{i}{\ell}) \intersec \supp(\h{j}{k}) = \emptyset$, implying $\sp{ \h{i}{\ell}}{ \h{j}{k} } = 0$, or that $\supp(\h{i}{\ell})$ is entirely contained in $\Ql{j}{k}$ or $\Qr{j}{k}$. Since $\h{j}{k}$ takes a constant value when restricted to each of these sets, we have that
	\begin{equation}
		\sp{ \h{i}{\ell}}{ \h{j}{k} } = c \cdot \d{i}{\ell} \cdot \Big(\nl{i}{\ell} \dl{i}{\ell} - \nr{i}{\ell} \dr{i}{\ell} \Big) = 0,
	\end{equation}
	where $c \in \R$ denotes the constant value that $\h{j}{k}$ takes on $\Ql{j}{k}$ or $\Qr{j}{k}$, respectively. 
	In a very similar way, one can show that $\lnorm{\h{0}{0}} = 1$ and $\sp{\h{i}{\ell}}{\h{0}{0}} = 0$ for all $(\ell,i) \in \Lambda$.
\end{proof}

\subsubsection{Step 3(a): Construction of a Dual Vector}
\label{eq:proof:step3a}

In this substep, let $\gaussian \in \R^n$ and $\tau > 0$ be fixed. Let us also recall the least-squares problem \eqref{eq:proofs:mw:boxconstr} with the box-constraint set $\Fext$ defined in \eqref{eq:proofs:mw:feasibleset}. Using the construction of $\H$ from Proposition~\ref{prop:proofs:mw:haarofbinary} and the identification between $\Lambda$ and $[N]$, we can now easily rephrase \eqref{eq:proofs:mw:boxconstr} as follows:
\begin{equation}\label{eq:proofs:mw:boxconstr_simpl}
\inf_{\dv \in \Fext} \lnorm{\gaussian - \tau \H \TV^\T \dv}^2 = g_{\p{0}{0}}^2 + \inf_{\dv \in \Fext} \sum_{(\ell,i) \in \Lambda} \Big[g_{\p{i}{\ell}} - \tau \d{i}{\ell} \big(\dl{i}{\ell} {w}_{\pl{i}{\ell}} + \dr{i}{\ell} {w}_{\pr{i}{\ell}} -  {w}_{\p{i}{\ell}}\big)\Big]^2,
\end{equation}
where we have set ${w}_0 \coloneqq {w}_n \coloneqq 0$.
Based on the formulation on the right-hand side, we now construct a feasible dual vector $\bar\dv \in \Fext$ (depending on $\gaussian$ and $\tau$) by iterating over the pivot elements level by level $(\Lambda \intersec \lambda_{1}, \Lambda \intersec \lambda_{2},  \Lambda \intersec \lambda_{3}, \dots)$. 
In this context, we will consider the subsets
\begin{align}
\Lambda^\circ
& \coloneqq  \big\{(\ell, i) \in \Lambda \setminus \Lambda^{L_0} \suchthat p^\mo \big( \pl{i}{\ell} \big) \not\in \Lambda^{L_0} \text{ and } p^\mo \big( \pr{i}{\ell} \big)  \not\in \Lambda^{L_0} \big\}, \\*
\Lambda^\natural
& \coloneqq  \big\{(\ell, i) \in \Lambda \setminus \Lambda^{L_0} \suchthat p^\mo \big( \pl{i}{\ell} \big) \in \Lambda^{L_0} \text{ or } p^\mo \big( \pr{i}{\ell} \big) \in \Lambda^{L_0} \big\},
\end{align}
so that $\Lambda = \Lambda^{L_0} \disjunion \Lambda^\circ \disjunion \Lambda^\natural$ (the bijection $p \colon \Lambda \to [N]$ was introduced in Proposition~\ref{prop:proofs:mw:tree_properties}\ref{prop:proofs:mw:tree_properties:bijective}).
\revision{A schematic visualization of this decomposition is shown Figure~\ref{fig:tree}.}
The cardinalities of these index sets can be controlled as follows:\enlargethispage{-\baselineskip}
\begin{lemma} \label{lem:proofs:mw:cardinality}
	The following cardinality bounds hold true:
	\begin{thmproperties}
	\item
		$\cardinality{\Lambda^{L_0}} = \sext$,
	\item
		$\cardinality{\Lambda^\circ} \leq n - \sext$,
	\item
		$\cardinality{\Lambda^\natural} \leq 4\sext \log_2(n-\sext)$.
	\end{thmproperties}
\end{lemma}
\begin{proof}
	The bounds of (i) and (ii) follow directly from the construction of $\Lambda^\circ$ and $\Lambda^{L_0}$. It remains to verify (iii). To that end, we note that every pivot element $\p{i}{\ell}$ with $(\ell,i) \in \Lambda^{L_0}$ can only be the left neighbor of no more than one vertex per layer $\ell_0 \geq L_0$: towards a contradiction, assume that $\p{i}{\ell}$ is the left neighbor of two vertices $(\ell_0, i_1), (\ell_0, i_2)$ in the same layer $\ell_0 \geq L_0$. Then $\p{i_1}{\ell_0}$ and $\p{i_2}{\ell_0}$ are both larger than $\p{i}{\ell}$ and they have a common ancestor $(k,j) \neq (\ell,i)$ with $\p{j}{k}>\p{i}{\ell}$. Assume without loss of generality that $\p{i_1}{\ell_0} > \p{i_2}{\ell_0}$. Then $\p{i}{\ell} < \p{j}{k} \leq \pl{i_1}{\ell_0}$, which is a contradiction. With a similar argument, one can show that every pivot element $\p{i}{\ell}$ with $(\ell,i) \in \Lambda^{L_0}$ can only be the right neighbor of no more than one vertex per layer $\ell_0 \geq L_0$.
	
	Consequently, in each of the remaining layers, there can be at most $\sext = \cardinality{\Lambda^{L_0}}$ elements with left (right) neighbors in $\Lambda^{L_0}$; note that $L_0$-th layer may also contain elements from $\Lambda^\natural$.
	According to Proposition~\ref{prop:proofs:mw:tree_properties}\ref{prop:proofs:mw:tree_properties:depth}, the number of the remaining layers after $\Lambda^{L_0}$ is bounded by $\ceil{\log_2(n-\sext)}$.
	Therefore, we obtain
	\begin{equation}
		\cardinality{\Lambda^\natural} \leq 2\sext \big(\log_2(n-\sext) + 1 \big) \leq 4\sext \log_2(n-\sext).
	\end{equation}
\end{proof}

\revision{
\begin{example}[Dyadic signals IV]
 The previous proof reveals that for a dyadic signal, we can obtain the sharper bound $\cardinality{\Lambda^\natural} \leq 2\sext \log_2(n / \sext)$, since there are only $\log_2(n / (\sext+1))$ layers remaining (see also Figure~\ref{fig:tree}). 
\end{example}
}

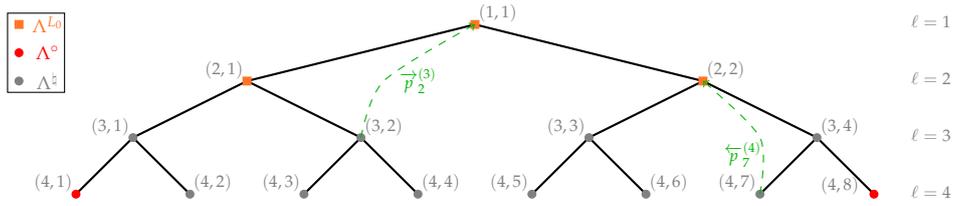
\begin{figure}
 \centering
 \begin{tikzpicture}[thick,scale=0.75, every node/.style={scale=0.7}]
  
  \definecolor{myo}{rgb}{1,0.46,0.15}
  \definecolor{myb}{rgb}{0.35,0.46,0.8}
  \definecolor{myb2}{rgb}{0.81,0.84,0.94}
  \definecolor{Sset}{gray}{0.5}
  \definecolor{myt}{rgb}{0.95,0.75,0.3}
  \definecolor{Sset2}{rgb}{0.35,0.46,0.8}
  
  \tikzstyle{signaldots}=[mark size=1.6pt,color=myb,rotate=-90]
  \tikzstyle{jumpdots}=[mark size=2pt,color=gray,rotate=-90]
  \tikzstyle{jumpdotsS}=[mark size=2pt,color=myo,rotate=-90] 
  \tikzstyle{treeS}=[mark size=2pt,color=myo,rotate=-90]
  \tikzstyle{treeSc}=[mark size=2pt,color=gray,rotate=-90]
  \tikzstyle{treecirc}=[mark size=2pt,color=red,rotate=-90]
  
  \draw (4,4) -- (8,5);
  \draw (12,4) -- (8,5);
  \draw (12,4) -- (14,3);
  \draw (12,4) -- (10,3);
  \draw (4,4) -- (2,3);
  \draw (4,4) -- (6,3);
  \draw (2,3) -- (1,2);
  \draw (2,3) -- (3,2);
  \draw (6,3) -- (7,2);
  \draw (6,3) -- (5,2);
  \draw (10,3) -- (9,2);
  \draw (10,3) -- (11,2);
  \draw (14,3) -- (15,2);
  \draw (14,3) -- (13,2);

  \node[treeS] at (8,5) {\pgfuseplotmark{square*}};
  \node[color=gray] at (8.4,5.2) {\small $(1,1)$}; 
  \node[treeS] at (4,4) {\pgfuseplotmark{square*}};
  \node[color=gray] at (3.6,4.2) {\small $(2,1)$}; 
  \node[treeS] at (12,4) {\pgfuseplotmark{square*}};
  \node[color=gray] at (12.4,4.2) {\small $(2,2)$};   
  \node[color=gray] at (14,3) {\pgfuseplotmark{*}};
  \node[color=gray] at (14.4,3.2) {\small $(3,4)$}; 
  \node[treeSc] at (2,3) {\pgfuseplotmark{*}};
  \node[color=gray] at (1.6,3.2) {\small $(3,1)$};
  \node[treeSc] at (6,3) {\pgfuseplotmark{*}};
  \node[color=gray] at (6.4,3.2) {\small $(3,2)$};
  \node[treecirc] at (1,2) {\pgfuseplotmark{*}};
  \node[color=gray] at (0.6,2.2) {\small $(4,1)$};
  \node[treeSc] at (3,2) {\pgfuseplotmark{*}};
  \node[color=gray] at (3.4,2.2) {\small $(4,2)$};
  \node[treeSc] at (5,2) {\pgfuseplotmark{*}};
  \node[color=gray] at (4.6,2.2) {\small $(4,3)$};
  \node[treeSc] at (7,2) {\pgfuseplotmark{*}};
  \node[color=gray] at (7.4,2.2) {\small $(4,4)$};
  \node[treeSc] at (11,2) {\pgfuseplotmark{*}};
  \node[color=gray] at (11.4,2.2) {\small $(4,6)$}; 
  \node[treeSc] at (13,2) {\pgfuseplotmark{*}};
  \node[color=gray] at (12.6,2.2) {\small $(4,7)$}; 
  \node[treecirc] at (15,2) {\pgfuseplotmark{*}};
  \node[color=gray] at (14.4,2.15) {\small $(4,8)$}; 
  \node[treeSc] at (10,3) {\pgfuseplotmark{*}};
  \node[color=gray] at (9.6,3.2) {\small $(3,3)$};
  \node[treeSc] at (9,2) {\pgfuseplotmark{*}};
  \node[color=gray] at (8.6,2.2) {\small $(4,5)$}; 
  
  \node[color=gray] at (16,5.05) {\small $\ell = 1$};
  \node[color=gray] at (16,4.05) {\small $\ell = 2$};
  \node[color=gray] at (16,3.05) {\small $\ell = 3$};
  \node[color=gray] at (16,2.05) {\small $\ell = 4$};

  \draw[black!30!green,->,thin,dashed] plot [smooth] coordinates {(6,3) (6.4,4.0) (8,5)};
  \node[color=black!30!green] at (7,4.0) {\small $\pr{2}{3}$};
  
  \draw[black!30!green,->,thin,dashed] plot [smooth] coordinates {(13,2) (13.0,3.0) (12,4)};
  \node[color=black!30!green] at (12.7,2.7) {\small $\pl{7}{4}$};
  
  \node[mark size=2pt,color=myo,rotate=-90] at (0,5) {\pgfuseplotmark{square*}};
  \node[color=myo] at (0.5,5) {$\Lambda^{L_0}$};
  
  \node[mark size=2pt,color=gray,rotate=-90] at (0,4) {\pgfuseplotmark{*}};
  \node[color=gray] at (0.5,4) {$\Lambda^{\natural}$};
  
  \node[mark size=2pt,color=red,rotate=-90] at (0,4.5) {\pgfuseplotmark{*}};
  \node[color=red] at (0.5,4.5) {$\Lambda^{\circ}$};
  
  \draw[thin] (-0.2,5.2) rectangle (0.8,3.8);

  \end{tikzpicture}
  \caption{\revision{\textbf{Decomposition of a signal-dependent tree.} The decomposition of the tree $\Lambda$ into its subsets ${\Lambda^{L_0} \disjunion \Lambda^\circ \disjunion \Lambda^\natural}$ is visualized for the dyadic signal from Figure~\ref{fig:trees:1}.}}
  \label{fig:tree}
\end{figure}

The specific form of \eqref{eq:proofs:mw:boxconstr_simpl} enables us to select the elements of $\bar\dv \in \Fext$ iteratively along the levels of $\Lambda$. 
While the entries corresponding to $\Lambda^{L_0}$ are fixed in advance by the sign pattern of~$\TV \grtr$, we need to make a case distinction between the sets $\Lambda^\circ$ and $\Lambda^\natural$.
In particular, for $(\ell,i) \in \Lambda^\circ$, we optimize the corresponding term in \eqref{eq:proofs:mw:boxconstr_simpl} over ${w}_{\p{i}{\ell}}$, using that ${w}_{\pl{i}{\ell}}$ and ${w}_{\pr{i}{\ell}}$ were both already selected in previous iterations. In order to ensure feasibility of the resulting dual vector~$\bar\dv$, the search space for each variable ${w}_{\p{i}{\ell}}$ needs to be slightly truncated in this procedure.
This construction is made precise by the following proposition. 
The nested structure of the truncated intervals introduced below is visualized in Figure~\ref{fig:proofs:mw:dualvec}.
\begin{proposition} \label{prop:proofs:mw:dualvec}
	Let $\gaussian = (g_1, \dots, g_n) \in \R^n$ and $\tau > 0$ be fixed. For $(\ell, i) \in \Lambda^{L_0}$, set
	\begin{equation}
	\bar{w}_{\p{i}{\ell}}  \coloneqq  [\sign(\TV \grtr)]_{\p{i}{\ell}}.
	\end{equation}
	Moreover, let $\bar{w}_0 \coloneqq \bar{w}_n \coloneqq 0$.
	Then, the choice
	\begin{equation}
	\bar{w}_{\p{i}{\ell}}  \coloneqq 
	\begin{cases}
	\clip[\Big]{\dl{i}{\ell} \bar{w}_{\pl{i}{\ell}} + \dr{i}{\ell} \bar{w}_{\pr{i}{\ell}}}{1 - \sqrt{2^{L_0-\ell}}}, & \text{if } (\ell, i) \in \Lambda^\natural, \\
	\argmin_{z \in I_i^{(\ell)}} \Big( g_{\p{i}{\ell}} - \tau \d{i}{\ell} \big( \dl{i}{\ell} \bar{w}_{\pl{i}{\ell}} + \dr{i}{\ell} \bar{w}_{\pr{i}{\ell}} - z \big) \Big)^2, & \text{if } (\ell, i) \in \Lambda^\circ,
	\end{cases}
	\end{equation}
	with
	\begin{equation}
	I_i^{(\ell)} \coloneqq \underbrace{ \intvcl{ -\gamma  \sqrt{2^{L_0 - \ell}}}{ \gamma \sqrt{2^{L_0 - \ell}} }}_{ \eqqcolon  I^{(\ell)}} {} + {} \big( \dl{i}{\ell} \bar{w}_{\pl{i}{\ell}} + \dr{i}{\ell} \bar{w}_{\pr{i}{\ell}} \big)
	\end{equation}
	and $\gamma  \coloneqq  \sqrt{2} - 1$ is well-defined.\footnote{More precisely, the entries of $\bar\dv$ are defined iteratively along the tree levels $\Lambda \intersec \lambda_{1}, \Lambda \intersec \lambda_{2},  \Lambda \intersec \lambda_{3}, \dots$; in particular, $\bar{w}_{\p{i}{\ell}}$ only depends on $\bar{w}_{\pl{i}{\ell}}$ and $\bar{w}_{\pr{i}{\ell}}$, which were both defined in previous layers.} In particular, we have that
	\begin{equation} \label{eq:proofs:mw:dualvec:wbox}
	\bar{w}_{\p{i}{\ell}} \in \underbrace{\intvcl{  -1 + \sqrt{2^{L_0 - \ell}}}{ 1 - \sqrt{2^{L_0 - \ell}} }}_{ \eqqcolon  B^{(\ell)}}
	\end{equation}
	for all $(\ell, i) \in \Lambda^\circ \disjunion \Lambda^\natural$, and therefore $\bar\dv \in \Fext$.
\end{proposition}
\enlargethispage{1.5\baselineskip}
\begin{proof}
	Clearly, $\bar{w}_{\p{i}{\ell}}$ is well-defined for $(\ell, i) \in \Lambda^{L_0}$. The second part of the claim follows by induction over $\ell=L_0,\dots,L$. Regarding the induction basis, we note that $\lambda_{L_0} \subset \Lambda^{L_0} \union \Lambda^\natural$, so that all values are well-defined and \eqref{eq:proofs:mw:dualvec:wbox} clearly holds on $\lambda_{L_0}$.
	
	Now, let $\ell \in \{L_0, \dots, L-1\}$. For the induction hypothesis, assume that the entries of $\bar\dv$ are well-defined on $\lambda_k$ for all $k = L_0,\dots,\ell$ and that \eqref{eq:proofs:mw:dualvec:wbox} holds true on these index sets. Then, $\bar\dv$ is also well-defined on $\lambda_{\ell+1}$, since all corresponding entries of $\bar\dv$ do only depend on previously defined entries of $\bar\dv$. Thus, it remains to show that \eqref{eq:proofs:mw:dualvec:wbox} is satisfied, i.e., $\bar{w}_{\p{i}{\ell+1}} \in B^{(\ell+1)}$ for all $(\ell+1, i) \in \lambda_{\ell+1}$. To see this, we assume without loss of generality that $(\ell+1, i) \in \Lambda^\circ$, since otherwise, the constraint \eqref{eq:proofs:mw:dualvec:wbox} is fulfilled by definition. Note that $B^{(k)}$ is a symmetric convex set for all $k = L_0, \dots, L$ and $B^{(k_1)} \subset B^{(k_2)}$ for $k_1 \leq k_2$. Due to the induction hypothesis and $\dl{i}{\ell+1} + \dr{i}{\ell+1} = 1$, we have that
	\begin{equation}\label{eq:proofs:mw:dualvec:B-ell}
	\dl{i}{\ell+1} \bar{w}_{\pl{i}{\ell+1}} + \dr{i}{\ell+1} \bar{w}_{\pr{i}{\ell+1}} \in B^{(\ell)}.
	\end{equation}
	Moreover,
	\begin{align}
	I^{(\ell+1)} + B^{(\ell)}
	&= \intvcl{ -1 + \sqrt{2^{L_0-\ell}} - ( \sqrt{2} - 1 ) \sqrt{2^{L_0-(\ell+1)}}}{ 1 - \sqrt{2^{L_0-\ell}} + ( \sqrt{2} - 1 ) \sqrt{2^{L_0-(\ell+1)}} } \\*
	&= \intvcl{ -1 + \sqrt{2^{L_0-(\ell+1)}}}{ 1 -  \sqrt{2^{L_0-(\ell+1)}} } = B^{(\ell+1)},
	\end{align}
	and therefore, we have $\bar{w}_{\p{i}{\ell+1}} \in I_i^{(\ell+1)} \subset I^{(\ell+1)} + B^{(\ell)} = B^{(\ell + 1)}$. From this, $\bar\dv \in \Fext$ follows immediately.
\end{proof}

\begin{figure}[t]
	\centering
	\includegraphics[width=0.7\linewidth]{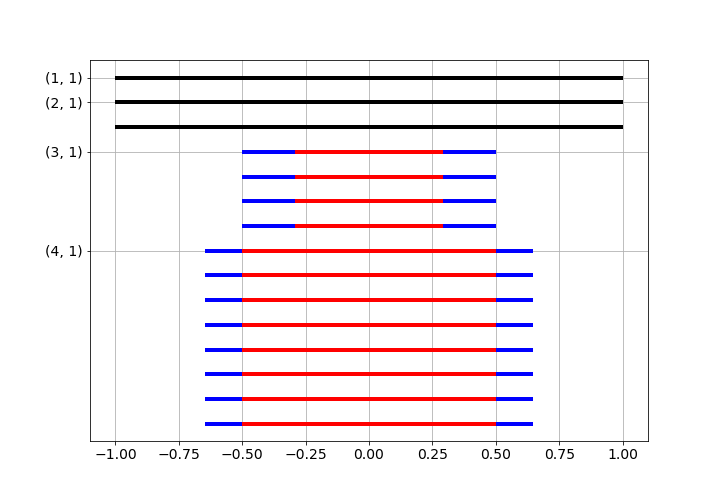}
	\caption{\textbf{The nested structure of the truncated intervals in Proposition~\ref{prop:proofs:mw:dualvec} for $L_0 = 2$ and $\sext = 3$.} The vertical axis denotes the different vertices of $\Lambda$ where the first vertex in each layer is labeled explicitly. The inner (red) interval is $B^{(\ell)}$, the outer (blue) interval is $I^{(\ell+1)}$. The black bars in the first two layers correspond to the entries of the dual vector on $\Lambda^{L_0}$, which certainly belong to the full unit interval.}
	\label{fig:proofs:mw:dualvec}
\end{figure}

\subsubsection{Step 3(b): Bounding the Conic Mean Width}
\label{eq:proof:step3b}

In this step, we establish a general bound for the conic mean width $\meanwidth[\conic]{\descset{\lnorm{\TV(\cdot)}[1], \grtr}}$ based on the estimate of Lemma~\ref{lem:proofs:mw:unitarypolar}.
Here, the choice of $\H$ follows the non-dyadic Haar wavelet transform from Proposition~\ref{prop:proofs:mw:haarofbinary}.
In order to invoke the bound of Lemma~\ref{lem:proofs:mw:unitarypolar}, let $\tau > 0$ be fixed for now and let $\gaussian \distributed \Normdistr{\vnull}{\I{n}}$. 

First, we proceed with the simplification of \eqref{eq:proofs:mw:boxconstr_simpl} and plug in the dual vector $\bar\dv \in \Fext$ constructed in Proposition~\ref{prop:proofs:mw:dualvec}. This leads to the following bound for the right-hand side of \eqref{eq:proofs:mw:unitarypolar}:
\begin{align} 
& \mean{}_{\gaussian} [ \inf_{\dv \in \Fext} \lnorm{\gaussian - \tau \H \TV^\T \dv}^2 ] \\*
= {} & \mean{}_{\gaussian} \Big[ g_{\p{0}{0}}^2 + \inf_{\dv \in \Fext} \sum_{(\ell,i) \in \Lambda} \Big(g_{\p{i}{\ell}} - \tau \d{i}{\ell} \big(\dl{i}{\ell} {w}_{\pl{i}{\ell}} + \dr{i}{\ell} {w}_{\pr{i}{\ell}} -  {w}_{\p{i}{\ell}}\big)\Big)^2 \Big] \label{eq:proofs:mw:mw_bound_pivotsum} \\
\leq {} & \underbrace{\mean{}_{\gaussian}[g_{\p{0}{0}}^2]}_{=1} + \sum_{(\ell,i) \in \Lambda} \underbrace{\mean{}_{\gaussian}\Big[\Big( g_{\p{i}{\ell}} - \tau \d{i}{\ell} \big(\dl{i}{\ell} \bar{w}_{\pl{i}{\ell}} + \dr{i}{\ell} \bar{w}_{\pr{i}{\ell}} -  \bar{w}_{\p{i}{\ell}}\big) \Big)^2 \Big]}_{\eqqcolon \e{i}{\ell}} = 1 + \sum_{(\ell,i) \in \Lambda} \e{i}{\ell}.
\end{align}
Before proceeding with individual estimates for the \emph{per-pivot errors} $\e{i}{\ell}$, we introduce the following \emph{balancing parameters}: For $(\ell, i) \in \Lambda$, let
\begin{equation}
\b{i}{\ell}  \coloneqq  \frac{\d{i}{\ell}}{c_\ell} \quad \text{with } c_\ell \coloneqq \sqrt{\frac{2^{\ell+1}}{n}},
\end{equation}
and for any subset $\Klambda \subset \Lambda$, let
\begin{equation}
\b{\min}{\Klambda}  \coloneqq  \min_{(\ell, i) \in \Klambda} \b{i}{\ell}
\qquad \text{and} \qquad
\b{\max}{\Klambda}   \coloneqq  \max_{(\ell, i) \in \Klambda} \b{i}{\ell}.
\end{equation}

\revision{
\begin{example}[Dyadic signals V]
 One can view $\b{i}{\ell}$ as a measure for the balance of the subtree with root $(\ell, i)$.
 Indeed, in the case of a dyadic signal (corresponding to a perfect tree), we simply have that $\b{i}{\ell} = 1$ for all $(\ell, i) \in \Lambda$. This follows immediately from the fact that $n = 2^L$ and $\d{i}{\ell} = \sqrt{2^{\ell-L+1}}$ (see Example~\ref{ex:dyadic_3}).
\end{example}
}

We now make a case distinction to estimate the per-pivot error terms:

\textbf{Case~1:} Let $(\ell,i) \in \Lambda^{L_0}$. Using that $\bar{w}_{\p{i}{\ell}}$, $\bar{w}_{\pl{i}{\ell}}$, and $\bar{w}_{\pr{i}{\ell}}$ are completely determined by $\sign(\TV \grtr)$, and therefore deterministic, we obtain the per-pivot error bound
\begin{align}
\e{i}{\ell}
&= \mean{}_{\gaussian}\Big[\Big( g_{\p{i}{\ell}} - \tau \d{i}{\ell} \big(\dl{i}{\ell} \bar{w}_{\pl{i}{\ell}} + \dr{i}{\ell} \bar{w}_{\pr{i}{\ell}} -  \bar{w}_{\p{i}{\ell}}\big)\Big)^2\Big] \\*
&= \begin{multlined}[t]
\mean{}_{\gaussian}[ g_{\p{i}{\ell}}^2 ] - \mean{}_{\gaussian}[ g_{\p{i}{\ell}} ] \cdot \tau \d{i}{\ell} \big(\dl{i}{\ell} \bar{w}_{\pl{i}{\ell}} + \dr{i}{\ell} \bar{w}_{\pr{i}{\ell}} -  \bar{w}_{\p{i}{\ell}}\big) \\*
+ \Big( \tau \d{i}{\ell} \underbrace{\big(\dl{i}{\ell} \bar{w}_{\pl{i}{\ell}} + \dr{i}{\ell} \bar{w}_{\pr{i}{\ell}} -  \bar{w}_{\p{i}{\ell}}\big)}_{\stackrel{(\ast)}{\leq} 2}\Big)^2
\end{multlined} \\
&\leq 1 + (2 \tau \d{i}{\ell})^2 \\
&= 1 + (2 \tau \b{i}{\ell} c_\ell)^2 \\
&\leq 1 + (2 \tau \b{i}{\ell} c_{L_0})^2 \\
&\leq 1 + (2 \tau \b{\max}{\Lambda^{L_0}} c_{L_0})^2,
\end{align}
where we have used that $\ell \leq L_0$ in the second last line and $(\ast)$ is due to $\dl{i}{\ell} + \dr{i}{\ell} = 1$ and $\bar{w}_{\p{i}{\ell}}, \bar{w}_{\pl{i}{\ell}}, \bar{w}_{\pr{i}{\ell}} \in \intvcl{-1}{1}$.

\textbf{Case~2:} Let $(\ell, i) \in \Lambda^\natural$. We note that $\bar{w}_{\pr{i}{\ell}}$ and $\bar{w}_{\pl{i}{\ell}}$ are independent of $g_{\p{i}{\ell}}$, due to the iterative construction of $\bar\dv$ from Proposition~\ref{prop:proofs:mw:dualvec}, and therefore, $\bar{w}_{\p{i}{\ell}}$ is also independent of $g_{\p{i}{\ell}}$. Using the definition of $\bar{w}_{\p{i}{\ell}}$, we then obtain the following per-pivot error bound:
\begin{align}
\e{i}{\ell}
&= \mean{}_{\gaussian}\Big[ \Big( g_{\p{i}{\ell}} - \tau \d{i}{\ell} \big(\dl{i}{\ell} \bar{w}_{\pl{i}{\ell}} + \dr{i}{\ell} \bar{w}_{\pr{i}{\ell}} -  \bar{w}_{\p{i}{\ell}}\big)\Big)^2 \Big] \\*
&= \begin{multlined}[t]
\mean{}_{\gaussian}[ g_{\p{i}{\ell}}^2 ] - \mean{}_{\gaussian}[ g_{\p{i}{\ell}} ] \cdot \mean{}_{\gaussian}\Big[  \tau \d{i}{\ell} \big(\dl{i}{\ell} \bar{w}_{\pl{i}{\ell}} + \dr{i}{\ell} \bar{w}_{\pr{i}{\ell}} -  \bar{w}_{\p{i}{\ell}}\big) \Big] \\
+ \mean{}_{\gaussian}\Big[ \Big( \tau \d{i}{\ell} \big(\dl{i}{\ell} \bar{w}_{\pl{i}{\ell}} + \dr{i}{\ell} \bar{w}_{\pr{i}{\ell}} -  \bar{w}_{\p{i}{\ell}}\big) \Big)^2 \Big]
\end{multlined}\\
&\stackrel{\mathllap{\text{Prop.~\ref{prop:proofs:mw:dualvec}}}}{\leq} 1 + \big(\tau \d{i}{\ell} \sqrt{2^{L_0 - \ell}}\big)^2 \\
&= 1 + \big(\tau \b{i}{\ell} c_\ell \sqrt{2^{L_0 - \ell}}\big)^2 \\
&\leq 1 + (\tau \b{\max}{\Lambda^\natural} c_{L_0})^2.
\end{align}

\textbf{Case~3:} Let $(\ell, i) \in \Lambda^\circ$. Using the definition of $\bar{w}_{\p{i}{\ell}}$, we obtain
\begin{align}
\e{i}{\ell}
&= \mean{}_{\gaussian}\Big[ \Big( g_{\p{i}{\ell}} - \tau \d{i}{\ell} \big(\dl{i}{\ell} \bar{w}_{\pl{i}{\ell}} + \dr{i}{\ell} \bar{w}_{\pr{i}{\ell}} -  \bar{w}_{\p{i}{\ell}}\big)\Big)^2\Big] \\*
&\stackrel{\mathllap{\text{Prop.~\ref{prop:proofs:mw:dualvec}}}}{=} \mean{}_{\gaussian}\Big[ \pospart[\big]{\abs{ g_{\p{i}{\ell}} } - \gamma  \tau \d{i}{\ell} \sqrt{2^{L_0 - \ell}}}^2 \Big] \\
&\leq 2\exp\Big(-\tfrac{1}{2} \big(\gamma \tau \d{i}{\ell} \sqrt{2^{L_0 - \ell}} \big)^2\Big) \\
&= 2\exp\Big(-\tfrac{1}{2} \big(\gamma \tau \b{i}{\ell} \underbrace{c_\ell \sqrt{2^{L_0 - \ell}}}_{= c_{L_0}} \big)^2\Big) \\*
&\leq 2\exp\Big(-\tfrac{1}{2} \big(\gamma \tau \b{\min}{\Lambda^\circ} c_{L_0} \big)^2\Big),
\end{align}
where the third line follows from the standard bound $\mean{}_{g \distributed \Normdistr{0}{1}}\big[ \pospart[normal]{\abs{g} - t }^2\big] \leq 2 e^{-t^2/2}$ for all $t > 0$; cf.~\cite[Eq.~(2.15)]{tropp2014convex}.

By summing up the above per-pivot error bounds, we obtain
\begin{align}
\sum_{(\ell,i) \in \Lambda^{L_0}} \e{i}{\ell}
&\leq \cardinality{\Lambda^{L_0}} + (2 \tau \b{\max}{\Lambda^{L_0}} c_{L_0})^2 \cardinality{\Lambda^{L_0}}, \\*
\sum_{(\ell,i) \in \Lambda^\natural} \e{i}{\ell}
&\leq \cardinality{\Lambda^\natural} + (\tau \b{\max}{\Lambda^\natural} c_{L_0})^2 \cardinality{\Lambda^\natural}, \\*
\sum_{(\ell,i) \in \Lambda^\circ} \e{i}{\ell}
&\leq \cardinality{\Lambda^\circ} \cdot 2\exp\Big(-\tfrac{1}{2} \big(\gamma \tau \b{\min}{\Lambda^\circ}  c_{L_0} \big)^2\Big). \label{eq:proofs:mw:pivoterr_sums}
\end{align}
We now select a (probably suboptimal) value for $\tau > 0$ to apply the bound for the conic mean width from Lemma~\ref{lem:proofs:mw:unitarypolar}. More specifically, we set $\tau = (\gamma \b{\min}{\Lambda^\circ} c_{L_0})^{-1} \sqrt{2 \log(n/\sext)}$, in order to compensate for the dominating cardinality of $\Lambda^\circ$. This leads to
\begin{align}
& \effdim[\conic]{\descset{\lnorm{\TV(\cdot)}[1], \grtr}} \\*
\leq {} & \inf_{\tilde\tau > 0} \mean{}_{\gaussian} [ \inf_{\dv \in \Fext} \lnorm{\gaussian - \tilde\tau \H \TV^\T \dv}^2 ] \\*
\leq {} & \mean{}_{\gaussian} [ \lnorm{\gaussian - \tau \H \TV^\T \bar\dv}^2 ] \\
\stackrel{\mathllap{\eqref{eq:proofs:mw:mw_bound_pivotsum}}}{=}  {} & 1 + \sum_{(\ell,i) \in \Lambda} \e{i}{\ell} \\
\stackrel{\mathllap{\eqref{eq:proofs:mw:pivoterr_sums}}}{\leq}  {} & \begin{multlined}[t] 1 + \cardinality{\Lambda^{L_0}} + \cardinality{\Lambda^\natural} + (\tau c_{L_0})^2 \Big( \cardinality{\Lambda^\natural} \big(\b{\max}{\Lambda^\natural}\big)^2 + 4\, \cardinality{\Lambda^{L_0}} \big(\b{\max}{\Lambda^{L_0}}\big)^2 \Big) \\ + 2\, \cardinality{\Lambda^\circ} \exp\Big(-\tfrac{1}{2} \big(\gamma \tau \b{\min}{\Lambda^\circ}  c_{L_0} \big)^2\Big) \end{multlined} \\
\leq {} & 1 + \cardinality{\Lambda^{L_0}} + \cardinality{\Lambda^\natural} + \frac{2}{\big(\gamma \b{\min}{\Lambda^\circ}\big)^2} \Big( \cardinality{\Lambda^\natural} \big(\b{\max}{\Lambda^\natural}\big)^2 + 4\, \cardinality{\Lambda^{L_0}} \big(\b{\max}{\Lambda^{L_0}}\big)^2 \Big) \log(\tfrac{n}{\sext}) + 2\, \cardinality{\Lambda^\circ} \cdot \frac{\sext}{n}. \quad \label{eq:proofs:mw:mw_bound_general_nonasymp}
\end{align}
Together with the cardinality bounds from Lemma~\ref{lem:proofs:mw:cardinality} and $\log_2(\cdot) \leq \tfrac{3}{2}\log(\cdot)$, we finally obtain
\begin{align}
\effdim[\conic]{\descset{\lnorm{\TV(\cdot)}[1], \grtr}} & \leq
\begin{multlined}[t]
1 + \sext + 4 \sext \log_2(n-\sext) + 8\bigg(\frac{\b{\max}{\Lambda^\natural}}{\gamma \b{\min}{\Lambda^\circ}}\bigg)^2 \sext \log_2(n-\sext) \log(\tfrac{n}{\sext}) \\*
+ 8\bigg(\frac{\b{\max}{\Lambda^{L_0}}}{\gamma \b{\min}{\Lambda^\circ}}\bigg)^2 \sext \log(\tfrac{n}{\sext}) + 2 \cdot \frac{\sext(n-\sext)}{n}
\end{multlined}\\*
& \leq 1 + 3\sext + \bigg[6 + 8\bigg(\frac{\b{\max}{\Lambda^{L_0}}}{\gamma \b{\min}{\Lambda^\circ}}\bigg)^2\bigg] \sext \log(n) + 12\bigg(\frac{\b{\max}{\Lambda^\natural}}{\gamma \b{\min}{\Lambda^\circ}}\bigg)^2 \sext \log^2(n). \qquad \label{eq:proofs:mw:mw_bound_general}
\end{align}

\revision{
\begin{example}[Dyadic signals VI]
In the case of a dyadic signal, the previous computations simplify again. 
Here, we directly obtain that $\e{i}{\ell} \leq 1 + 8 \tau^2 (s+1)/n$ for every $(\ell,i) \in \Lambda^{L_0} \cup \Lambda^\natural$ and $\e{i}{\ell} \leq 2 \exp(-(\gamma \tau)^2 (s+1)/n)$ for $(\ell,i) \in \Lambda^\circ$. 
Summing up the per-pivor error as in~\eqref{eq:proofs:mw:pivoterr_sums} (recall that $\cardinality{\Lambda^\natural} \leq 2\sext \log_2(n / \sext)$) and choosing $\tau =  \gamma^{-1} \sqrt{\tfrac{n}{s+1}\log (n/(s+1))}$ results in the overall bound
\begin{equation}
 \effdim[\conic]{\descset{\lnorm{\TV(\cdot)}[1], \grtr}} \lesssim \s\log^2(\tfrac{2n}{\s}).
\end{equation}
Hence, the proof is already complete for dyadic signals. In fact, we have also obtained a slightly improved bound, due to the fact that $\cardinality{\Lambda^\natural}$ is better controlled for perfect binary trees. 
\end{example}
}

\subsubsection{Step 4: Controlling the Balancing Parameters}
\label{eq:proof:step4}

It remains to control the three balancing parameters $\b{\max}{\Lambda^{L_0}}$, $\b{\max}{\Lambda^\natural}$, and $\b{\min}{\Lambda^\circ}$ in \eqref{eq:proofs:mw:mw_bound_general}.
For this purpose, we need to make the choice of the extended gradient support $\Supp \supset \ssupp$ explicit.
The basic idea behind our construction is to select $\Supp$ in such a way that the resulting tree $\Lambda$ becomes as ``balanced'' as possible and $\Lambda^{L_0}$ becomes a perfect binary tree.
Unfortunately, this construction turns out to be quite technical and is therefore divided into the substeps~\hyperref[eq:proof:step4a]{(a)}--\hyperref[eq:proof:step4f]{(f)} below.

Before presenting the details, let us provide a short overview of our strategy: Intuitively, one can think of the actual gradient support $\ssupp \subset [N]$ as a partition of $[n]$ into $\s + 1$ elements. Our goal is now to refine $\ssupp$ by adding more (``ghost'') elements such that the resulting partition $\Supp$ becomes almost equidistant (but may contain more elements than $\ssupp$).
Importantly, such a refinement can be achieved with $\sext = \cardinality{\Supp} \asymp \s / \SC$ elements; see substep~\hyperref[eq:proof:step4a]{(a)}. In fact, the precise choice of $\Supp$ relies on the $\SC$-separation property of $\grtr$ in conjunction with a real-valued auxiliary grid $\SuppEq \subset \R$, introduced in substep~\hyperref[eq:proof:step4b]{(b)}; this grid forms an equidistant partition of the interval $\intvop{0}{n}$ and is used in substep~\hyperref[eq:proof:step4c]{(c)} to specify those (rounded) grid elements that need to be added to $\ssupp$. The almost-equidistance property of the resulting set $\Supp$ is then verified in substep~\hyperref[eq:proof:step4d]{(d)}.
Finally, we relate the pairwise distances between the elements of $\Supp$ to the cardinality parameters $\nl{i}{\ell}$ and $\nr{i}{\ell}$ defined at the beginning of Step~\hyperref[eq:proof:step2b]{2(b)}; see \eqref{eq:proofs:mw:step4:n_control}. This allows us to show that $\b{\max}{\Lambda^{L_0}} \lesssim 1$ in substep~\hyperref[eq:proof:step4e]{(e)}, as well as $\b{\max}{\Lambda^\natural} \lesssim 1$ and $\b{\min}{\Lambda^\circ} \gtrsim 1$ in substep~\hyperref[eq:proof:step4f]{(f)}.

We now elaborate all these substeps in detail:
\begin{prooflist}
\item \label{eq:proof:step4a}
	We intend to design an extended support set $\Supp = \{\jumpext_1, \dots, \jumpext_{\sext}\} \subset [N]$ with $\jumpext_1 < \dots < \jumpext_{\sext}$ and such that
	\begin{equation}
	\sext \coloneqq 2^{\ceil{\log_2(\tfrac{\s + 1}{\SC})}} - 1.
	\end{equation}
	First, we observe that the actual gradient-sparsity $\s$ can be controlled in terms of $\sext$ and the separation constant $\SC$:
	\begin{equation}\label{eq:proofs:mw:step4:sext_control}
	\sext + 1 \geq 2^{\log_2(\tfrac{\s + 1}{\SC})} = \frac{\s + 1}{\SC},
	\end{equation}
	and
	\begin{equation}\label{eq:proofs:mw:step4:sext_control_upper}
	\sext + 1 = 2^{\ceil{\log_2(\tfrac{\s + 1}{\SC})}} \leq 2^{\log_2(2\cdot \tfrac{\s + 1}{\SC})} = 2\cdot \frac{\s + 1}{\SC} \leq \frac{4\s}{\SC}.
	\end{equation}
\item \label{eq:proof:step4b}
	In order to make the choice of $\jumpext_1, \dots, \jumpext_{\sext}$ explicit, we introduce equidistant auxiliary points
	\begin{equation}
	\tilde{\jumpext}_i \coloneqq i \cdot \frac{n}{\sext+1}, \quad i = 1, \dots, \sext,
	\end{equation} 
	and we set $\SuppEq \coloneqq \{\tilde{\jumpext}_1, \dots, \tilde{\jumpext}_{\sext}\} \subset \R$; note that these auxiliary points are not necessarily integers.
	Moreover, let $\jump_1,\dots, \jump_s \in [N]$ denote the $\s$ elements of the gradient support, i.e., $\ssupp = \{\jump_1,\dots, \jump_s\}$. 
	
	Then, for every $j \in [\s]$, we select\footnote{If the choice of $i_j$ is not unique, which is the case when $\jump_j$ is exactly the midpoint of $\tilde{\jumpext}_i$ and $\tilde{\jumpext}_{i + 1}$ for some $i \in [\sext-1]$, we make the convention $i_j \coloneqq i$.}
	\begin{equation}
	i_j \in \argmin_{i \in [\sext]} \ \abs{\jump_j - \tilde{\jumpext}_i}.
	\end{equation}
	In other words, $\tilde{\jumpext}_{i_j}$ is the element in $\SuppEq$ that is the closest to the jump discontinuity $\jump_j$.
	Since $\SuppEq$ defines an equidistant grid on the interval $\intvop{0}{n}$, it is not hard to see that 
	\begin{equation}\label{eq:proofs:mw:step4:supp_equid_control}
	\abs{\jump_j - \tilde{\jumpext}_{i_j}} \leq \frac{1}{2} \cdot \frac{n}{\sext+1}, \quad j = 1, \dots, \s,
	\end{equation}
	where we have also used the $\SC$-separation of $\grtr$ to conclude that $\jump_j $ cannot be too close to the boundary points $\tilde{\jumpext}_0 \coloneqq 0$ and $\tilde{\jumpext}_{\sext + 1} \coloneqq n$:
	\begin{align}
	\abs{\jump_j - 0} &\geq \SC \cdot \frac{n}{\s+1} \stackrel{\eqref{eq:proofs:mw:step4:sext_control}}{\geq} \frac{n}{\sext+1} > \frac{1}{2} \cdot \frac{n}{\sext+1}, \\*
	\abs{\jump_j - n} &\geq \SC \cdot \frac{n}{\s+1} \stackrel{\eqref{eq:proofs:mw:step4:sext_control}}{\geq} \frac{n}{\sext+1} > \frac{1}{2} \cdot \frac{n}{\sext+1}.
	\end{align}
\item \label{eq:proof:step4c}
	Introducing the index set $\mathcal{I} \coloneqq \{i_1, \dots, i_s\} \subset [\sext]$, we observe that $\cardinality{\mathcal{I}} = \s$, so that $j \mapsto i_j$ is a bijection between $[\s]$ and $\mathcal{I}$.
	In other words, none of the elements in $\SuppEq$ gets selected for more than one element of $\ssupp$; indeed, this claim is again a consequence of the $\SC$-separation of $\grtr$: if ${i_j} = {i_{j'}}$ for some $j, j' \in [\s]$, we would have
	\begin{equation}
	\abs{\jump_j - \jump_{j'}} \leq \abs{\jump_j - \tilde{\jumpext}_{i_j}} + \abs{\tilde{\jumpext}_{i_{j'}} - \jump_{j'}} \stackrel{\eqref{eq:proofs:mw:step4:supp_equid_control}}{\leq} \frac{n}{\sext+1} \stackrel{\eqref{eq:proofs:mw:step4:sext_control}}{\leq} \SC \cdot \frac{n}{\s+1}.
	\end{equation}
	The above construction of the auxiliary points $\SuppEq$ implies that equality can only hold if $j = j'$, and in any other case, we would have a contradiction to \eqref{eq:results:msc}.
	
	Consequently, the following extended gradient support set is well-defined and satisfies $\cardinality{\Supp} = \sext$:
	\begin{equation}
	\Supp = \{\jumpext_1, \dots, \jumpext_{\sext}\} \coloneqq \ssupp \disjunion \{ \round{\tilde\jumpext_i} \suchthat i \in \setcompl{\mathcal{I}} = [\sext]\setminus \mathcal{I} \} \subset [N].
	\end{equation}
	Note that $\round{\tilde\jumpext_i} \not\in \ssupp$ for $i \in \setcompl{\mathcal{I}}$ is due to \eqref{eq:proofs:mw:step4:supp_equid_control} and $n / (\sext+1) \geq \SC n / (4\s) > 1$, where the latter is a consequence of \eqref{eq:proofs:mw:step4:sext_control_upper} and the assumption $\SC \geq 8\s / n$ of Theorem~\ref{thm:results:mwbound}.
	In particular, we have that
	\begin{equation}
	\abs{\jumpext_i - \tilde\jumpext_i} \leq \frac{1}{2} \cdot \frac{n}{\sext+1}, \quad i = 1, \dots, \sext,
	\end{equation}
	where the case $i \in \mathcal{I}$ is due to \eqref{eq:proofs:mw:step4:supp_equid_control} and the case $i \in \setcompl{\mathcal{I}}$ is due to rounding.
	Moreover, recall that the elements of $\Supp$ are ordered, i.e., $\jumpext_1 < \dots < \jumpext_{\sext}$.
\item \label{eq:proof:step4d}
	We now aim to control the pairwise distance between the elements of $\Supp$. Let $i, j \in [\sext]$.
	If $i \neq j$, then we have the following upper bound:
	\begin{align}
	\abs{\jumpext_i - \jumpext_j} &\leq \abs{\jumpext_i - \tilde\jumpext_i} + \abs{\tilde\jumpext_i - \tilde\jumpext_j} + \abs{\tilde\jumpext_j - \jumpext_j} \\*
	&\leq \frac{1}{2} \cdot \frac{n}{\sext+1} + \abs{i-j} \cdot \frac{n}{\sext+1} + \frac{1}{2} \cdot \frac{n}{\sext+1} \\
	&= (\abs{i-j} + 1) \cdot \frac{n}{\sext+1} \\
	&\leq 2\, \abs{i-j} \cdot \frac{n}{\sext+1}.
	\end{align}
	For $\abs{i - j} \geq 2$, we have the following lower bound:
	\begin{align}
	\abs{\jumpext_i - \jumpext_j} &\geq  \abs{\tilde\jumpext_i - \tilde\jumpext_j} - \abs{\jumpext_i - \tilde\jumpext_i} - \abs{\tilde\jumpext_j - \jumpext_j} \\*
	&\geq \abs{i-j} \cdot \frac{n}{\sext+1} - \frac{1}{2} \cdot \frac{n}{\sext+1} - \frac{1}{2} \cdot \frac{n}{\sext+1} \\
	&= (\abs{i-j} - 1) \cdot \frac{n}{\sext+1} \\
	&\geq \frac{1}{2} \, \abs{i-j} \cdot \frac{n}{\sext+1}.
	\end{align}
	For $\abs{i - j} = 1$, we make a case distinction to obtain a similar lower bound: if $\jumpext_i, \jumpext_j \in \ssupp$, then the $\SC$-separation of $\grtr$ implies that
	\begin{equation}
	\abs{\jumpext_i - \jumpext_j} \geq \SC \cdot \frac{n}{\s+1} \stackrel{\eqref{eq:proofs:mw:step4:sext_control}}{\geq} \frac{n}{\sext+1} > \frac{1}{2} \, \abs{i-j} \cdot \frac{n}{\sext+1}.
	\end{equation}
	If $\jumpext_i \in \ssupp$ and $\jumpext_j = \round{\tilde\jumpext_{j}}$ for some $j \in \setcompl{\mathcal{I}}$, then 
	\begin{align}
	\abs{\jumpext_i - \jumpext_j} &\geq \abs{\jumpext_i - \tilde\jumpext_{j}} - \abs{\tilde\jumpext_{j} - \round{\tilde\jumpext_{j}}} \\*
	&\geq \frac{1}{2}  \cdot \frac{n}{\sext+1} - \frac{1}{2} \\*
	&\geq \frac{1}{4} \, \abs{i-j}  \cdot \frac{n}{\sext+1},
	\end{align}
	where the last line is due to $n / (\sext+1) \geq \SC n / (4\s) \geq 2$. Finally, if $\jumpext_i = \round{\tilde\jumpext_{i}}$ and $\jumpext_j = \round{\tilde\jumpext_{j}}$ for some $i, j \in \setcompl{\mathcal{I}}$, we have that
	\begin{align}
	\abs{\jumpext_i - \jumpext_j} &\geq \abs{\tilde\jumpext_{i} - \tilde\jumpext_{j}} - \abs{\tilde\jumpext_{i} - \round{\tilde\jumpext_{i}}} - \abs{\tilde\jumpext_{j} - \round{\tilde\jumpext_{j}}} \\*
	&\geq \frac{n}{\sext+1} - 1 \\*
	&\geq \frac{1}{2} \, \abs{i-j}  \cdot \frac{n}{\sext+1},
	\end{align}
	where the last line is again due to $n / (\sext+1) \geq 2$ and $\abs{i - j} = 1$.
	
	Finally, we define the auxiliary boundary points $\jumpext_0 \coloneqq 0$ and $\jumpext_{\sext+1} \coloneqq n$, and it is not hard to see that the above estimates remain valid if $i, j \in \{0,1, \dots, \sext+1\}$. Hence, in total, we obtain an approximate isometry between the sets $\{0,1, \dots, \sext+1\}$ and $\{\jumpext_0,\jumpext_1, \dots, \jumpext_{\sext+1}\}$:
	\begin{equation}\label{eq:proofs:mw:step4:supp_control}
	\frac{1}{4} \, \abs{i-j}  \cdot \frac{n}{\sext+1} \leq \abs{\jumpext_i - \jumpext_j} \leq 2 \, \abs{i-j}  \cdot \frac{n}{\sext+1}, \quad i,j = 0,1, \dots, \sext+1.
	\end{equation}
	Note that this relationship can be seen as a relaxation of an equidistant gradient support.
\item \label{eq:proof:step4e}
	Let us recall the tree-related notation from Iteration~\ref{def:proofs:mw:signaltree} and Definition~\ref{def:proofs:mw:signaltree:Q} as well as the basic properties from Proposition~\ref{prop:proofs:mw:tree_properties} and Proposition~\ref{prop:proofs:mw:tree_properties:Q}. Since $\cardinality{\Supp} = \sext = 2^{L_0} - 1$, the resulting binary tree $\Lambda^{L_0} = p^\mo (\Supp) = p^\mo(\{\jumpext_1, \dots, \jumpext_{\sext}\})$ is perfect and its depth satisfies
	\begin{equation}
	L_0 = \log_2(\sext + 1) = \ceil[\big]{\log_2(\tfrac{\s + 1}{\SC})}.
	\end{equation}
	Now, we aim to control the parameters $\d{i}{\ell}$ defined in \eqref{eq:proofs:mw:d-parameters} at the beginning of Step~\hyperref[eq:proof:step2b]{2(b)}. For this purpose, let $(\ell,i) \in \Lambda^{L_0}$ and let $r,\olarr{r},\orarr{r} \in \{0, 1, \dots, \sext+1\}$ be such that $\jumpext_r = \p{i}{\ell} \in \Supp$, $\jumpext_{\olarr{r}} = \pl{i}{\ell} \in \Supp \union \{\jumpext_{0}\}$, and $\jumpext_{\orarr{r}} = \pr{i}{\ell} \in \Supp \union \{\jumpext_{\sext+1}\}$.
	Then, we observe that
	\begin{align}
	\nl{i}{\ell} &= \cardinality{\Ql{i}{\ell}} = \abs{\pl{i}{\ell} - \p{i}{\ell}} = \abs{\jumpext_{\olarr{r}} - \jumpext_r}, \\*
	\nr{i}{\ell} &= \cardinality{\Qr{i}{\ell}} = \abs{\p{i}{\ell} - \pr{i}{\ell}} = \abs{\jumpext_r - \jumpext_{\orarr{r}}}.
	\end{align}
	These expressions can be controlled by the approximate isometry relation \eqref{eq:proofs:mw:step4:supp_control} as follows:
	\begin{align}
		\frac{1}{4} \, \abs{r - \olarr{r}}  \cdot \frac{n}{\sext+1} \leq \nl{i}{\ell} \leq 2 \, \abs{r - \olarr{r}}  \cdot \frac{n}{\sext+1}, \\
		\frac{1}{4} \, \abs{r - \orarr{r}}  \cdot \frac{n}{\sext+1} \leq \nr{i}{\ell} \leq 2 \, \abs{r - \orarr{r}}  \cdot \frac{n}{\sext+1}.
	\end{align}
	Since $\Lambda^{L_0}$ is perfect and $\jumpext_0 < \jumpext_1 < \dots < \jumpext_{\sext+1}$, it is not hard to see (yet crucial) that
	\begin{equation}\label{eq:proofs:mw:step4:index_control}
		\abs{r - \olarr{r}} = \abs{r - \orarr{r}} = 2^{L_0 - \ell},
	\end{equation}
	and therefore
	\begin{equation}
	\abs{r - \olarr{r}}  \cdot \frac{n}{\sext+1} = \abs{r - \orarr{r}}  \cdot \frac{n}{\sext+1} =  2^{L_0 - \ell} \cdot \frac{n}{2^{L_0}} = \frac{n}{2^\ell}.
	\end{equation}
	In total, we obtain
	\begin{equation}\label{eq:proofs:mw:step4:n_control}
	\frac{1}{4}  \cdot \frac{n}{2^\ell} \leq \nl{i}{\ell}, \nr{i}{\ell} \leq  2 \cdot \frac{n}{2^\ell},
	\end{equation}
	which particularly implies that
	\begin{equation}
	\d{i}{\ell} = \sqrt{\frac{\nl{i}{\ell} + \nr{i}{\ell}}{\nl{i}{\ell} \cdot \nr{i}{\ell}}} = \sqrt{\frac{1}{\nl{i}{\ell}} + \frac{1}{\nr{i}{\ell}}} \leq \sqrt{8} \cdot \sqrt{\frac{2^\ell}{n}}.
	\end{equation}
	This allows us to bound the balancing parameter $\b{\max}{\Lambda^{L_0}}$ as follows:
	\begin{equation}\label{eq:proofs:mw:step4:L0}
	\b{\max}{\Lambda^{L_0}} = \max_{(\ell,i) \in \Lambda^{L_0}} \b{i}{\ell} = \max_{(\ell,i) \in \Lambda^{L_0}} \frac{\d{i}{\ell}}{c_\ell} \leq \sqrt{8} \cdot \sqrt{\frac{2^\ell}{n}} \cdot \sqrt{\frac{n}{2^{\ell+1}}} = 2.
	\end{equation}
\item \label{eq:proof:step4f}
	For $1 \leq \ell \leq L$, we set 
	\begin{equation}
	\n{\min}{\ell} \coloneqq \min_{\substack{i\in[2^{\ell-1}], \\ \SuppS{i}{\ell} \neq \emptyset}} \big( \min \{ \nl{i}{\ell}, \nr{i}{\ell} \} \big) \quad \text{and} \quad \n{\max}{\ell} \coloneqq \max_{\substack{i\in[2^{\ell-1}], \\ \SuppS{i}{\ell} \neq \emptyset}} \big( \max \{ \nl{i}{\ell}, \nr{i}{\ell} \} \big).
	\end{equation}
	According to \eqref{eq:proofs:mw:step4:n_control}, we know that
	\begin{equation}\label{eq:proofs:mw:step4:n_minmax_control_L0}
	\frac{1}{4}  \cdot \frac{n}{2^{L_0}} \leq \n{\min}{{L_0}} \leq \n{\max}{{L_0}} \leq 2 \cdot \frac{n}{2^{L_0}}.
	\end{equation}
	Now let $L_0 < \ell \leq L$. The bound \eqref{eq:proofs:mw:tree_properties:Q_S:card_iter} from Proposition~\ref{prop:proofs:mw:tree_properties:Q}\ref{prop:proofs:mw:tree_properties:Q_S} implies that
	\begin{equation}
	\frac{\n{\min}{{\ell-1}} - 1}{2} \leq \n{\min}{{\ell}} \leq \n{\max}{{\ell}} \leq \frac{\n{\max}{{\ell-1}} + 1}{2}.
	\end{equation}
	A straightforward induction argument then yields
	\begin{equation}
	\frac{\n{\min}{{L_0}} - 2^{\ell - L_0} + 1}{2^{\ell - L_0}} \leq \n{\min}{{\ell}} \leq \n{\max}{{\ell}} \leq \frac{\n{\max}{{L_0}} + 2^{\ell - L_0} - 1}{2^{\ell - L_0}},
	\end{equation}
	and together with \eqref{eq:proofs:mw:step4:n_minmax_control_L0}, we have that
	\begin{equation}\label{eq:proofs:mw:step4:n_minmax_control}
	\frac{1}{4}  \cdot \frac{n}{2^{\ell}} - 1 \leq \n{\min}{{\ell}} \leq \n{\max}{{\ell}} \leq 2  \cdot \frac{n}{2^{\ell}} + 1.
	\end{equation}
	Since $\n{\min}{{\ell}} \geq 1$, it follows that
	\begin{equation}
	\frac{1}{4}  \cdot \frac{n}{2^{\ell}} \leq \n{\min}{{\ell}} + 1 \leq 2 \n{\min}{{\ell}},
	\end{equation}
	and therefore, for every $(\ell,i) \in \Lambda$ with $\ell > L_0$, we have that
	\begin{equation}
	\b{i}{\ell} = \frac{\d{i}{\ell}}{c_\ell} = \sqrt{\frac{n}{2^{\ell+1}}} \cdot \sqrt{\frac{1}{\nl{i}{\ell}} + \frac{1}{\nr{i}{\ell}}} \leq \sqrt{\frac{n}{2^{\ell+1}}} \cdot \sqrt{\frac{2}{\n{\min}{\ell}}} \leq \sqrt{\frac{n}{2^{\ell+1}}} \cdot \sqrt{\frac{16 \cdot 2^\ell}{n} } \leq \sqrt{8}.
	\end{equation}
	To show a lower bound for $\b{i}{\ell}$, let us first assume that $\n{\max}{{\ell}} \geq 2$. From \eqref{eq:proofs:mw:step4:n_minmax_control}, it follows that
	\begin{equation}
	\frac{\n{\max}{{\ell}}}{2} \leq \n{\max}{{\ell}} - 1 \leq 2  \cdot \frac{n}{2^{\ell}},
	\end{equation}
	so that
	\begin{equation}
	\b{i}{\ell} = \sqrt{\frac{n}{2^{\ell+1}}} \cdot \sqrt{\frac{1}{\nl{i}{\ell}} + \frac{1}{\nr{i}{\ell}}} \geq \sqrt{\frac{n}{2^{\ell+1}}} \cdot \sqrt{\frac{2}{\n{\max}{\ell}}} \geq \sqrt{\frac{n}{2^{\ell+1}}} \cdot \sqrt{\frac{2^\ell}{2n} } \geq \frac{1}{2}.
	\end{equation}
	Finally, we assume that $\n{\min}{{\ell}} = \n{\max}{{\ell}} = 1$. The construction of Iteration~\ref{def:proofs:mw:signaltree} implies that this can be only the case if $\ell = L = \depth(\Lambda)$.
	Since $\cardinality{\SuppS{i}{L_0+1}} \leq \n{\max}{{L_0}} \leq 2 n / 2^{L_0}$ by Proposition~\ref{prop:proofs:mw:tree_properties:Q}\ref{prop:proofs:mw:tree_properties:Q_S} and \eqref{eq:proofs:mw:step4:n_minmax_control_L0}, one can argue analogously to Proposition~\ref{prop:proofs:mw:tree_properties}\ref{prop:proofs:mw:tree_properties:depth} that the following bound for the depth of~$\Lambda$ holds true:
	\begin{align}
	L \leq L_0 + \ceil{\log_2(2 \cdot \tfrac{n}{2^{L_0}})} \leq \log_2(2^{L_0}) + \log_2(4 \cdot \tfrac{n}{2^{L_0}}) = \log_2(4 n).
	\end{align}
	Consequently, we have that $n  / 2^L \geq {1} / {4}$, which implies for all $i\in[2^{L-1}]$ that
	\begin{equation}
	\b{i}{L} = \sqrt{\frac{n}{2^{L+1}}} \cdot \sqrt{\frac{1}{\nl{i}{L}} + \frac{1}{\nr{i}{L}}} = \sqrt{\frac{n}{2^{L+1}}} \cdot \sqrt{2} \geq \frac{1}{2}.
	\end{equation}
	
	In total, we obtain the following bounds for $\b{\min}{\Lambda^\circ}$ and $\b{\max}{\Lambda^\natural}$:
	\begin{equation}\label{eq:proofs:mw:step4:natural_circ}
	\b{\max}{\Lambda^\natural} \leq \sqrt{8} \quad \text{and} \quad \b{\min}{\Lambda^\circ} \geq \frac{1}{2}.
	\end{equation}
\end{prooflist}

\subsubsection{Step 5: Proof of Theorem \ref{thm:results:mwbound}}
\label{eq:proof:step5}

In this final step, we simply combine the outcomes of Step~\hyperref[eq:proof:step3a]{3} and Step~\hyperref[eq:proof:step4]{4} to verify the bound \eqref{eq:results:mwbound} in Theorem~\ref{thm:results:mwbound}. Indeed, our bound for the conic mean width in \eqref{eq:proofs:mw:mw_bound_general} and our bounds for the balancing parameters in \eqref{eq:proofs:mw:step4:L0} and \eqref{eq:proofs:mw:step4:natural_circ} yield the claim:
\begin{align}
\effdim[\conic]{\descset{\lnorm{\TV(\cdot)}[1], \grtr}} &\leq 1 + 3\sext + \bigg[6 + 8\underbrace{\bigg(\frac{\b{\max}{\Lambda^{L_0}}}{\gamma \b{\min}{\Lambda^\circ}}\bigg)^2}_{\lesssim 1}\bigg] \sext \log(n) + 12 \underbrace{\bigg(\frac{\b{\max}{\Lambda^\natural}}{\gamma \b{\min}{\Lambda^\circ}}\bigg)^2}_{\lesssim 1} \sext \log^2(n) \\*
&\lesssim \sext\log^2(n) \stackrel{\eqref{eq:proofs:mw:step4:sext_control_upper}}{\leq} \frac{4}{\SC} \cdot \s\log^2(n), \label{eq:proofs:mw:step5:final}
\end{align}
where we have also used that $\sext \geq 1$, which is due to \eqref{eq:proofs:mw:step4:sext_control} and $\SC \leq 1$.
\qed

\begin{remark}
\begin{rmklist}
\item\label{rmk:proofs:mw:nonasymp}
	\textbf{Non-asymptotic bounds.} It is worth pointing out that the general estimate in \eqref{eq:proofs:mw:mw_bound_general_nonasymp} can be seen as a \emph{non-asymptotic} bound on the conic mean width, since it does not involve any hidden universal constants. However, this expression depends on various model parameters, such as the tree set cardinalities and the balancing parameters. Thus, although more accurate, \eqref{eq:proofs:mw:mw_bound_general_nonasymp} provides a much less informative sampling-rate bound than our final outcome in~\eqref{eq:proofs:mw:step5:final}.
\item\label{rmk:proofs:mw:signs}
	\textbf{The role of sign patterns.} While the sign pattern of $\TV\grtr$ does not affect the asymptotic-order bound \eqref{eq:proofs:mw:step5:final} in Step~\hyperref[eq:proof:step5]{5}, it allows for refinements at least in some special cases. For example, if $\sign(\TV\grtr)$ would be positive on $\supp(\TV\grtr)$, then one would simply have $\e{i}{\ell} = 1$ for every $(\ell,i) \in \Lambda^{L_0}$; cf.~Case~1 in Step~\hyperref[eq:proof:step3b]{3(b)}. \qedhere
\end{rmklist}
\end{remark}\label{rmk:proofs:mw}

\subsection{Proof of Theorem~\ref{thm:results:stable}}
\label{subsec:proofs:stable}

Let us consider the following signal vector
\begin{equation}
\tilde{\x} \coloneqq \psinv{\TV} \proj{\ssupp} \TV \grtr + \lambda \1_n \in \R^n,
\end{equation}
where the coefficient $\lambda \in \R$ represents the ``constant part'' of $\tilde{\x}$, which is left unspecified for the moment.
Due to
\begin{equation}
\TV \tilde{\x} = \underbrace{\TV \psinv{\TV}}_{= \I{N}} \proj{\ssupp} \TV \grtr + \lambda \underbrace{\TV\1_n}_{= \vnull} = \proj{\ssupp} \TV \grtr,
\end{equation}
we observe that $\tilde{\x}$ is $\s$-gradient-sparse and $\TV \tilde{\x}$ is a best $\s$-term approximation of $\TV \grtr$ (with respect to the $\l{1}$-norm).
Now, we apply Proposition~\ref{prop:results:stable:general} with $R = 1$ and a rescaled version of $\tilde{\x}$:
\begin{equation}
\grtrsparse \coloneqq \frac{\lnorm{\TV\grtr}[1]}{\lnorm{\TV\tilde{\x}}[1]} \cdot \tilde{\x} = \frac{\lnorm{\TV\grtr}[1]}{\lnorm{\proj{\ssupp} \TV \grtr}[1]} \cdot \tilde{\x} = \tau(\grtr) \cdot \big( \psinv{\TV} \proj{\ssupp} \TV \grtr + \lambda \1_n \big).
\end{equation}
Then we have $\lnorm{\TV\grtr}[1] = \lnorm{\TV\grtrsparse}[1]$ and the assumption~\eqref{eq:results:stable:general:meas} follows from Theorem~\ref{thm:results:mwbound} and \eqref{eq:results:stable:meas}. Hence, with probability at least $1 - e^{-\probsuccess^2/2}$, the following error bound holds true:
\begin{equation}
\lnorm{\solu - \grtr} \leq \lnorm{\grtr - \grtrsparse} + \frac{2\noiseparam}{\pospart{\sqrt{m-1} - \sqrt{m_0-1}}} \lesssim \lnorm{\grtr - \grtrsparse} + \frac{\noiseparam}{\sqrt{m}},
\end{equation}
where we have used that $m \geq C \cdot m_0$ for a sufficiently large universal constant $C > 1$ (this can be easily achieved by slightly enlarging the hidden constant in \eqref{eq:results:stable:meas}).
Finally, using the definition of $\grtrsparse$, we obtain
\begin{align}
\lnorm{\grtr - \grtrsparse} &= \lnorm[\big]{\grtr - \tau(\grtr) \cdot \big( \psinv{\TV} \proj{\ssupp} \TV \grtr + \lambda \1_n \big)} \\*
&= \lnorm[\big]{\grtr - \tau(\grtr) \cdot \psinv{\TV} \TV \grtr - \tau(\grtr) \lambda \1_n + \tau(\grtr) \cdot  \psinv{\TV} \proj{\setcompl\ssupp} \TV \grtr} \\*
&\leq \lnorm[\big]{\grtr - \tau(\grtr) \cdot \grtr + \tau(\grtr) \sp{\tfrac{1}{n}\1_n}{\grtr} \1_n - \tau(\grtr) \lambda \1_n } + \lnorm[\big]{ \tau(\grtr) \cdot  \psinv{\TV} \proj{\setcompl\ssupp} \TV \grtr},
\end{align}
and by choosing $\lambda = \sp{\tfrac{1}{n}\1_n}{\grtr} / \tau(\grtr)$, it follows that
\begin{align}
\lnorm{\grtr - \grtrsparse} &\leq \big(\tau(\grtr) - 1 \big) \lnorm[\big]{\grtr - \sp{\tfrac{1}{n}\1_n}{\grtr} \1_n} + \tau(\grtr) \lnorm[\big]{\psinv{\TV} \proj{\setcompl\ssupp}\TV\grtr} .
\end{align} \qed


\section*{Acknowledgments}
{\smaller
M.G.\ is supported by the Bundesministerium für Bildung und Forschung (BMBF) through the Berliner Zentrum for Machine Learning (BZML), Project AP4.
M.M.\ is supported by the DFG Priority Programme DFG-SPP 1798 Grants KU 1446/21 and KU 1446/23.
\revision{Moreover, the authors would like to thank the anonymous referees for their useful comments and suggestions which have helped to improve the original manuscript.}

}

\renewcommand*{\bibfont}{\smaller}
\begin{refcontext}[sorting=nyt]
	\printbibliography
\end{refcontext}

\newpage
\listoftodos

\end{document}